\documentclass[10pt,final,journal,twocolumn,singlespaced]{IEEEtran}
\usepackage{algorithm}
\usepackage{algpseudocode}
\usepackage{amsmath}
\usepackage{amssymb}
\usepackage{amsthm}
\usepackage{bm}
\usepackage{booktabs}
\usepackage{color}
\usepackage{cite}
\usepackage{epsfig}
\usepackage{epstopdf}
\usepackage{flafter}
\usepackage{float}
\usepackage{ifthen}
\usepackage{multirow}
\usepackage{amsfonts}
\usepackage{makecell} 
\usepackage{pifont}
\usepackage{subfigure}
\usepackage{times}
\usepackage{threeparttable}
\usepackage{url} 
\usepackage{subfigure,graphicx}
\usepackage{color,multirow}
\usepackage{makecell}
\usepackage{setspace}
\usepackage{cite} 
\usepackage{url} 
\usepackage{ifthen} 
\usepackage[T1]{fontenc}
\usepackage{graphicx}
\usepackage{changepage}
\usepackage[colorlinks,linkcolor=red,anchorcolor=gray,citecolor=green,urlcolor=black]{hyperref}
\usepackage{graphics}
\usepackage{booktabs}

\newtheorem{lemma}{\textbf{Lemma}}

\newcommand{\tabincell}[2]{\begin{tabular}{@{}#1@{}}#2\end{tabular}}
\usepackage{color}
\definecolor{myGreen}{rgb}{0.945,0.972,0.980}

\begin{document}
\title{S2S-WTV: Seismic Data Noise Attenuation Using Weighted Total Variation Regularized Self-Supervised Learning
\thanks{(Zitai Xu and Yisi Luo contribute equally to this work.)}
\thanks{The authors are with the School of Mathematics and Statistics, Xi'an Jiaotong University, Xi'an, P.R.China.}}
\author{Zitai Xu,
Yisi Luo,
Bangyu Wu, \IEEEmembership{Member, IEEE},
Deyu Meng, \IEEEmembership{Member, IEEE}}
\maketitle
\begin{abstract}
Seismic data often undergoes severe noise due to environmental factors, which seriously affects subsequent applications. Traditional hand-crafted denoisers such as filters and regularizations utilize interpretable domain knowledge to design generalizable denoising techniques, while their representation capacities may be inferior to deep learning denoisers, which can learn complex and representative denoising mappings from abundant training pairs. However, due to the scarcity of high-quality training pairs, deep learning denoisers may sustain some generalization issues over various scenarios. In this work, we propose a self-supervised method that combines the capacities of deep denoiser and the generalization abilities of hand-crafted regularization for seismic data random noise attenuation. Specifically, we leverage the Self2Self (S2S) learning framework with a trace-wise masking strategy for seismic data denoising by solely using the observed noisy data. Parallelly, we suggest the weighted total variation (WTV) to further capture the horizontal local smooth structure of seismic data. Our method, dubbed as S2S-WTV, enjoys both high representation abilities brought from the self-supervised deep network and good generalization abilities of the hand-crafted WTV regularizer and the self-supervised nature. Therefore, our method can more effectively and stably remove the random noise and preserve the details and edges of the clean signal. To tackle the S2S-WTV optimization model, we introduce an alternating direction multiplier method (ADMM)-based algorithm. Extensive experiments on synthetic and field noisy seismic data demonstrate the effectiveness of our method as compared with state-of-the-art traditional and deep learning-based seismic data denoising methods. 
\end{abstract}
\begin{IEEEkeywords}
Seismic data,
noise attenuation,
self-supervised,
deep convolutional network,
ADMM.
\end{IEEEkeywords}
\IEEEpeerreviewmaketitle
\section{Introduction} \label{sec:Int}
\IEEEPARstart{S}{eismic} data has broad application prospects in many geophysical applications such as event detection\cite{TGRS_19_detection}, inversion\cite{TGRS_20_inversion}, seismic diffraction separation and imaging\cite{TGRS_22_diff}, etc. However, seismic data often suffer from noise degradation due to environmental factors like ocean waves and wind, or mechanical failures in receiver arrays, which seriously affect subsequent applications. Therefore, seismic data noise attenuation is an important pre-processing step that benefits many geophysical applications.\par 
In the era before deep learning, seismic data denoising methods were mainly based on hand-crafted denoisers such as transforms, filters, and regularizations. The transform-based methods utilized different transforms such as wavelet transform\cite{wavelet_1,wavelet_2}, Fourier transform\cite{FT}, and shearlet transform\cite{shearlet} to obtain sparse features and conduct noise attenuation on the transformed coefficients. The filter-based methods, such as grey filter\cite{filter_IGARSS}, time-frequency peak filter\cite{GRSL_TFPF_1,TGRS_TFPF_2}, and fission particle filter\cite{GRSL_15_filter} were also powerful tools for seismic data denoising. Recently, there have emerged many regularization-based methods for seismic noise attenuation. The most widely considered one is the low-rank regularization\cite{Geoph_NNM,TGRS_17_lowrank,GRSL_20_lowrank,TGRS_20_lowrank,TGRS_21_tensor}. Along this line, many techniques such as convex or nonconvex approximations of rank\cite{TGRS_NNM,JSTARS_19_lowrank,GRSL_WNNM}, Hankel low-rank approximation\cite{TGRS_Hankel_1,TGRS_Hankel_2}, low-rank factorization\cite{TGRS_17_fac}, and tensor singular value decomposition\cite{JSTARS_tSVD} were proposed to explore the hidden low-rank structures of seismic data for denoising. Besides low-rankness, other hand-crafted regularizers and techniques such as sparse coding\cite{Geo_SC,Geo_SC_2}, total variation\cite{GRSL_TV,TGRS_TV,JAG_TV,GRSL_TV_2,TGRS_DTGV}, non-local similarity\cite{JAG_nonlocal_0,JAG_nonlocal,TGRS_21_PCA}, and dictionary learning\cite{GP_DL,GRSL_DL,GJI_DL,JAG_DL,TGRS_19_DL} were also extensively studied for seismic data noise attenuation. These hand-crafted denoisers, which were based on interpretable domain knowledge, enjoy good generalization abilities for different datasets and noises.\par 
In the past few years, deep learning has emerged as a popular tool for seismic data noise attenuation. The pioneer works, e.g., \cite{TGRS_19_DNN,GRSL_19_ResNet,TGRS_20_3D}, mainly utilized deep neural networks (DNNs) with pairs of clean and noisy seismic data to supervisedly train the denoising network, which could learn a good deep denoising prior from the big data. Later works proposed to learn more realistic and representative denoising mappings via improved network structures, such as residual network\cite{TGRS_21_ResNet,AG_Res}, multiscale network\cite{TGRS_22_multi,TGRS_20_multi,TGRS_22_multi_2}, feature fusion network\cite{TGRS_fusion}, and generative adversarial network\cite{TGRS_21_GAN,TGRS_GAN_21}. Meanwhile, many modern learning strategies were exploited to enhance the noise attenuation ability of deep networks, such as loss balance\cite{TGRS_21_LB}, pre-trained model adaptation\cite{RS_Pretrained}, and diffusion model\cite{TGRS_DM}. Due to the powerful representation abilities of DNNs, these deep denoisers can learn complex and effective denoising mappings and thus obtain impressive results in the training domain. However, these methods rely on collecting a large number of training data to surpass hand-crafted denoisers, in which, however, high-quality seismic noisy-clean data pairs are always hard to collect due to the lack of data sources and complex field noisy scenarios, which inevitably limits the applicabilities of deep learning seismic denoising methods in out-of-distribution field datasets.\par
In general, DNNs have sufficient representation abilities, but may lack generalization guarantee, while hand-crafted denoisers have good generalization abilities, but may lack strong representation capacities as compared to DNNs. Thus, it is interesting and imperative to combine the expressiveness of DNNs and the generalization abilities of hand-crafted regularizations to more effectively attenuate random noise in seismic data. To meet this pressing challenge, we propose a new self-supervised seismic data noise attenuation method, which takes advantage of both the representation abilities of DNNs and the generalization abilities of hand-crafted regularizers. Specifically, we leverage the self-supervised dropout DNN (Self2Self, or simply S2S)\cite{S2S} for seismic data noise attenuation by solely using the noisy observation without other training data. Different from the classical S2S, we design a trace-wise Bernoulli masking strategy to more effectively remove field noise in seismic data. Meanwhile, we suggest the hand-crafted weighted total variation (WTV) regularization under the S2S framework to capture the local smooth structures of seismic data. Our method, dubbed as S2S-WTV, combines the self-supervised DNN and hand-crafted WTV, and thus enjoys both the expressiveness of DNNs and the generalizability and interpretability of domain knowledge. By virtue of such a combination, our method can well handle complex field noise and robustly preserve the fine details of the geological structure. Finally, we introduce an alternating direction multiplier method (ADMM)-based algorithm to address the resulting seismic data denoising model. In summary, this work has the following contributions:
\begin{itemize}
\item We propose S2S-WTV, a self-supervised deep learning method for seismic data noise attenuation by solely using the observed noisy data. Our method combines the powerful representation abilities of self-supervised DNN and generalization abilities of hand-crafted WTV regularization to faithfully remove irregular noise in seismic data and preserve the fine structures of the clean signal.
\item We elaborately design a trace-wise masking strategy to train the self-supervised DNN, which can better adapt to field noise in seismic data than the original element-wise masks. Meanwhile, we propose a fine-tuning procedure to efficiently handle high-dimensional seismic data. To minimize the self-supervised loss, we introduce an ADMM-based algorithm to optimize the DNN parameters.
\item Extensive experiments on synthetic and field noisy seismic data validate the effectiveness and superiority of our S2S-WTV over state-of-the-art hand-crafted and deep learning-based seismic data noise attenuation methods.  
\end{itemize}\par
The rest of this paper is organized as follows. Sec. \ref{Sec_rela} discusses some related work on self-supervised seismic data denoising. Sec. \ref{Sec_method} introduces the proposed S2S-WTV method and algorithm. Sec. \ref{Sec_exp} carries out extensive experiments to show the effectiveness of our method. Sec. \ref{Sec_dis} gives some discussions of our method. Sec. \ref{Sec_con} concludes this paper.
\section{Related Work}\label{Sec_rela}
Recently, there emerged quite a few unsupervised/self-supervised seismic noise attenuation methods. The pioneer work\cite{TGRS_sparse} used the unsupervised sparse penalty loss to learn sparse features for seismic denoising, which did not require data labels. Oliveira et al.\cite{TGRS_21_self} proposed a self-supervised method for seismic denoising by building noisy/clean pairs without supervised data. Li et al.\cite{GRSL_GAN} conducted unsupervised learning with unpaired data via cycle-generative adversarial networks for seismic denoising. The merits of these methods are that they do not require pairs of noisy/clean seismic training data, which significantly eases the burden of training data collection. However, these methods still need to collect a large number of unlabeled/unpaired seismic data to train the network, which is still an unwilling process.\par   
Another type of unsupervised methods were based on the deep image prior (DIP)\cite{DIP}, which only used a single noisy observation to train the DNN for noise attenuation\cite{BSNet,TGRS_DIP_Hu_2,GRSL_22_DIP_2,TGRS_3DDIP}. These methods were mainly based on the fact that an untrained convolutional DNN (mostly with a U-Net structure) can fit the signal part of the noisy observation before fitting the noisy part. Thus by early stopping one can achieve noise attenuation by using such an intrinsic prior of DNN in an unsupervised manner. Several techniques and enhancements were proposed based on DIP. For example, Saad et al.\cite{DIP_attention} introduced the attention module into the DIP network for seismic denoising. Liu et al.\cite{GRSL_22_DIP} employed adjacent traces of noisy seismic data as the inputs and labels, which could suppressed unpredictable random nosie with only the observed noisy data. Saad and Chen\cite{GP_PATCHUNET} incoperated patch division into the DIP network for more stable noise attenuation. Wang et al.\cite{TGRS_unfolding} utilized the deep unfolding technique of sparse coding model for unsupervised seismic denoising. Qian et al.\cite{TGRS_22_Qian} utilized the unsupervised Stein's unbiased risk estimate loss function in the transformed domain for seismic denoising. These methods all utilized a single noisy observation to train the DNN for seismic data denoising, which got rid of external training process.\par 
Our method also uses a single noisy observation to train the DNN in a self-supervised manner. However, to our best knowledge, the combination of self-supervised deep learning and hand-crafted regularizations has not been stuided in existing literatures on seismic data denoising. Our method simultaneously leverages the S2S learning strategy and the hand-crafted WTV regularizer, which brings the wisdom from both worlds to handle the challenging seismic data denoising problem. The proposed S2S-WTV combines the representation abilities of self-supervised DNN and the generalization abilities of hand-crafted regularization, which can more faithfully attenuate complex noise in seismic data and preserve the details of geological structure. Thus, our method is significantly different from the above DIP-based methods.\par 
\section{The Proposed Method}\label{Sec_method}
In this section, we introduce the proposed S2S-WTV for seismic noise attenuation. We first present the proposed trace-wise masked S2S learning paradigm. Then, we introduce the WTV regularization in the self-supervised learning model, followed by the ADMM algorithm and fine-tuning strategy to efficiently optimize the S2S-WTV model.
\subsection{S2S Learning with Trace-Wise Masking}
\subsubsection{Training Loss Design}
Suppose that we are given the noisy seismic data denoted by a matrix ${\bf Y}\in{\mathbb R}^{H\times W}$, where $H$ denotes the height (number of time samples) and $W$ denotes the width (number of traces). We follow the basic noisy model \cite{TGRS_19_DNN,TGRS_3DDIP}, which assumes that $\bf Y$ is an addition of the underlying clean seismic signal $\bf X$ and random noise $\bf N$:
\begin{equation}\label{noisy_model}
{\bf Y} = {\bf X} + {\bf N}.
\end{equation}
The seismic noise attenuation aims to estimate the underlying $\bf X$ from the observed noisy data $\bf Y$. Under the maximum a posterior (MAP) framework, the noise attenuation can be formulated as the following optimization model over $\bf X$:
\begin{equation}
\min_{\bf X}\left\lVert {\bf Y} - {\bf X}\right\rVert_F^2 + \phi({\bf X}),
\end{equation}
where $\left\lVert\cdot\right\rVert_{F}$ denotes the matrix Frobenius norm and $\phi(\cdot)$ is a prior term (regularization) that characterizes $\bf X$. In this work, we specify the prior term as a self-supervised ``deep prior'' conveyed by an untrained (randomly initialized) deep convolutional neural network (CNN), i.e., the self-supervised CNN generates the desired $\bf X$ and we only consider the fidelity term $\left\lVert {\bf Y} - {\bf X}\right\rVert_F^2$:
\begin{equation}\label{loss_1}
\min_{\theta}\left\lVert {\bf Y} - {\bf X}\right\rVert_F^2,\;{\rm where}\;{\bf X} = f_\theta({\bf Y}).
\end{equation}
 \begin{figure}[t]
	\scriptsize
	\setlength{\tabcolsep}{0.9pt}
	\begin{center}
		\begin{tabular}{cccc}
			\includegraphics[width=0.115\textwidth]{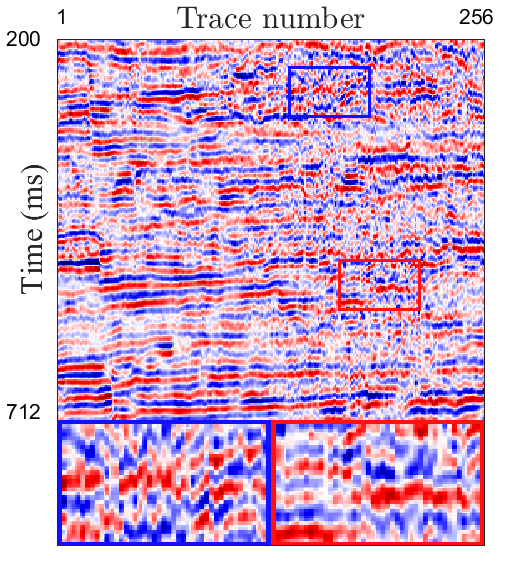}&
			\includegraphics[width=0.115\textwidth]{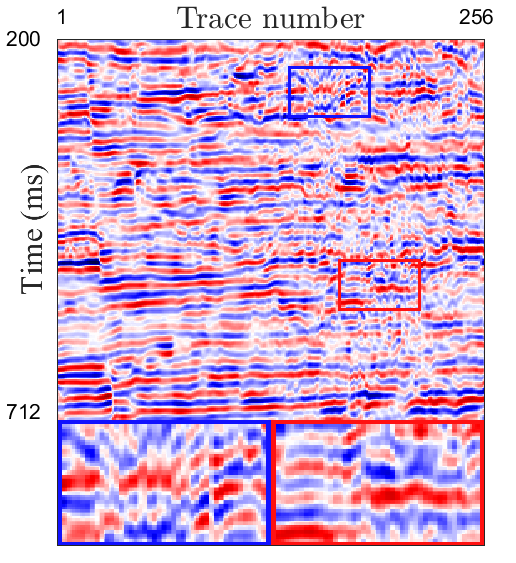}&
			\includegraphics[width=0.115\textwidth]{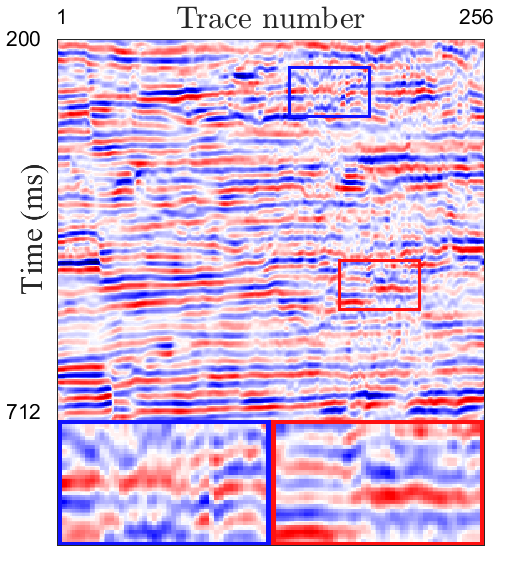}&
			\includegraphics[width=0.115\textwidth]{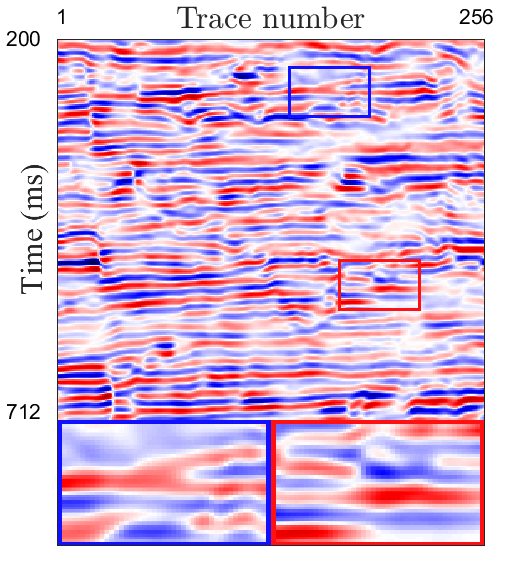}\\
			&
			\includegraphics[width=0.115\textwidth]{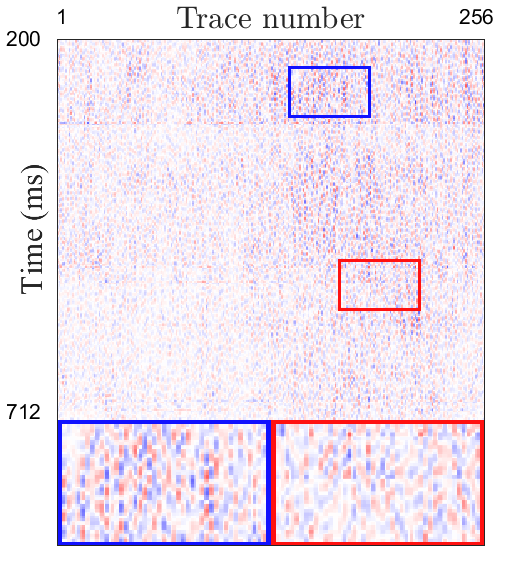}&
			\includegraphics[width=0.115\textwidth]{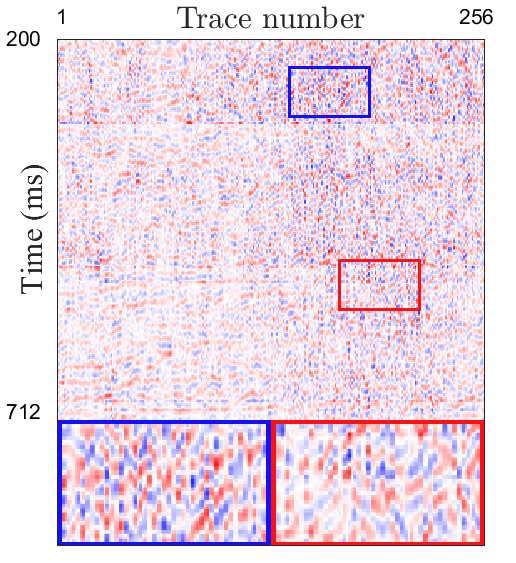}&
			\includegraphics[width=0.115\textwidth]{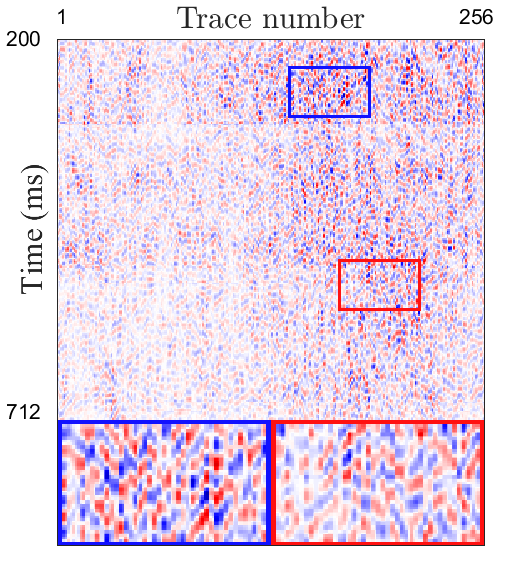}\\
			&LS 0.301&LS 0.547&LS 0.261\\
			Noisy& Element-wise mask&Row-wise mask&Trace-wise mask\\
			\vspace{-0.5cm}
		\end{tabular}
	\end{center}
	\caption{The noise attenuation results (the first row), the corresponding residual map between the noisy data and denoising results (the second row), and the LS (local similarity) values by using our S2S-WTV with different masking strategies on field noisy seismic data {\it X}. The trace-wise masking strategy is more effective since the correlations between adjacent traces of the seismic signal are stronger than those between adjacent rows. Thus, it is easier to use the unmasked traces to predict the masked clean traces.\label{fig_mask}}
	\vspace{-0.4cm}
\end{figure}Here, $f_\theta(\cdot)$ denotes the CNN with learnable parameters $\theta$. It was thoroughly demonstrated in literatures\cite{DIP,DIP_attention} that an untrained CNN with proper structures can itself reveal effective natural signal priors for noise attenuation. Thus the denoising task can be done by solely using the observed data $\bf Y$ to train the CNN via the self-supervised loss (\ref{loss_1}).\par
Although the self-supervised model is simple and concise, its performance is not stable since it is sensitive to the iteration number. The optimal iteration number is hard to determine as the CNN $f_\theta(\cdot)$ could eventually generate the noisy data and early stopping is needed to reconstruct a clean sginal. To cope with these challenges, we leverage the S2S\cite{S2S} learning paradigm for seismic data denoising, which foucses on the variance reduction to aviod overfitting to noisy data. The core concept of S2S is to reducing the variance (noise) of the output by masking some elements and predicting the other elements multiple times, and the average results could effectively reduce the variance of the predicted output.\par 
More specifically, we build some Bernoulli sampled instances
of the observed seismic data $\bf Y$ as training data:
\begin{equation}\label{eq_Y_n}
\widehat{\bf Y}_n = {\bf M}_n\odot{\bf Y},\;n=1,2,\cdots,N,
\end{equation}
where $\{\widehat{\bf Y}_n\}_{n=1}^{N}$ denote the generated samples, $\{{\bf M}_n\}_{n=1}^N$ (${\bf M}_n\in\{0,1\}^{H\times W}$ for all $n=1,2,\cdots,N$) denote the masks for generating these samples, and $\odot$ denotes the element-wise product. The masks $\{{\bf M}_n\}_{n=1}^N$ are generated by using a trace-wise Bernoulli sampling strategy; see details in Sec. \ref{trace_mask}.\par 
Using these sampled instances, we form the following self-supervised loss for an untrained deep CNN $f_\theta(\cdot)$:
\begin{equation}\label{loss_2}
\min_\theta \sum_{n=1}^N \left\lVert ({\bf Y}-{\widehat{\bf X}_n })\odot({\bf 1}-{\bf M}_n)\right\rVert_{F}^2,\;{\rm where}\;\widehat{\bf X}_n =f_\theta(\widehat{\bf Y}_n).
\end{equation}
The loss of each instance $\widehat{\bf Y}_n$ is computed only on elements masked by ${\bf M}_n$. Many such instances with different masks could ensure that all elements of $\bf Y$ are included during the training. Meanwhile, using the trace-wise Bernoulli sampled instances to train the CNN produces similar training losses to that of using pairs of trace-wise Bernoulli sampled instance $\{\widehat{\bf Y}_n\}_{n=1}^{N}$ and ground-truth data $\bf X$ in terms of expectation (see Lemma \ref{lemma}), which indicates that the self-supervised CNN can learn a meaningful denoising mapping as similar to supervised learning. Once the CNN is trained via the loss function (\ref{loss_2}), one can feed some newly masked instances into the network to generate multiple predictions and calculate their average as the noise attenuation result, which could effectively reduce the variance of the output; see inference details in Sec. \ref{Sec_inference}.
 \begin{figure*}[t]
	\scriptsize
	\setlength{\tabcolsep}{0.9pt}
	\begin{center}
		\begin{tabular}{c}
			\includegraphics[width=0.85\textwidth]{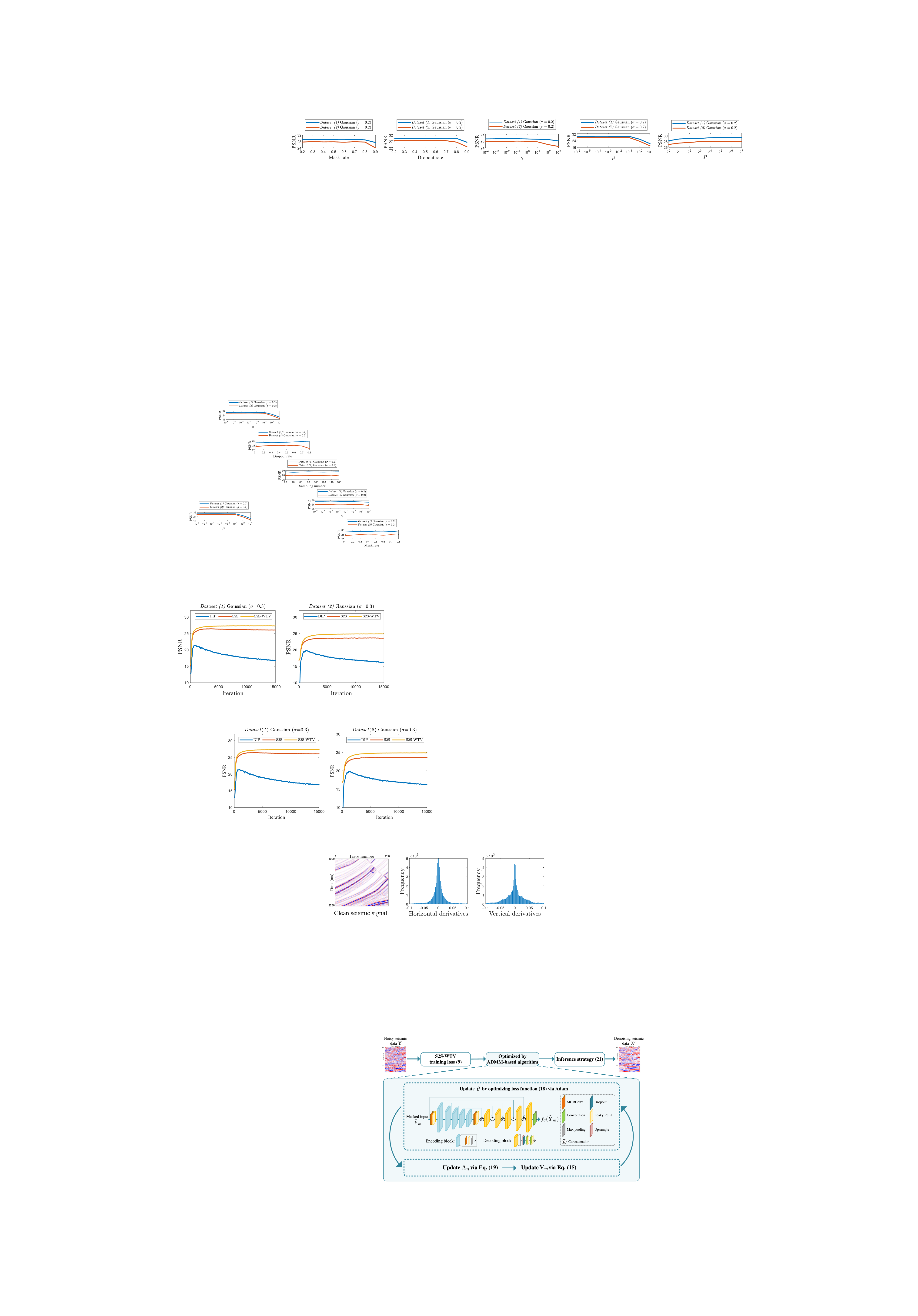}
			\vspace{-0.2cm}
		\end{tabular}
	\end{center}
	\caption{The overall flowchart of the proposed S2S-WTV for seismic data noise attenuation.\label{fig_flow}}
	\vspace{-0.3cm}
\end{figure*}
 \begin{figure}[t]
	\scriptsize
	\setlength{\tabcolsep}{0.9pt}
	\begin{center}
		\begin{tabular}{c}
			\includegraphics[width=0.47\textwidth]{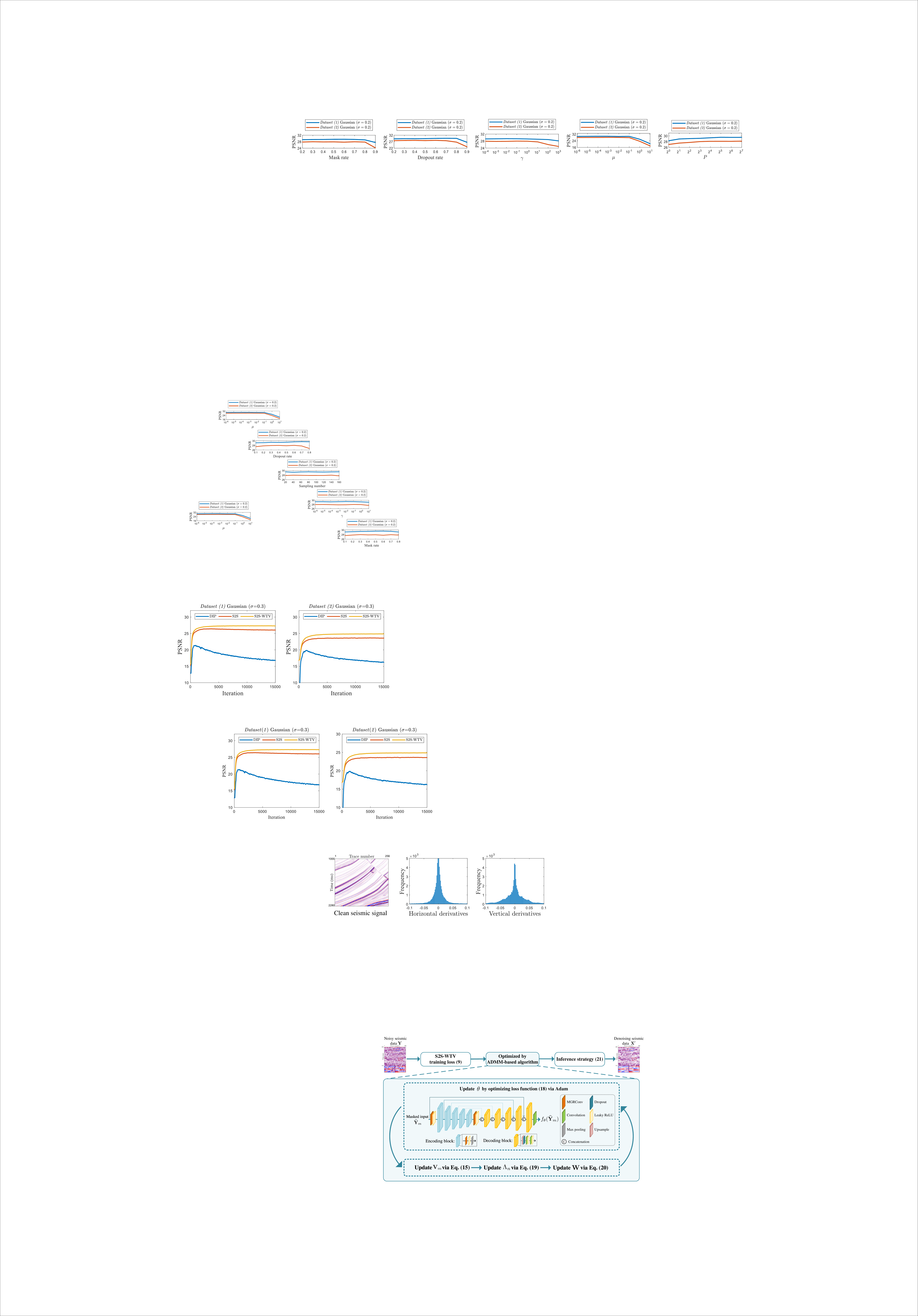}
			\vspace{-0.2cm}
		\end{tabular}
	\end{center}
	\caption{The frequency of horizontal/vertical derivative values of a clean seismic signal. The horizontal derivatives are more focused to zeros, which indicates that the signal is smoother along the horizontal direction. However, due to some non-smooth structures such as geological faults, there are still some horizontal derivative values that are not close to zero. The WTV regularization can distinguish between more or less smooth structures by assigning different weights to the elements of the derivative matrix.\label{WTV_explain}}
	\vspace{-0.2cm}
\end{figure}
\subsubsection{Trace-Wise Masking Strategy}\label{trace_mask}
It was shown in \cite{S2S} that the binary Bernoulli mask, whose elements are independently sampled from a Bernoulli distribution, is suitable for image denoising under the S2S framework. However, we discover that such an element-wise Bernoulli sampling strategy is not effective enough for seismic data denoising; see Fig. \ref{fig_mask}. Digging into the denoising process, we find the underlying reason is that the correlations between adjacent time samples (rows) of seismic data are weak, while the correlations between adjacent traces (columns) are strong due to the smooth geological structure. Thus, it is hard to use the information of adjacent rows to predict the masked clean row due to the irrelevance between them, but it is much easier to use the information of adjacent columns to predict the masked clean column due to their high relevance. Such an asymmetry phenomenon makes element-wise masks not optimal for learning the intrinsic correlations of seismic data.\par 
Motivated by this phenomenon, we design a trace-wise masking strategy in the S2S framework to help the model better adapt to the structures of seismic data. Specifically, we impose the mask to some Bernoulli sampled columns of the data, and then use the unmasked columns to predict the masked columns. Formally speaking, we first generate a binary Bernoulli vector ${\bf v}_n\in\{0,1\}^{W}$ for the $n$-th instance, where the elements of ${\bf v}_n$ are independently drawn from a Bernoulli distribution with probability $p\in(0,1)$, i.e., $P({\bf v}_{n_{(j)}}=1)=p$ and $P({\bf v}_{n_{(j)}}=0)=1-p$, in which ${\bf v}_{n_{(j)}}$ and ${\bf v}_{n_{(k)}}$ are independent for all $j,k$s. Here, ${\bf v}_{n_{(j)}}$ denotes the $j$-th element of ${\bf v}_n$. Using the binary vector, the trace-wise mask ${\bf M}_n$ is constructed by 
\begin{equation}\label{eq_M}
\begin{aligned}
\begin{split}
{\bf M}_{n_{(i,j)}}=\left \{
\begin{array}{lr}
    1,& {\bf v}_{n_{(j)}}=1,\\
    0,& {\bf v}_{n_{(j)}}=0,
\end{array}
\right.
\end{split}
\end{aligned}
\end{equation}
where ${\bf M}_{n_{(i,j)}}$ denotes the $i,j$-th element of ${\bf M}_n$. In Fig. \ref{fig_mask}, we show the noise attenuation results of our method with different masking strategies. We can see that the trace (column)-wise masks are more effective than element-wise and row-wise masks due to the high correlations of adjacent traces of seismic data, which validates the superiority of trace-wise masks for seismic data denoising. The trace-wise masking strategy together with the self-supervised learning can effectively reduce the variance of the output and learn an effective denoising mapping, which is highly related to training the CNN with supervised pairs $\{\widehat{\bf Y}_n,{\bf X}\}_{i=1}^N$; see Lemma \ref{lemma}.
\subsection{WTV Regularized Self-Supervised Training Loss}
Although the S2S learning paradigm is effective for noise attenuation, it does not take full use of the intrinsic domain knowledge of seismic data, which may limit its generalization performances over complex scenarios. Therefore, we introduce the WTV regularization into the training loss to capture the intrinsic local smooth structure of seismic data to further enhance the denoising capability. The WTV regularization can preserve the local smoothness of seismic data by minimizing the derivative values between adjacent elements, which increases the impedance to irregular noise. Specifically, the WTV of a matrix ${\bf X}\in{\mathbb R}^{H\times W}$ is computed by
\begin{equation}
\left\lVert{\bf X}\right\rVert_{\scriptscriptstyle{\rm TV},{\bf W}}:=\left\lVert{\bf W}\odot({\nabla_h{{\bf X}})}\right\rVert_{\ell_1},
\end{equation} 
where $\left\lVert\cdot\right\rVert_{\ell_1}$ denotes the matrix $\ell_1$-norm, $\nabla_h$ denotes the horizontal derivative operator defined as
\begin{equation}
(\nabla_h{\bf X})_{(i,j)}:={\bf X}_{(i,j+1)}-{\bf X}_{(i,j)}, 
\end{equation}
and ${\bf W}\in{\mathbb R}^{H\times(W-1)}$ is the non-negative weight matrix that assigns different weights to different elements of the derivative matrix $\nabla_h{\widehat{\bf X}}$. Here, we only consider the horizontal derivatives of the signal and neglects the vertical derivatives since the seismic data possesses stronger smoothness along the horizontal direction, i.e., the smoothness between adjacent traces are more distinct than those in the vertical direction; see Fig. \ref{WTV_explain}. This has also been emphasized in Sec. \ref{trace_mask}, i.e., the correlations between adjacent traces are stronger than the correlations between adjacent time samples.\par 
Meanwhile, we employ the weight matrix $\bf W$ to distinguish between smoother regions and less smooth edges and deails, such as the geological fault. Such regions appear to be not smooth and thus can be assigned with a lower weight in $\bf W$. The weight matrix is automatically updated in the denoising process; see details in the algorithm in Sec. \ref{Sec_ADMM}.\par 
We introduce the WTV regularization into the self-supervised training loss (\ref{loss_2}) and re-write the loss as
\begin{equation}\label{loss_WTV}
\begin{split}
\min_\theta\sum_{n=1}^N \big{(}&\left\lVert ({\bf Y}-f_\theta(\widehat{\bf Y}_n))\odot({\bf 1}-{\bf M}_n)\right\rVert_{F}^2\\
&\quad\quad\quad\quad\quad\quad\quad
+\gamma\left\lVert f_\theta(\widehat{\bf Y}_n)\right\rVert_{\scriptscriptstyle{\rm TV},{\bf W}}\big{)},
\end{split}
\end{equation}
where $\gamma$ is a trade-off parameter. The WTV-regularized self-supervised loss (\ref{loss_WTV}) only uses the observed noisy data $\bf Y$ as training data. Training the CNN $f_\theta(\cdot)$ using such a self-supervised loss is very related to training the CNN with pairs of supervised data $\{\widehat{\bf Y}_n,{\bf X}\}_{i=1}^N$ and WTV regularization, as stated in Lemma \ref{lemma}. 
\begin{lemma}\label{lemma}
Suppose that $\{{\bf M}_n\}_{n=1}^N$ are the trace-wise Bernoulli masks defined as in (\ref{eq_M}). Assume that (\ref{noisy_model}) holds and each element of the noisy matrix $\bf N$ follows Gaussian distribution with zero mean and variance ${\bf \sigma}^2\in{\mathbb R}^{H\times W}$ (${\bf \sigma}_{(i,j)}^2$ indicates the variance of ${\bf N}_{(i,j)}$), then the following equality (with $\bf N$ being the random variable) is true for any $f_\theta(\cdot)$: 
\begin{equation}
\begin{split}
&{\mathbb E}\Big{[}\sum_{n=1}^N \big{(}\left\lVert ({\bf Y}-f_\theta(\widehat{\bf Y}_n))\odot({\bf 1}-{\bf M}_n)\right\rVert_{F}^2
\\&\quad\quad\quad\quad\quad\quad\quad\quad\quad\quad\;\;
\quad
+\gamma\left\lVert f_\theta(\widehat{\bf Y}_n)\right\rVert_{\scriptscriptstyle{\rm TV},{\bf W}}
\big{)}\Big{]}\\
=&{\mathbb E}\Big{[}\sum_{n=1}^N \big{(}\left\lVert ({\bf X}-f_\theta(\widehat{\bf Y}_n))\odot({\bf 1}-{\bf M}_n)\right\rVert_{F}^2\\
&\quad\;\;
\quad+\left\lVert{\bf \sigma}\odot({\bf 1}-{\bf M}_n)\right\rVert_F^2+\gamma\left\lVert f_\theta(\widehat{\bf Y}_n)\right\rVert_{\scriptscriptstyle{\rm TV},{\bf W}}\big{)}\Big{]}.
\end{split}
\end{equation}
\end{lemma}
\begin{proof}
Note that the choices of binary masks $\{{\bf M}_n\}_{n=1}^N$ and the introduce of WTV regularization do not hurt the correctness of Proposition 1 in \cite{S2S}, and directly follow this theory yields the desired result.  
\end{proof}
Lemma \ref{lemma} indicates that in terms of expectation, using pairs of trace-wise masked samples $\widehat{\bf Y}_n$ and the observed data $\bf Y$ with WTV regularization to train the CNN produces similar training losses to that of using pairs of trace-wise masked samples and the ground-truth $\bf X$ with WTV regularization. Therefore, our self-supervised loss (\ref{loss_WTV}) is expected to learn an effective denoising mapping even without ground-truth training data.\par 
Moreover, our S2S-WTV simultaneously enjoys the powerful representation abilities of self-supervised CNN and the generalization abilities of WTV regularizer with interpretable domain knowledge, and thus is expected to learn an effective denoising mapping. Compared to traditional hand-crafted denoisers, our S2S-WTV leverages the expressiveness of CNN to better capture the complex structures of seismic signals. Compared to previous deep learning denoisers, S2S-WTV does not need additional training data and benefits from the generalizability of domain knowledge brought from the WTV regularizer. Therefore, our method has better generalization abilities for various types of seismic data and noises, while classical deep learning denoisers may sometimes suffer from poor generalization performances over different seismic datasets (e.g., pre-stack and post-stack seismic data) and noises (e.g., random Gaussian noise and bandpass noise); see experimental validations in Sec. \ref{Sec_exp}.  
\begin{algorithm}[t]
	\begin{spacing}{1.02}
		\renewcommand\arraystretch{1.2}
		\caption[Caption for LOF]{Self-Supervised Training Strategy of S2S-WTV for Seismic Data Noise Attenuation}\label{alg_1}
		\begin{algorithmic}[1]
			\renewcommand{\algorithmicrequire}{\textbf{Input:}} 
			\Require
			Noisy seismic data $\{{\bf Y}_k\}_{k=1}^K$, iteration numbers ${T}_1$ and $T_k<<T_1$ ($k=2,3,\cdots,K$), hyperparamters $\gamma,\mu$;
			\renewcommand{\algorithmicrequire}{\textbf{Initialization:}} 
			\Require Randomly initialize $\theta_1$, generate trace-wise masks $\{{\bf M}_n\}_{n=1}^N$, ${\Lambda_n}={\bf 0}$, $t=0$, ${\bf W}={\bf 1}$;
			\For {$k$=1:$K$}
			\State $\theta_k=\theta_1$; Generate training instances via (\ref{eq_Y_n});
			\While {$t\leq {T}_k$}
			\State Update $\mathbf{\mathcal{V}}_n$s via (\ref{V_solution});
			\State Update $\theta_k$ via (\ref{loss_theta});
			\State Update ${{\Lambda}}_n$s via (\ref{lambda_n});
			\If{$t<3000$}
			\State Update the weight matrix $\bf W$ via (\ref{eq_W});
			\EndIf
			\State $t = t+1$;
			\EndWhile
			\EndFor
			\State Using the trained denoising CNNs $\{f_{\theta_k}(\cdot)\}_{k=1}^K$ to predict clean signals $\{{\bf X}'_k\}_{k=1}^K$ via (\ref{eq_inference});
			\renewcommand{\algorithmicrequire}{\textbf{Output:}}
			\Require The estimated clean seismic signals $\{{\bf X}'_k\}_{k=1}^K$;
		\end{algorithmic}
	\end{spacing}
\end{algorithm}
\begin{table*}[!h]
	\caption{The average quantitative results by different methods for noise attenuation in synthetic post-stack seismic {\it DATASETS (1)-(3)}. The {\bf BEST} and \underline{second-best} values are highlighted. (PSNR $\uparrow$, SSIM $\uparrow$, and LS $\downarrow$)\label{tab_denoising}}\vspace{-0.4cm}
	\begin{center}
		\scriptsize
		\setlength{\tabcolsep}{2.1pt}
		\begin{spacing}{1.2}
			\begin{tabular}{clccccccccccccccccccc}
				\toprule
				\multicolumn{2}{c}{Noise}&\multicolumn{3}{c}{Gaussian ($\sigma=0.1$)}&\multicolumn{3}{c}{Gaussian ($\sigma=0.2$)}&\multicolumn{3}{c}{Gaussian ($\sigma=0.3$)}&\multicolumn{3}{c}{Bandpass ($\sigma=0.1$)}&\multicolumn{3}{c}{Bandpass ($\sigma=0.2$)}&\multicolumn{3}{c}{Bandpass ($\sigma=0.3$)}&\multirow{3}*{\tabincell{c}{
						{Time}\\{(second)}}}\\
				\cmidrule{1-20}
				Data&Method&PSNR &SSIM &LS \;\; &PSNR &SSIM &LS \;\; &PSNR &SSIM &LS \;\; &PSNR &SSIM &LS \;\;&PSNR &SSIM &LS \;\;&PSNR &SSIM &LS &~\\
				\midrule
				\multirow{8}*{\tabincell{c}{
						{\it Dataset (1)}\\{(256$\times$256)}}}
				&Observed&{20.04}&{0.851}&{\--\--}\;\;&{13.99}&{0.578}&{\--\--}\;\;&{10.52}&{0.381}&{\--\--}\;\;&{28.63}&{0.976}&{\--\--}\;\;&{22.56}&{0.911}&{\--\--}\;\;&{19.10}&{0.822}&{\--\--}\;\;&{\--\--}\\
				&BM3D&{28.81}&{0.976}&{0.096}\;\;&\underline{27.04}&\underline{0.953}&{0.099}\;\;&\underline{25.60}&\underline{0.923}&{0.127}\;\;&{30.99}&{0.985}&{0.199}\;\;&{29.74}&{0.981}&{0.147} \;\;&\underline{28.88}&\underline{0.976}&{0.159}\;\;&{2}\\
				&WNNM&{29.93}&{0.981}&{0.232}\;\;&{24.07}&{0.896}&{0.383}\;\;&{22.74}&{0.856}&{0.573}\;\;&{32.15}&{0.989}&{0.298}\;\;&{30.11}&{0.982}&{0.295}\;\;&{28.78}&{0.975}&{0.357}\;\;&{15} \\
				&MSSA&{25.60}&{0.952}&{0.144}\;\;&{19.94}&{0.832}&{0.149}\;\;&{16.59}&{0.686}&{0.155}\;\;&{31.58}&{0.988}&{0.186}\;\;&{27.51}&{0.969}&{0.177} \;\;&{24.48}&{0.939}&{0.181}\;\;&{2}\\
				&DDAE&{22.62}&{0.911}&{0.116}\;\;&{14.98}&{0.635}&{0.139}\;\;&{10.82}&{0.302}&{0.142}\;\;&{27.70}&{0.968}&{0.310}\;\;&{25.07}&{0.930}&{0.136} \;\;&{20.70}&{0.821}&{0.153}\;\;&{7}\\
				&DIP&{28.15}&{0.972}&\underline{0.084}\;\;&{25.32}&{0.947}&\underline{0.091}\;\;&{19.80}&{0.784}&\underline{0.087}\;\;&{31.39}&{0.984}&{0.123}\;\;&{27.75}&{0.971}&\underline{0.109} \;\;&{24.58}&{0.941}&{0.114}\;\;&{189}\\
				&PATCHUNET&\underline{31.57}&\underline{0.988}&{0.090}\;\;&{26.73}&{0.944}&{0.097}\;\;&{23.95}&{0.914}&{0.130}\;\;&\underline{35.19}&\underline{0.995}&\underline{0.118}\;\;&\underline{30.58}&\underline{0.984}&{0.125} \;\;&{26.81}&{0.964}&{0.178}\;\;&{101}\\
				&S2S-WTV&\bf{34.67}&\bf{0.990}&\bf{0.080}\;\;&\bf{29.28}&\bf{0.959}&\bf{0.079}\;\;&\bf{27.05}&\bf{0.955}&\bf{0.074}\;\;&\bf{40.10}&\bf{0.998}&\bf{0.111}\;\;&\bf{35.05}&\bf{0.993}&\bf{0.105} \;\;&\bf{32.33}&\bf{0.990}&\bf{0.109}\;\;&{154}\\
				\midrule
				\multirow{8}*{\tabincell{c}{
						{\it Dataset (2)}\\{(256$\times$256)}}}
				&Observed&{20.01}&{0.812}&{\--\--}\;\;&{13.93}&{0.517}&{\--\--}\;\;&{10.48}&{0.323}&{\--\--}\;\;&{26.53}&{0.951}&{\--\--}\;\;&{20.55}&{0.831}&{\--\--}\;\;&{17.05}&{0.687}&{\--\--}\;\;&{\--\--}\\
				&BM3D&{28.59}&{0.965}&{0.153}\;\;&\underline{26.24}&\underline{0.939}&{0.154}\;\;&\underline{24.22}&\underline{0.895}&{0.212}\;\;&{29.50}&{0.972}&{0.218}\;\;&\underline{28.87}&\underline{0.968}&{0.213}\;\;&{27.11}&{0.951}&{0.377}\;\;&{2} \\
				&WNNM&\underline{29.34}&\underline{0.971}&{0.221}\;\;&{22.51}&{0.878}&{0.706}\;\;&{19.93}&{0.604}&{0.644}\;\;&\underline{29.63}&\underline{0.973}&{0.302}\;\;&{28.67}&{0.966}&{0.315}\;\;&\underline{27.40}&\underline{0.955}&{0.424}\;\;&{33} \\
				&MSSA&{21.17}&{0.834}&{0.226}\;\;&{17.33}&{0.564}&{0.213}\;\;&{15.25}&{0.418}&{0.184}\;\;&{26.87}&{0.954}&{0.234}\;\;&{21.19}&{0.845}&{0.226} \;\;&{17.73}&{0.704}&{0.224}\;\;&{1}\\
				&DDAE&{21.25}&{0.808}&{0.127}\;\;&{14.53}&{0.516}&{0.143}\;\;&{10.24}&{0.317}&{0.140}\;\;&{26.15}&{0.934}&{0.278}\;\;&{21.14}&{0.821}&{0.181}\;\;&{17.74}&{0.708}&{0.182}\;\;&{7} \\
				&DIP&{25.53}&{0.935}&\underline{0.109}\;\;&{21.57}&{0.846}&\underline{0.104}\;\;&{19.16}&{0.716}&\underline{0.100}\;\;&{28.05}&{0.954}&\underline{0.141}\;\;&{24.25}&{0.912}&\underline{0.135} \;\;&{20.44}&{0.811}&\underline{0.150}\;\;&{183}\\
				&PATCHUNET&{23.74}&{0.885}&{0.170}\;\;&{22.27}&{0.843}&{0.142}\;\;&{20.85}&{0.788}&{0.139}\;\;&{24.15}&{0.896}&{0.240}\;\;&{22.95}&{0.864}&{0.204} \;\;&{21.51}&{0.816}&{0.203}\;\;&{186}\\
				&S2S-WTV&\bf{33.23}&\bf{0.987}&\bf{0.093}\;\;&\bf{28.13}&\bf{0.964}&\bf{0.091}\;\;&\bf{25.14}&\bf{0.925}&\bf{0.099}\;\;&\bf{36.02}&\bf{0.994}&\bf{0.127}\;\;&\bf{31.01}&\bf{0.981}&\bf{0.117} \;\;&\bf{28.18}&\bf{0.964}&\bf{0.144}\;\;&{155}\\
				\midrule
				\multirow{8}*{\tabincell{c}{
						{\it Dataset (3)}\\{(256$\times$256)}}}
				&Observed&{20.01}&{0.802}&{\--\--}\;\;&{13.98}&{0.505}&{\--\--}\;\;&{10.47}&{0.311}&{\--\--}\;\;&{27.04}&{0.953}&{\--\--}\;\;&{21.07}&{0.838}&{\--\--}\;\;&{17.61}&{0.701}&{\--\--}\;\;&{\--\--}\\
				&BM3D&\underline{29.98}&\underline{0.973}&{0.141}\;\;&\underline{26.52}&\underline{0.940}&{0.156}\;\;&\underline{24.91}&\underline{0.911}&{0.251}\;\;&\underline{31.11}&\underline{0.980}&{0.210}\;\;&\underline{29.95}&\underline{0.973}&{0.203}\;\;&{28.08}&\underline{0.959}&{0.322} \;\;&{2}\\
				&WNNM&{29.11}&{0.966}&{0.240}\;\;&{25.94}&{0.936}&{0.539}\;\;&{21.52}&{0.732}&{0.613}\;\;&{30.01}&{0.973}&{0.326}\;\;&{28.91}&{0.965}&{0.319} \;\;&\underline{28.19}&\underline{0.959}&{0.368}\;\;&{33}\\
				&MSSA&{22.05}&{0.835}&{0.248}\;\;&{18.95}&{0.710}&{0.181}\;\;&{16.13}&{0.549}&{0.166}\;\;&{27.81}&{0.961}&{0.182}\;\;&{21.93}&{0.862}&{0.194} \;\;&{18.48}&{0.737}&{0.208}\;\;&{1}\\
				&DDAE&{21.05}&{0.804}&{0.131}\;\;&{14.47}&{0.219}&{0.150}\;\;&{10.86}&{0.160}&{0.143}\;\;&{26.56}&{0.849}&{0.278}\;\;&{21.66}&{0.836}&{0.183} \;\;&{17.97}&{0.708}&{0.178}\;\;&{6}\\
				&DIP&{27.14}&{0.951}&\underline{0.100}\;\;&{21.99}&{0.848}&\underline{0.098}\;\;&{19.56}&{0.713}&\underline{0.108}\;\;&{28.74}&{0.915}&\underline{0.163}\;\;&{25.95}&{0.939}&\underline{0.124} \;\;&{22.24}&{0.857}&\underline{0.129}\;\;&{179}\\
				&PATCHUNET&{27.97}&{0.959}&{0.115}\;\;&{25.32}&{0.926}&{0.108}\;\;&{22.72}&{0.864}&{0.119}\;\;&{29.17}&{0.969}&\underline{0.163}\;\;&{26.90}&{0.948}&{0.143}\;\;&{24.71}&{0.915}&{0.159} \;\;&{163}\\
				&S2S-WTV&\bf{33.26}&\bf{0.988}&\bf{0.094}\;\;&\bf{27.68}&\bf{0.946}&\bf{0.095}\;\;&\bf{25.37}&\bf{0.921}&\bf{0.096}\;\;&\bf{37.30}&\bf{0.995}&\bf{0.123}\;\;&\bf{31.91}&\bf{0.984}&\bf{0.118} \;\;&\bf{28.57}&\bf{0.966}&\bf{0.117}\;\;&{153}\\
				\bottomrule
			\end{tabular}
		\end{spacing}
	\end{center}
	\vspace{-0.6cm}
\end{table*}   
\subsection{Network Structures}
The overall structure of the CNN $f_\theta(\cdot)$ is illustrated in Fig. \ref{fig_flow}. The deep CNN is based on a U-Net structure, which contains the encoding stage, decoding stage, and skip connections. Motivated by recent work of self-supervised denoising network\cite{MGRConv}, we leverage the mask guided residual convolution (MGRConv) in the encoding stage, which can more faithfully extract useful information of the input under random masks\cite{MGRConv}. The dropout mechanism is used in each decoding block by following \cite{S2S}. The CNN takes the masked seismic data $\widehat{\bf Y}_n$ as the input and is expected to output the clean signal. We suggest readers referring to \cite{MGRConv,DIP,S2S,DIP_attention} for more details and discussions about the design of CNN structures for self-supervised learning.
  \begin{figure*}[!h]
	\scriptsize
	\setlength{\tabcolsep}{0.9pt}
	\begin{center}
		\begin{tabular}{ccccccccc}
			\includegraphics[width=0.108\textwidth]{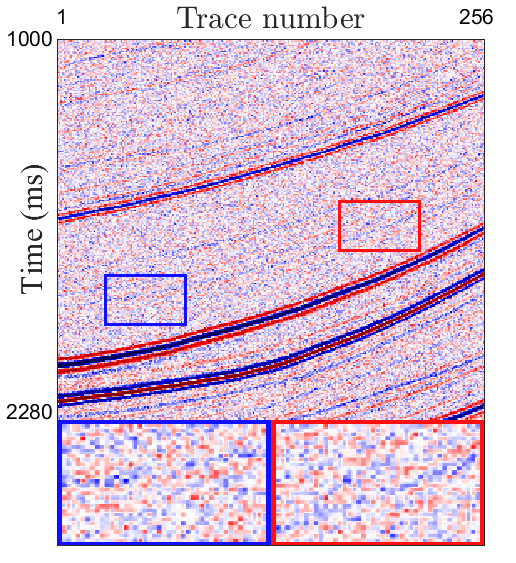}&
			\includegraphics[width=0.108\textwidth]{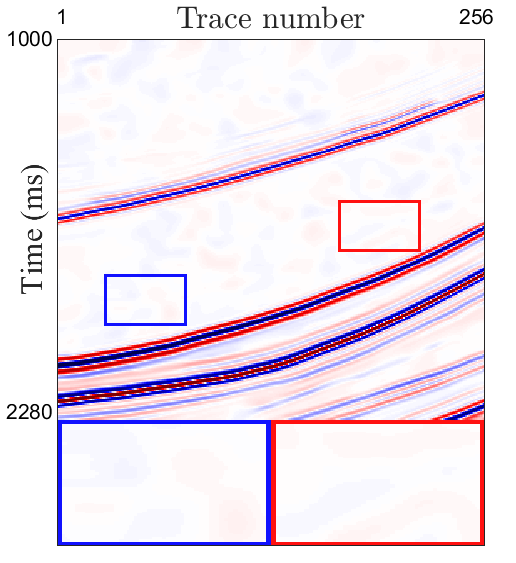}&
			\includegraphics[width=0.108\textwidth]{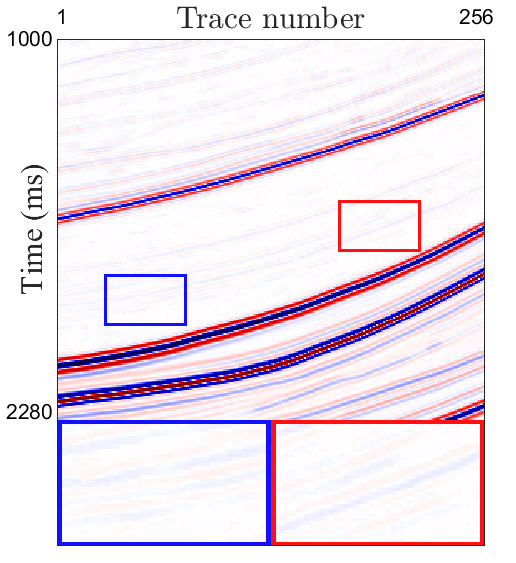}&
			\includegraphics[width=0.108\textwidth]{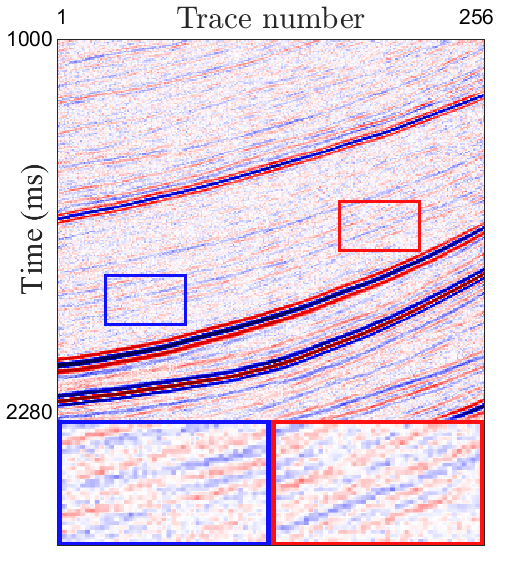}&
			\includegraphics[width=0.108\textwidth]{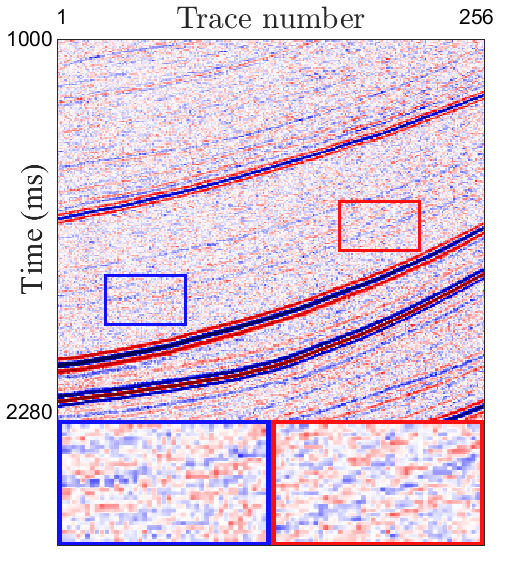}&
			\includegraphics[width=0.108\textwidth]{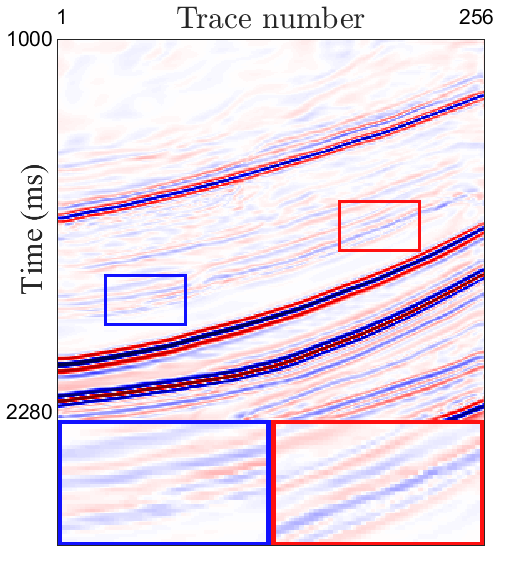}&
			\includegraphics[width=0.108\textwidth]{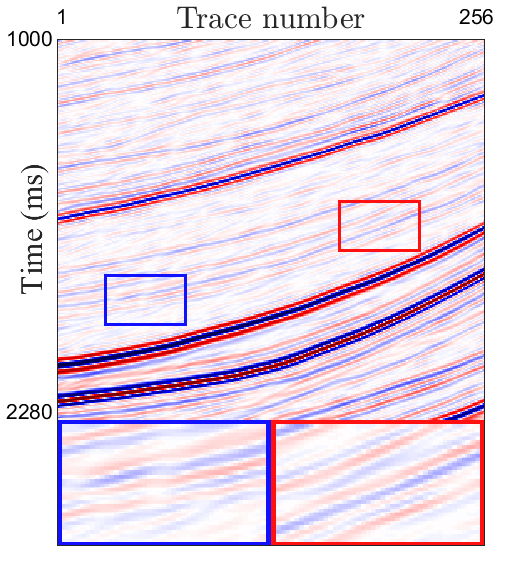}&
			\includegraphics[width=0.108\textwidth]{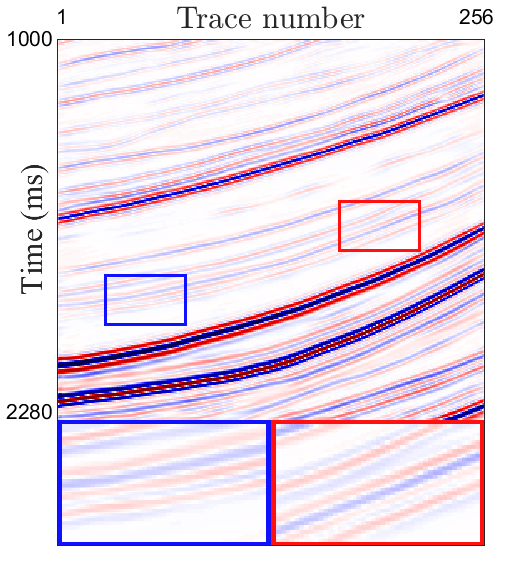}&
			\includegraphics[width=0.108\textwidth]{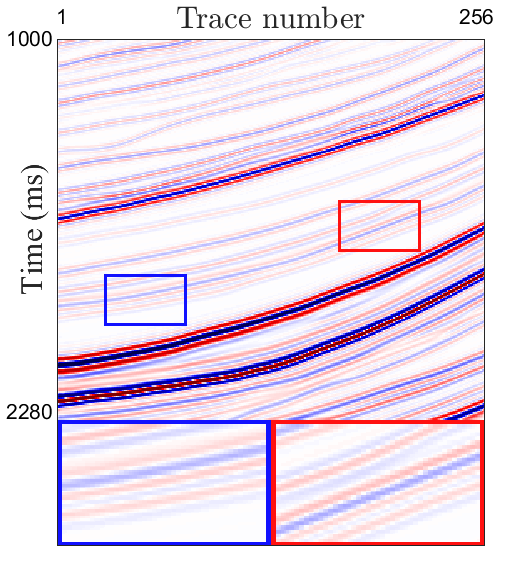}\\
			&
			\includegraphics[width=0.108\textwidth]{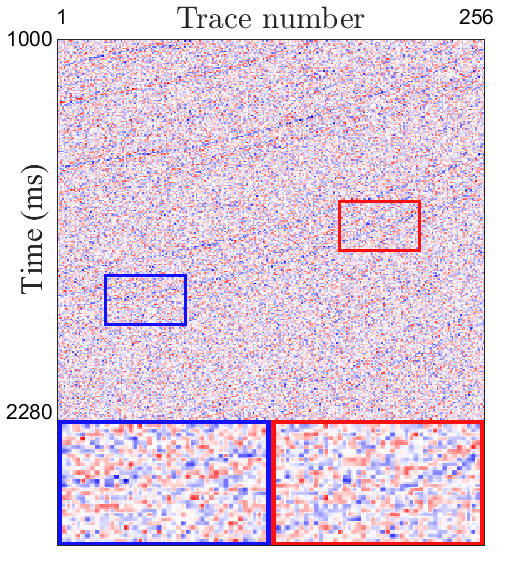}&
			\includegraphics[width=0.108\textwidth]{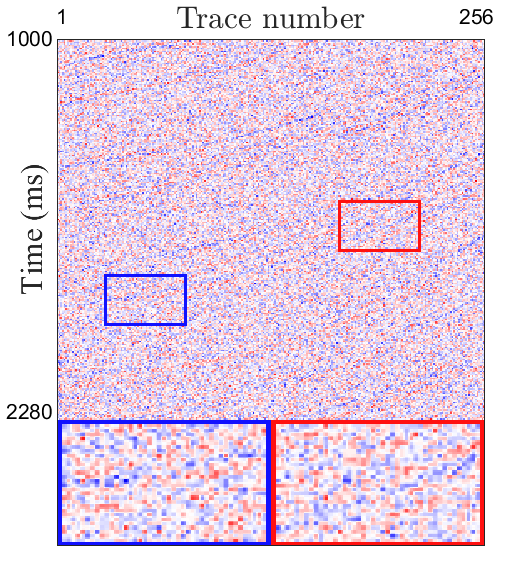}&
			\includegraphics[width=0.108\textwidth]{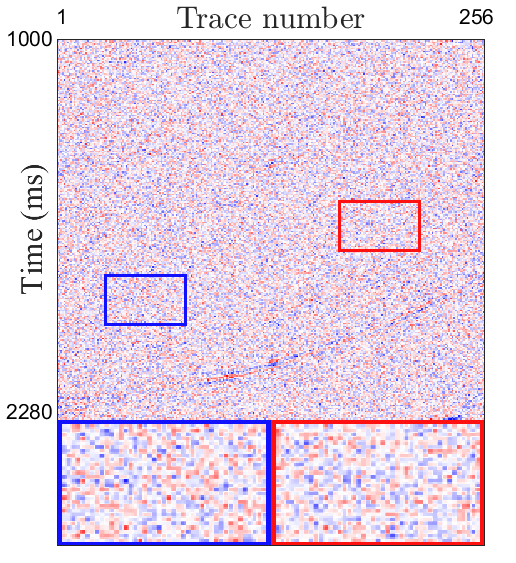}&
			\includegraphics[width=0.108\textwidth]{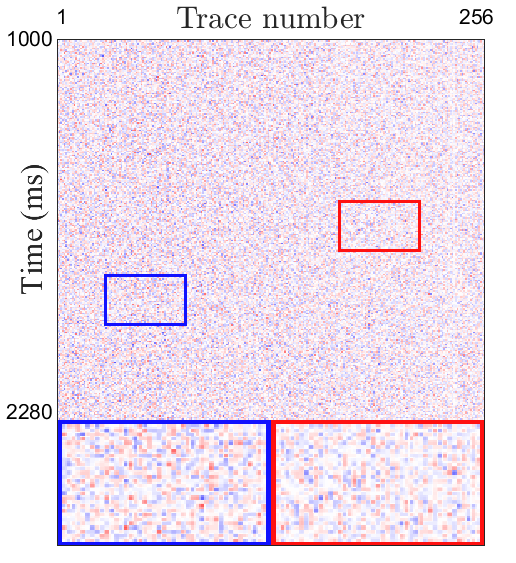}&
			\includegraphics[width=0.108\textwidth]{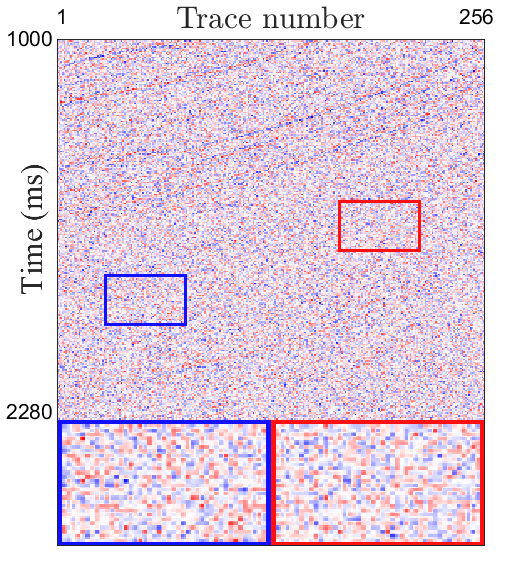}&
			\includegraphics[width=0.108\textwidth]{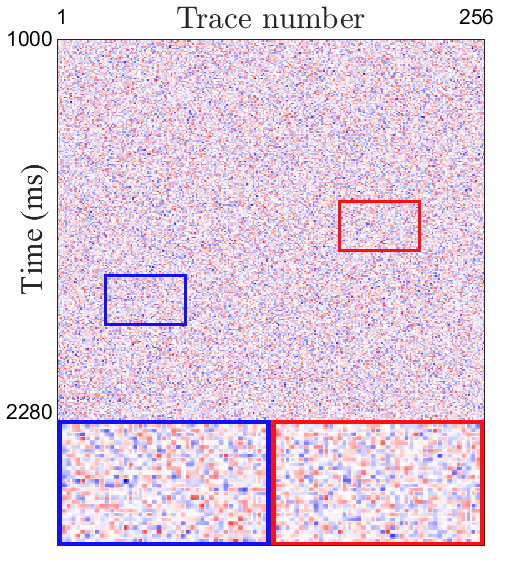}&
			\includegraphics[width=0.108\textwidth]{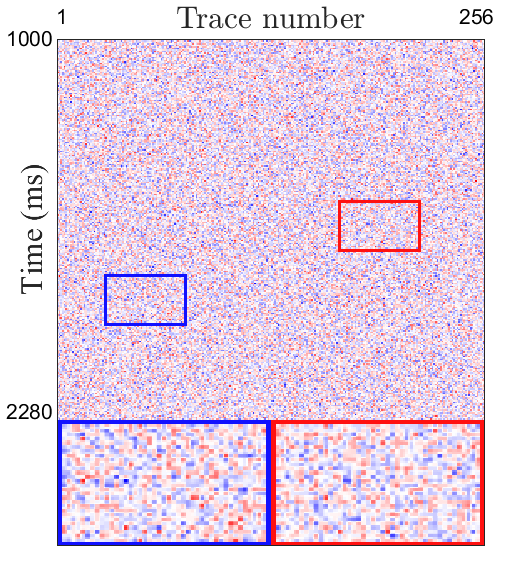}&
			\includegraphics[width=0.108\textwidth]{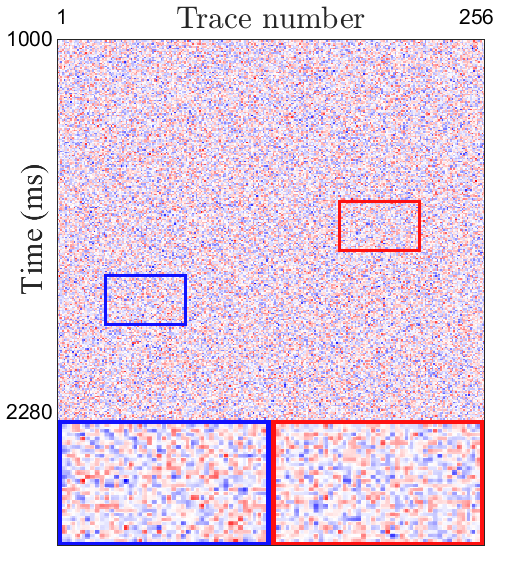}\\
			\includegraphics[width=0.108\textwidth]{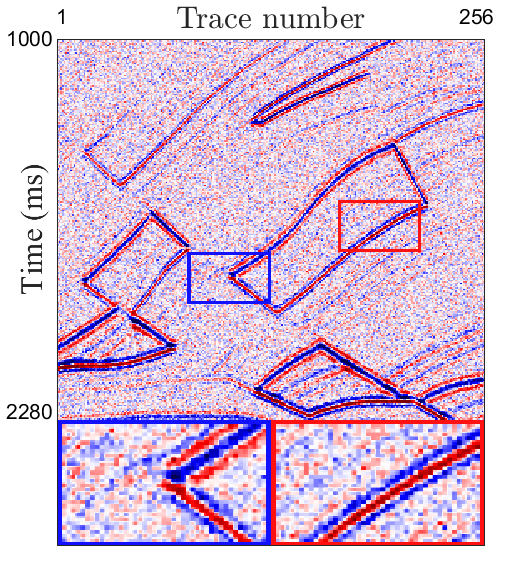}&
			\includegraphics[width=0.108\textwidth]{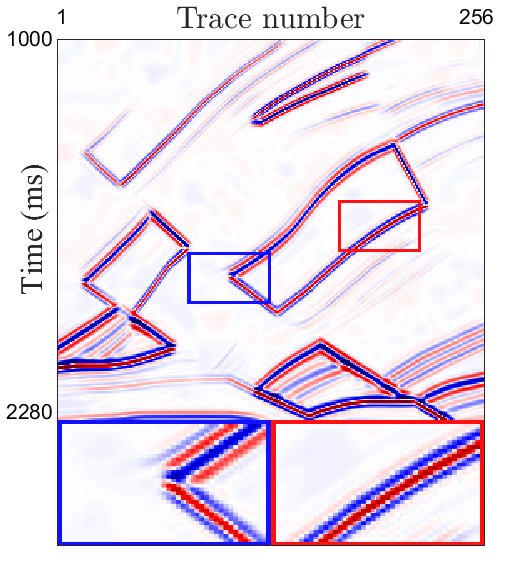}&
			\includegraphics[width=0.108\textwidth]{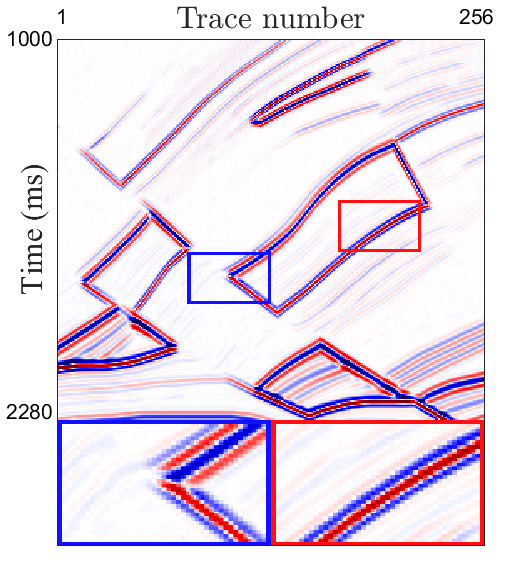}&
			\includegraphics[width=0.108\textwidth]{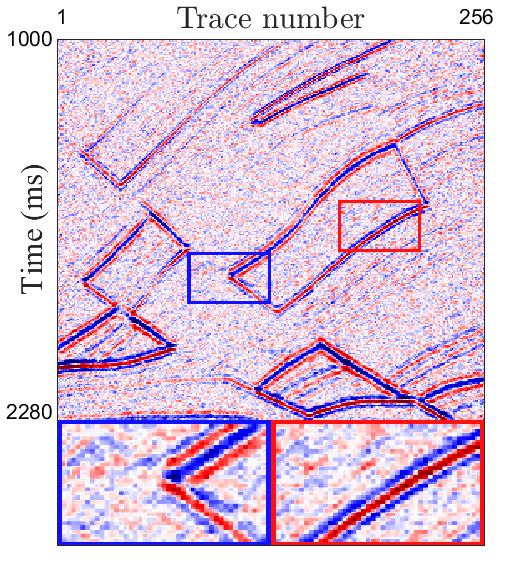}&
			\includegraphics[width=0.108\textwidth]{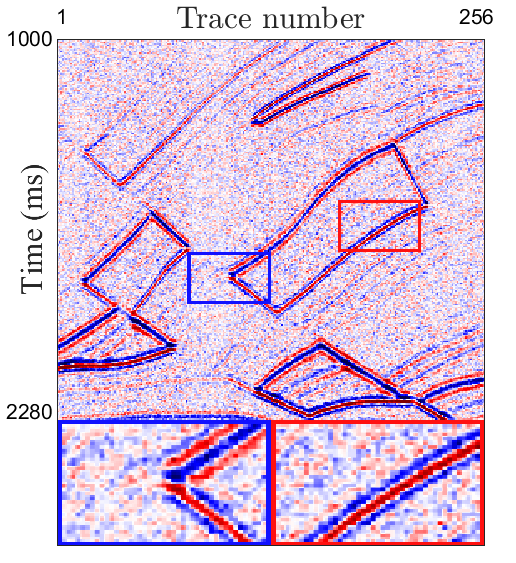}&
			\includegraphics[width=0.108\textwidth]{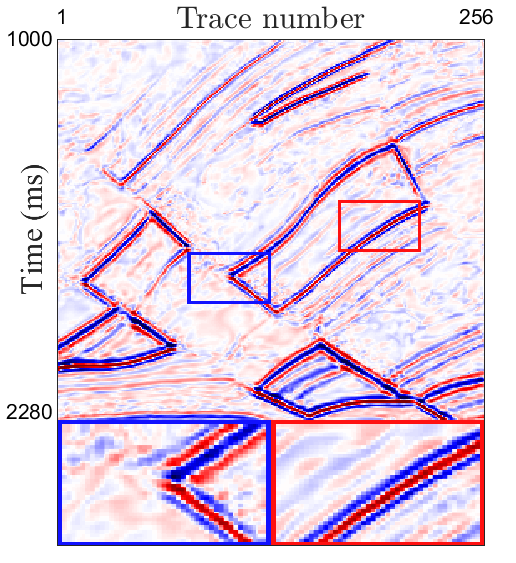}&
			\includegraphics[width=0.108\textwidth]{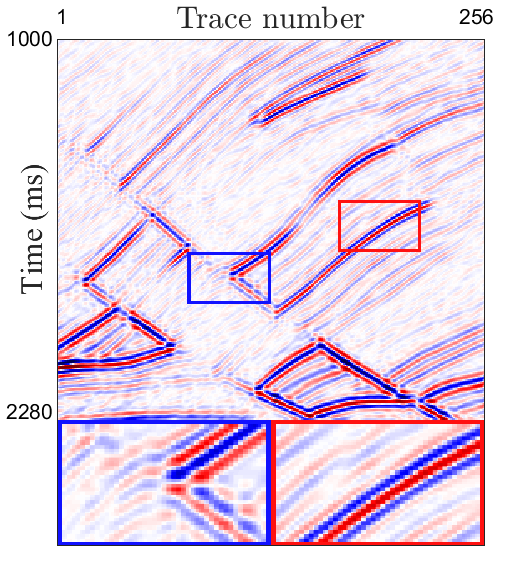}&
			\includegraphics[width=0.108\textwidth]{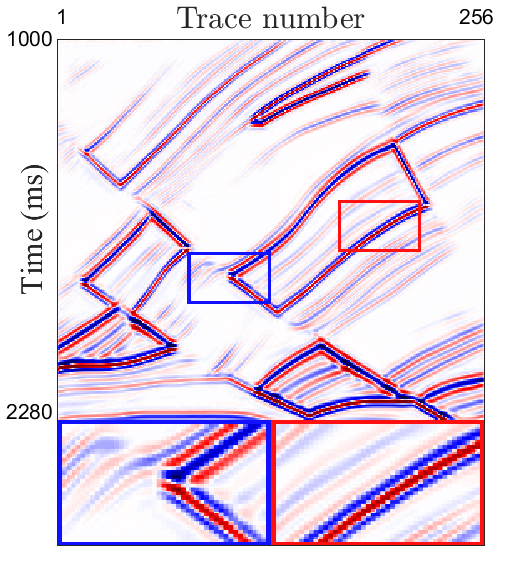}&
			\includegraphics[width=0.108\textwidth]{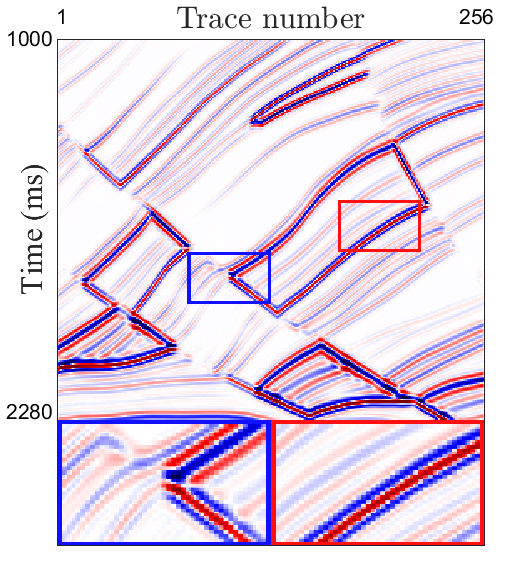}\\
			&
			\includegraphics[width=0.108\textwidth]{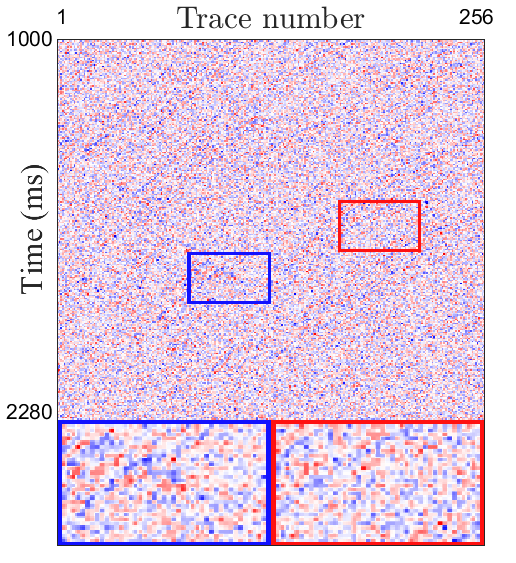}&
			\includegraphics[width=0.108\textwidth]{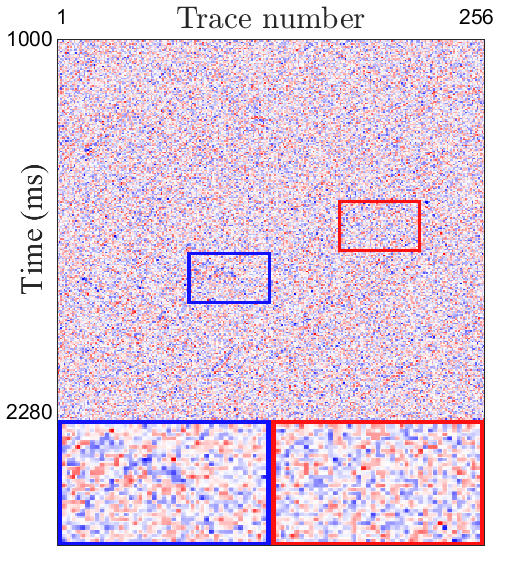}&
			\includegraphics[width=0.108\textwidth]{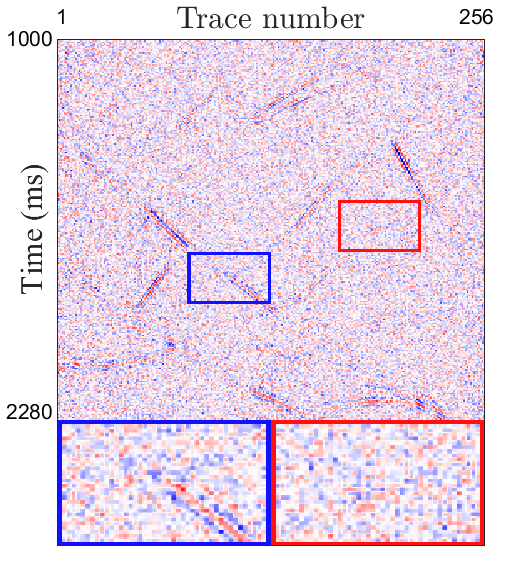}&
			\includegraphics[width=0.108\textwidth]{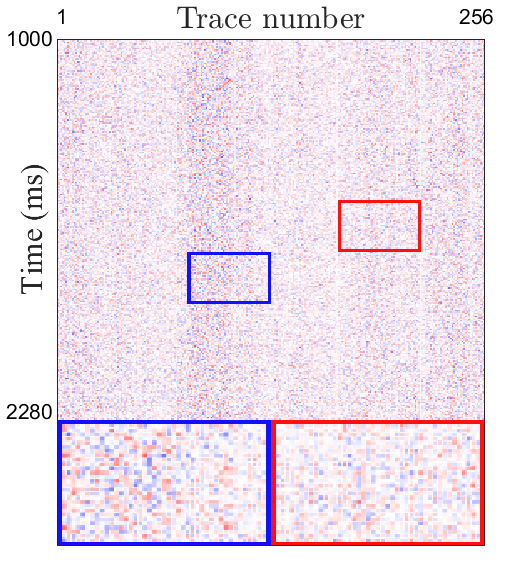}&
			\includegraphics[width=0.108\textwidth]{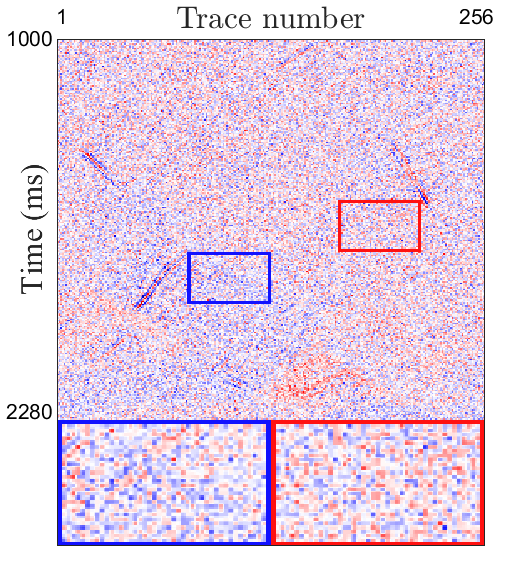}&
			\includegraphics[width=0.108\textwidth]{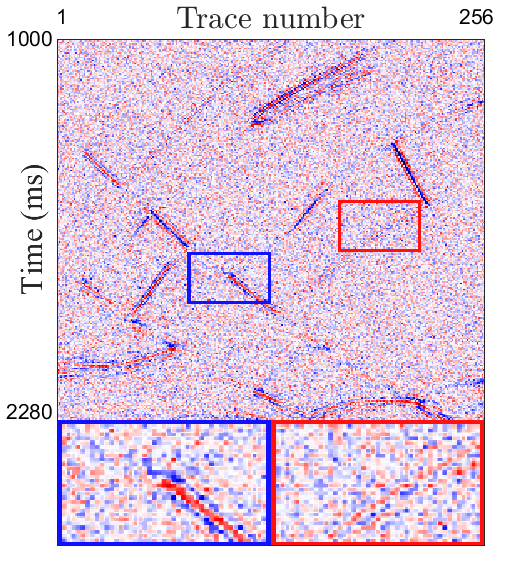}&
			\includegraphics[width=0.108\textwidth]{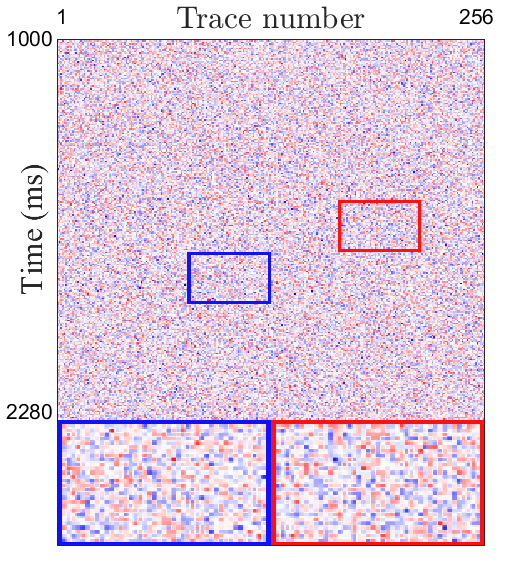}&
			\includegraphics[width=0.108\textwidth]{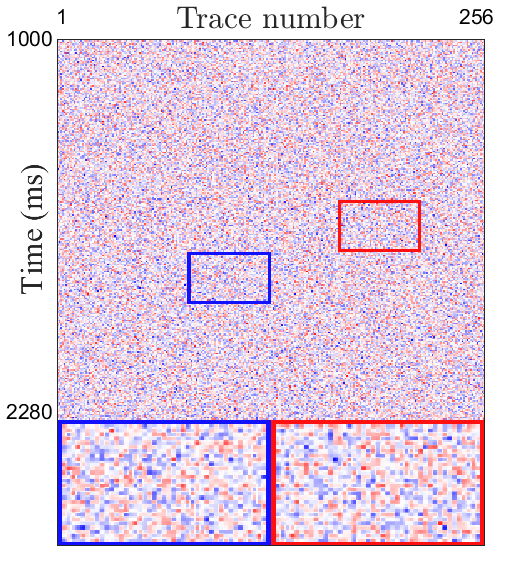}\\
			\includegraphics[width=0.108\textwidth]{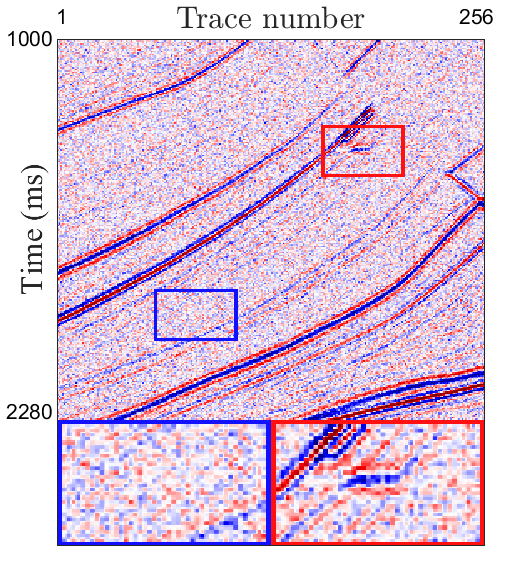}&
			\includegraphics[width=0.108\textwidth]{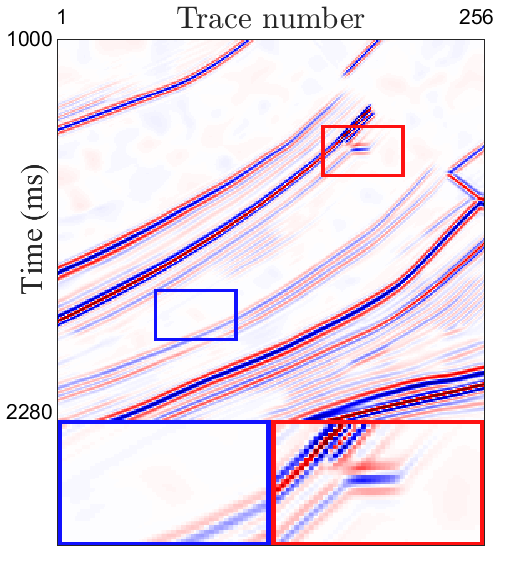}&
			\includegraphics[width=0.108\textwidth]{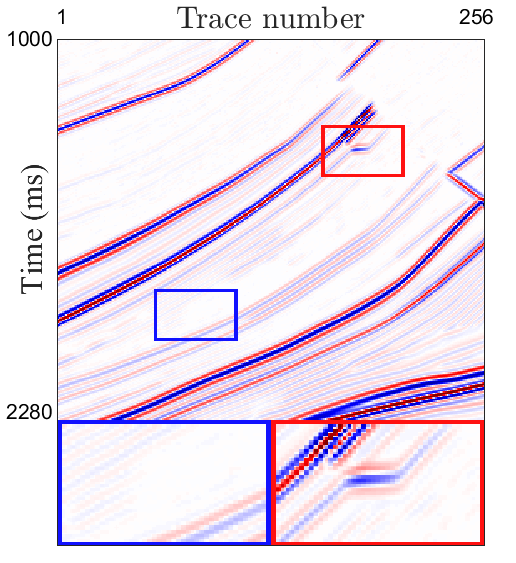}&
			\includegraphics[width=0.108\textwidth]{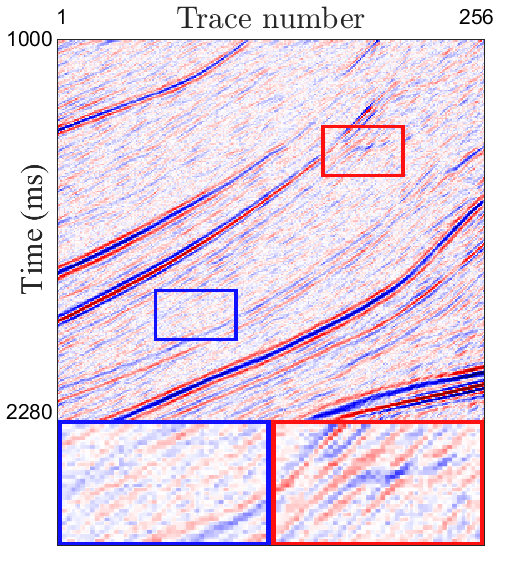}&
			\includegraphics[width=0.108\textwidth]{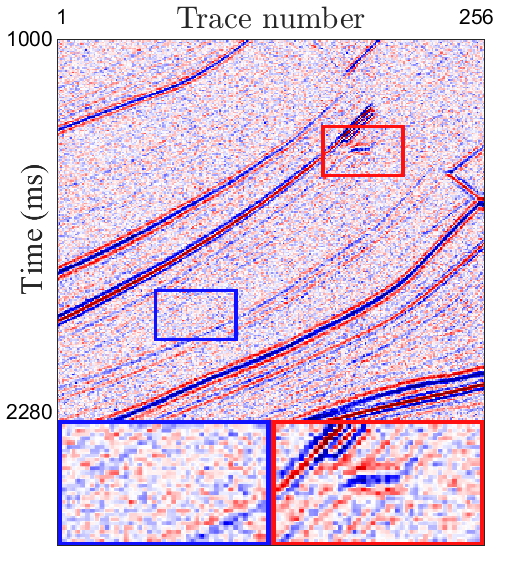}&
			\includegraphics[width=0.108\textwidth]{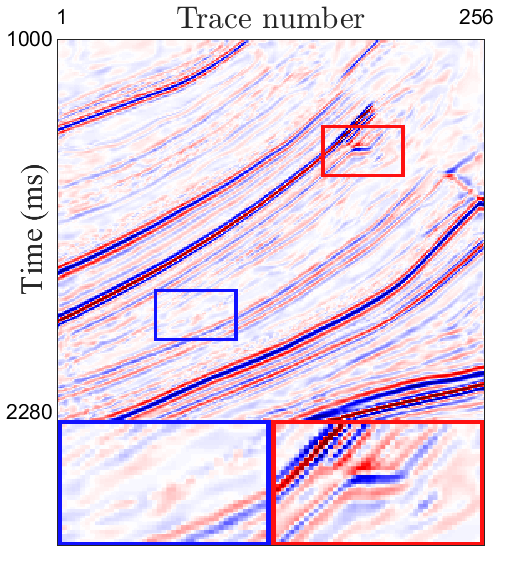}&
			\includegraphics[width=0.108\textwidth]{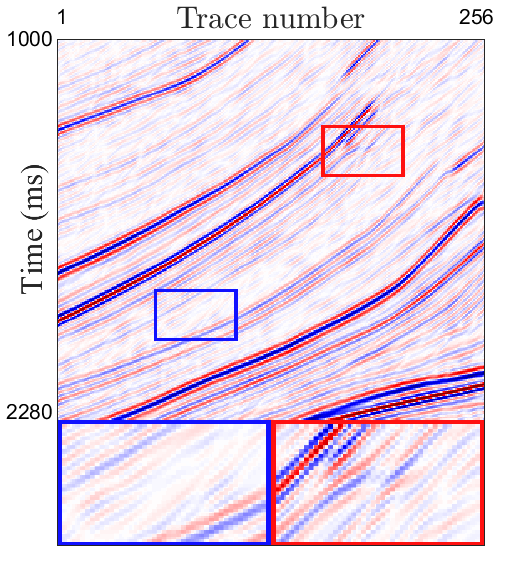}&
			\includegraphics[width=0.108\textwidth]{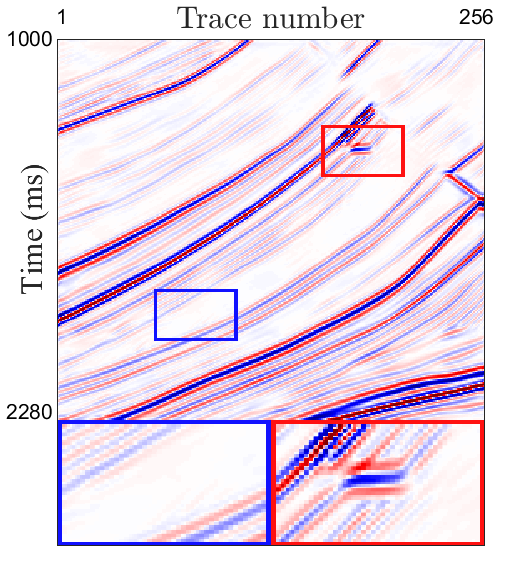}&
			\includegraphics[width=0.108\textwidth]{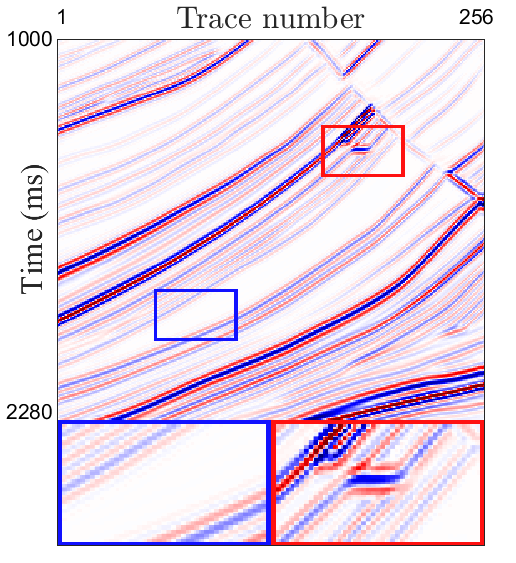}\\
			&
			\includegraphics[width=0.108\textwidth]{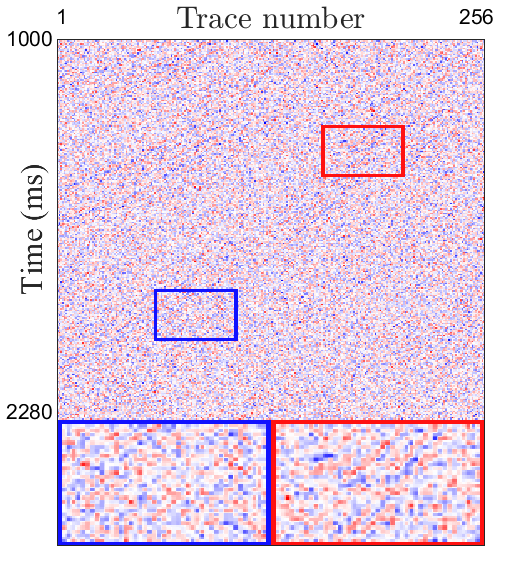}&
			\includegraphics[width=0.108\textwidth]{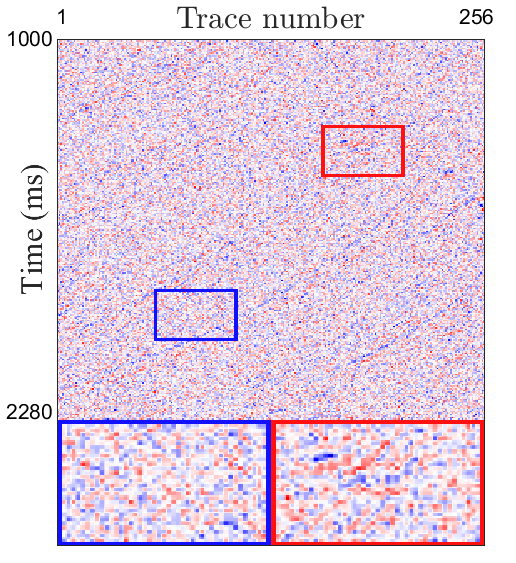}&
			\includegraphics[width=0.108\textwidth]{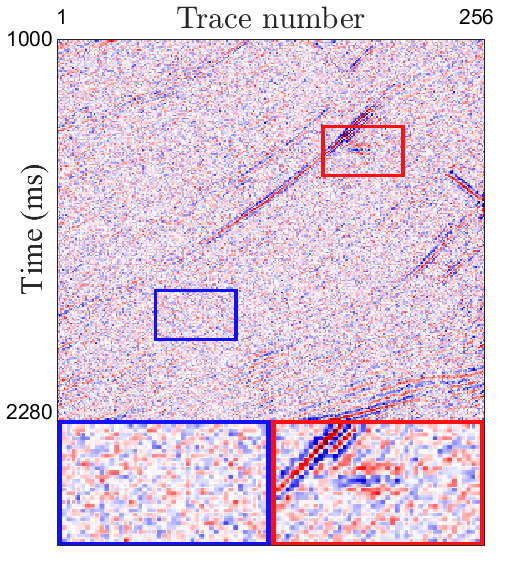}&
			\includegraphics[width=0.108\textwidth]{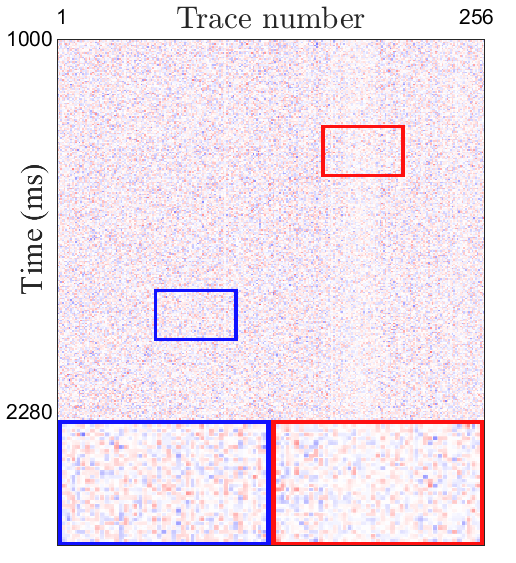}&
			\includegraphics[width=0.108\textwidth]{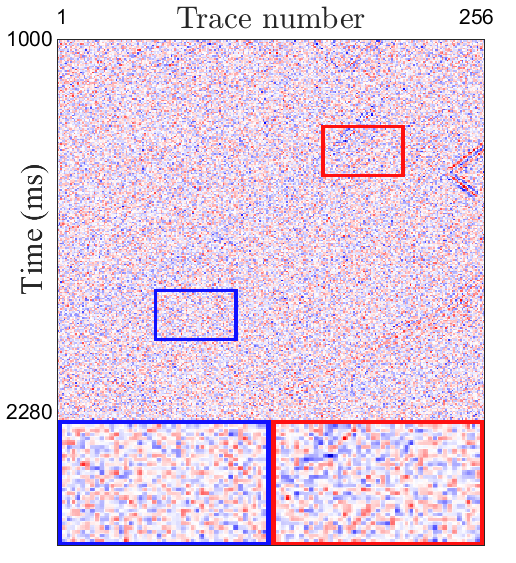}&
			\includegraphics[width=0.108\textwidth]{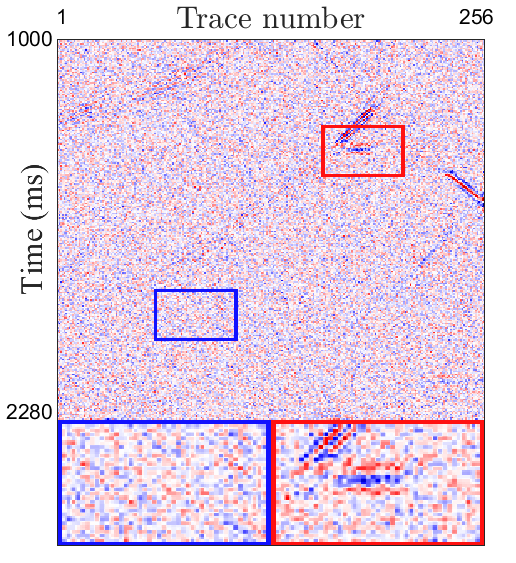}&
			\includegraphics[width=0.108\textwidth]{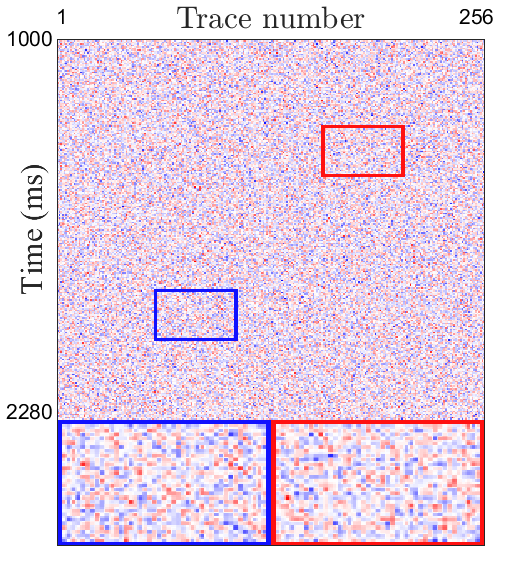}&
			\includegraphics[width=0.108\textwidth]{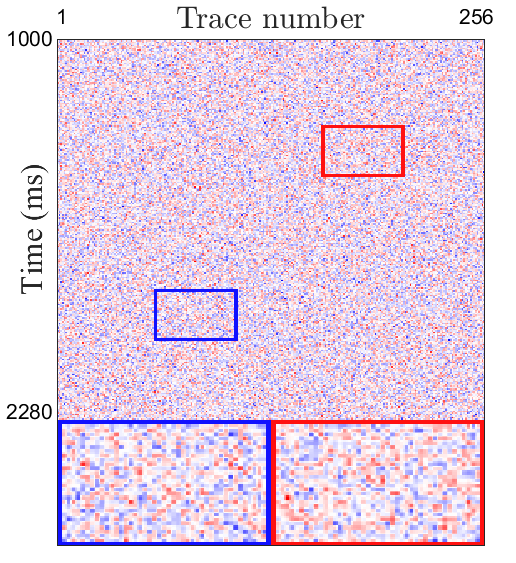}\\
			Noisy&BM3D\cite{BM3D}&WNNM\cite{WNNM}&MSSA\cite{MSSA}&DDAE\cite{DDAE}&DIP\cite{DIP}&PATCHUNET\cite{GP_PATCHUNET}&S2S-WTV&Original\\
			\vspace{-0.6cm}
		\end{tabular}
	\end{center}
	\caption{The noise attenuation results by different methods (the first, third, and fifth rows) and the corresponding residual maps between the noisy data and denoising results (the second, fourth, and sixth rows) on synthetic post-stack {\it Datasets (1)-(3)} with Gaussian noise ($\sigma=0.1$).\label{fig_results_1}}
	\vspace{-0.4cm}
\end{figure*}
\subsection{Training Loss Optimization}
\subsubsection{ADMM-Based Algorithm}\label{Sec_ADMM}
To minimize the self-supervised training loss (\ref{loss_WTV}), we suggest an ADMM-based algorithm. Specifically, by introducing auxiliary variables ${\bf V}_n$s ($n=1,2,\cdots,N$), model (\ref{loss_WTV}) can be re-formulated as
\begin{equation}
\begin{split}
&\min_{\theta,{\bf V}_ns}\;\sum_{n=1}^N \big{(}\left\lVert ({\bf Y}-f_\theta(\widehat{\bf Y}_n))\odot({\bf 1}-{\bf M}_n)\right\rVert_{F}^2\\
&\quad\quad\quad\quad
+\gamma\left\lVert{\bf W}\odot{\bf V}_n\right\rVert_{\ell_1}\big{)},\;{\rm s.t.}\;{\bf V}_n=\nabla_h{f_\theta(\widehat{\bf Y}_n)}.
\end{split}
\end{equation}
The corresponding augmented Lagrangian function is given as follows by attaching Lagrangian multipliers $\Lambda_n$s:
\begin{equation}
\begin{split}
&{\mathcal L}(\theta,{\bf V}_n,{\Lambda_n})\\=&\sum_{n=1}^N \big{(}\left\lVert ({\bf Y}-f_\theta(\widehat{\bf Y}_n))\odot({\bf 1}-{\bf M}_n)\right\rVert_{F}^2
+\gamma\left\lVert{\bf W}\odot{\bf V}_n\right\rVert_{\ell_1}\\&+\frac{\mu}{2}\left\lVert\nabla_h{f_\theta(\widehat{\bf Y}_n)}-{\bf V}_n\right\rVert_{F}^2+\langle\Lambda_n,\nabla_h{f_\theta(\widehat{\bf Y}_n)}-{\bf V}_n\rangle\big{)},
\end{split}
\end{equation}
where $\langle\cdot,\cdot\rangle$ denotes the inner product between two matrices and $\mu$ is the penalty parameter. The joint minimization problem can be decomposed into easier and smaller subproblems, followed by the update of Lagrangian multipliers.\par
{\bf ${\bf V}_n$ Subproblem} The ${\bf V}_n$ subproblem is
\begin{equation}
\min_{{\bf V}_ns}\sum_{n=1}^N(\frac{\mu}{2}\left\lVert\nabla_h{f_\theta(\widehat{\bf Y}_n)}+\frac{\Lambda_n}{\mu}-{\bf V}_n\right\rVert_{F}^2+\gamma\left\lVert{\bf W}\odot{\bf V}_n\right\rVert_{\ell_1}).
\end{equation}
In practice, directly tackling the large-scale problem with $N$ samples is computationally expensive. Thus, we only consider one individual sample (which is randomly selected from all samples) at each iteration. Suppose that the $m$-th sample is selected at the $(t+1)$-th iteration, then the corresponding ${\bf V}_m$ subproblem is
\begin{equation}\label{min_V}
\min_{{\bf V}_m}\frac{\mu}{2}\left\lVert\nabla_h{f_{\theta^t}(\widehat{\bf Y}_m)}+\frac{\Lambda_m^t}{\mu}-{\bf V}_m\right\rVert_{F}^2+\gamma\left\lVert{\bf W}\odot{\bf V}_m\right\rVert_{\ell_1}.
\end{equation}
\begin{lemma}\label{lemma_TV}
The exact solution of (\ref{min_V}) is 
\begin{equation}\label{V_solution}
({\bf V}_m^{t+1})_{(i,j)} = Soft_{\frac{\gamma{\bf W}_{(i,j)}}{\mu}}\big{(}(\nabla_h{f_{\theta^t}(\widehat{\bf Y}_m)}+\frac{\Lambda_m^t}{\mu})_{(i,j)}\big{)},
\end{equation}
where $Soft_{v}(\cdot):=sgn(\cdot)\max\{|\cdot|-v,0\}$ denotes the soft thresholding operator applied on each element of the input.
\end{lemma}
\begin{proof}
Note that (\ref{min_V}) can be equally formulated as the following element-wise optimization:
\begin{equation}\label{element_V}
\begin{split}
\min_{({\bf V}_m)_{(i,j)}}\frac{\mu}{2}\big{(}(\nabla_h{f_{\theta^t}(\widehat{\bf Y}_m)}+&\frac{\Lambda_m^t}{\mu})_{(i,j)}-({\bf V}_m)_{(i,j)}\big{)}^2\\&\;\;\;\;+\gamma{\bf W}_{(i,j)}|({\bf V}_m)_{(i,j)}|,
\end{split}
\end{equation}
where we have used the non-negativity of ${\bf W}$. It is known that the following statement is true for any $y\in{\mathbb R}$:
\begin{equation}\label{soft}
Soft_{\frac{\lambda}{2}}(y)=\arg\min_{x}(y-x)^2+\lambda|x|.
\end{equation}
Combining (\ref{element_V}) and (\ref{soft}) we see that (\ref{V_solution}) is the exact solution of (\ref{element_V}), and thus is the exact solution of (\ref{min_V}).
\end{proof}
After obtaining the optimal variable ${\bf V}_m^{t+1}$, we share it with other instances, i.e., we set ${\bf V}_n^{t+1}={\bf V}_m^{t+1}$ for all $n=1,2,\cdots,N$. The sharing strategy is reasonable because although we feed different instances into the CNN, the desired outputs are consistent, i.e., the clean signal.\par 
{\bf $\theta$ Subproblem} Similarly, we consider the $m$-th sample in the $(t+1)$-th iteration and formulate the $\theta$ subproblem as
\begin{equation}\label{loss_theta}
\begin{split}
&\min_\theta \left\lVert ({\bf Y}-f_\theta(\widehat{\bf Y}_m))\odot({\bf 1}-{\bf M}_m)\right\rVert_{F}^2\\&
\quad\quad\quad\quad\quad\quad
+\frac{\mu}{2}\left\lVert\nabla_h{f_\theta(\widehat{\bf Y}_m)}+\frac{\Lambda_m^t}{\mu}-{\bf V}_m^t\right\rVert_{F}^2.
\end{split}
\end{equation}
Due to the high non-convexity and nonlinearity of the above problem, which includes the update of CNN parameters $\theta$, we consider using the adaptive moment estimation (Adam) algorithm\cite{Adam}. Specifically, we employ one step of the Adam in each iteration of the ADMM-based algorithm to obtain the updated CNN parameters $\theta^{t+1}$ by using the loss (\ref{loss_theta}).\par
{\bf $\Lambda_n$ Update} The multipliers $\Lambda_n$s are updated by using the $m$-th instance as the CNN input:
\begin{equation}\label{lambda_n}
{\Lambda}_n^{t+1} = {\Lambda}_n^{t}+\mu(\nabla_h{f_{\theta^t}(\widehat{\bf Y}_m)}-{\bf V}_n^t).
\end{equation}\par
The overall flowchart of the ADMM-based algorithm for seismic data noise attenuation in a self-supervised manner is illustrated in Fig. \ref{fig_flow}. Moreover, after the update of variables in the ADMM-based algorithm, we further update the weight matrix $\bf W$ of the WTV regularization every 100 iterations through the following paradigm:
\begin{equation}\label{eq_W}
\begin{aligned}
\begin{split}
({\bf W}^{t+1})_{(i,j)} =\left \{
\begin{array}{lr}
    \vspace{0.1cm}
    \frac{\left\lVert{\bf Y}-f_{\theta^t}(\widehat{\bf Y}_m)\right\rVert^2_F}{2HW|(\nabla_h{f_{\theta^t}(\widehat{\bf Y}_m)})_{(i,j)}|},&t\;{\rm mod}\;100 = 0,\\
     \quad\quad\quad({\bf W}^{t})_{(i,j)},&t\;{\rm mod}\;100 \neq 0.
\end{array}
\right.
\end{split}
\end{aligned}
\end{equation}
Intuitively, we assign higher weights to smoother components and lower weights to less smooth components. The update of the weight matrix can help preserve the signal details and edges as much as possible. The weight matrix $\bf W$ is updated every 100 iterations and we fix the weight matrix after a certain iteration number (3000 in practice) due to the convergence.
\begin{table*}[!t]
	\caption{The average quantitative results by different methods for noise attenuation in synthetic pre-stack seismic {\it DATASETS (4)-(5)}. The {\bf BEST} and \underline{second-best} values are highlighted. (PSNR $\uparrow$, SSIM $\uparrow$, and LS $\downarrow$)\label{tab_denoising_2}}\vspace{-0.4cm}
	\begin{center}
		\scriptsize
		\setlength{\tabcolsep}{2.1pt}
		\begin{spacing}{1.2}
			\begin{tabular}{clccccccccccccccccccc}
				\toprule
				\multicolumn{2}{c}{Noise}&\multicolumn{3}{c}{Gaussian ($\sigma=0.1$)}&\multicolumn{3}{c}{Gaussian ($\sigma=0.2$)}&\multicolumn{3}{c}{Gaussian ($\sigma=0.3$)}&\multicolumn{3}{c}{Bandpass ($\sigma=0.1$)}&\multicolumn{3}{c}{Bandpass ($\sigma=0.2$)}&\multicolumn{3}{c}{Bandpass ($\sigma=0.3$)}&\multirow{3}*{\tabincell{c}{
						{Time}\\{(second)}}}\\
				\cmidrule{1-20}
				Data&Method&PSNR &SSIM &LS \;\; &PSNR &SSIM &LS \;\; &PSNR &SSIM &LS \;\; &PSNR &SSIM&LS \;\;&PSNR &SSIM&LS \;\;&PSNR &SSIM&LS&~ \\
				\midrule
				\multirow{8}*{\tabincell{c}{
						{\it Dataset (4)}\\{(192$\times$192)}}}
				&Observed&{20.01}&{0.757}&{\--\--}\;\;&{13.93}&{0.432}&{\--\--}\;\;&{10.42}&{0.256}&{\--\--}\;\;&{28.22}&{0.953}&{\--\--}\;\;&{22.04}&{0.831}&{\--\--}\;\;&{18.58}&{0.694}&{\--\--}\;\;&{\--\--}\\
				&BM3D&\underline{30.70}&\underline{0.971}&{0.091}\;\;&{28.35}&{0.936}&{0.113}\;\;&\underline{26.58}&\underline{0.910}&{0.173}\;\;&\underline{32.10}&\underline{0.979}&{0.136}\;\;&\underline{31.17}&\underline{0.974}&{0.163} \;\;&\underline{29.68}&\underline{0.963}&{0.224}\;\;&{1}\\
				&WNNM&{29.82}&{0.963}&{0.198}\;\;&\underline{29.47}&\underline{0.952}&{0.446}\;\;&{17.46}&{0.615}&{0.892}\;\;&{30.23}&{0.966}&{0.251}\;\;&{29.32}&{0.958}&{0.350}\;\;&{28.90}&{0.955}&{0.526} \;\;&{18}\\
				&MSSA&{21.20}&{0.793}&{0.166}\;\;&{17.21}&{0.404}&{0.192}\;\;&{15.23}&{0.294}&{0.195}\;\;&{28.53}&{0.956}&{0.179}\;\;&{22.41}&{0.837}&{0.216} \;\;&{19.16}&{0.707}&{0.226}\;\;&{1}\\
				&DDAE&{21.16}&{0.763}&{0.255}\;\;&{14.27}&{0.346}&{0.146}\;\;&{11.10}&{0.076}&{0.153}\;\;&{27.60}&{0.945}&{0.274}\;\;&{22.44}&{0.837}&{0.216} \;\;&{18.71}&{0.695}&{0.214}\;\;&{6}\\
				&DIP&{28.34}&{0.953}&\underline{0.072}\;\;&{20.90}&{0.773}&\underline{0.103}\;\;&{18.90}&{0.632}&\underline{0.100}\;\;&{31.96}&\underline{0.979}&\underline{0.121}\;\;&{26.38}&{0.901}&\underline{0.116} \;\;&{22.67}&{0.840}&\underline{0.146}\;\;&{182}\\
				&PATCHUNET&{25.65}&{0.890}&{0.180}\;\;&{23.98}&{0.837}&{0.184}\;\;&{22.49}&{0.780}&{0.209}\;\;&{25.84}&{0.896}&{0.245}\;\;&{24.56}&{0.861}&{0.280} \;\;&{23.48}&{0.829}&{0.293}\;\;&{94}\\
				&S2S-WTV&\bf{34.61}&\bf{0.989}&\bf{0.071}\;\;&\bf{30.10}&\bf{0.961}&\bf{0.113}\;\;&\bf{27.15}&\bf{0.935}&\bf{0.093}\;\;&\bf{37.57}&\bf{0.994}&\bf{0.092}\;\;&\bf{32.69}&\bf{0.982}&\bf{0.101} \;\;&\bf{30.09}&\bf{0.968}&\bf{0.114}\;\;&{107}\\
				\midrule
				\multirow{8}*{\tabincell{c}{
						{\it Dataset (5)}\\{(192$\times$192)}}}
				&Observed&{19.93}&{0.670}&{\--\--}\;\;&{13.93}&{0.337}&{\--\--}\;\;&{10.43}&{0.180}&{\--\--}\;\;&{28.00}&{0.928}&{\--\--}\;\;&{22.11}&{0.771}&{\--\--}\;\;&{18.56}&{0.599}&{\--\--}\;\;&{\--\--}\\
				&BM3D&{30.68}&{0.955}&\underline{0.087}\;\;&\underline{29.00}&\underline{0.924}&\underline{0.092}\;\;&{25.78}&{0.829}&\underline{0.099}\;\;&{32.74}&{0.972}&{0.126}\;\;&\underline{32.06}&\underline{0.967}&\underline{0.117}\;\;&{29.23}&{0.934}&\underline{0.072}\;\;&{1}\\
				&WNNM&\underline{31.92}&\underline{0.967}&{0.214}\;\;&{27.45}&{0.909}&{0.721}\;\;&{14.79}&{0.357}&{0.942}\;\;&\underline{33.16}&\underline{0.974}&{0.250}\;\;&{31.92}&{0.966}&{0.395}\;\;&\underline{29.96}&\underline{0.945}&{0.592} \;\;&{19}\\
				&MSSA&{24.49}&{0.833}&{0.171}\;\;&{19.47}&{0.603}&{0.182}\;\;&{16.26}&{0.404}&{0.189}\;\;&{28.11}&{0.920}&{0.178}\;\;&{25.58}&{0.864}&{0.215} \;\;&{23.25}&{0.784}&{0.222}\;\;&{1}\\
				&DDAE&{20.47}&{0.696}&{0.146}\;\;&{14.18}&{0.338}&{0.147}\;\;&{10.53}&{0.177}&{0.157}\;\;&{27.57}&{0.917}&{0.277}\;\;&{22.24}&{0.770}&{0.224}\;\;&{18.65}&{0.595}&{0.213} \;\;&{7}\\
				&DIP&{25.64}&{0.864}&{0.092}\;\;&{19.08}&{0.524}&{0.102}\;\;&{15.55}&{0.323}&{0.102}\;\;&{30.43}&{0.950}&\underline{0.116}\;\;&{25.60}&{0.874}&{0.139} \;\;&{22.74}&{0.770}&{0.125}\;\;&{182}\\
				&PATCHUNET&{31.75}&{0.966}&{0.186}\;\;&{28.27}&{0.920}&{0.332}\;\;&\underline{25.99}&\underline{0.839}&{0.392}\;\;&{32.93}&\underline{0.974}&{0.243}\;\;&{28.91}&{0.936}&{0.364} \;\;&{26.06}&{0.881}&{0.412}\;\;&{94}\\
				&S2S-WTV&\bf{33.84}&\bf{0.979}&\bf{0.057}\;\;&\bf{29.16}&\bf{0.935}&\bf{0.076}\;\;&\bf{27.64}&\bf{0.889}&\bf{0.053}\;\;&\bf{35.48}&\bf{0.985}&\bf{0.082}\;\;&\bf{33.00}&\bf{0.975}&\bf{0.114} \;\;&\bf{30.23}&\bf{0.952}&\bf{0.068}\;\;&{108}\\
				\bottomrule
			\end{tabular}
		\end{spacing}
	\end{center}
	\vspace{-0.4cm}
\end{table*}    
\begin{figure*}[!h]
	\scriptsize
	\setlength{\tabcolsep}{0.9pt}
	\begin{center}
		\begin{tabular}{ccccccccc}
			\includegraphics[width=0.108\textwidth]{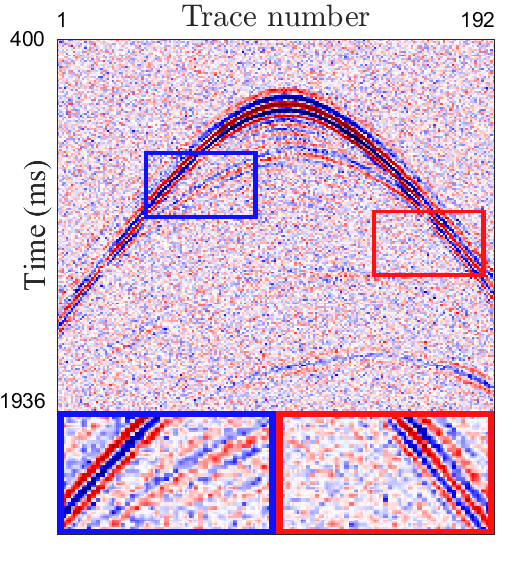}&
			\includegraphics[width=0.108\textwidth]{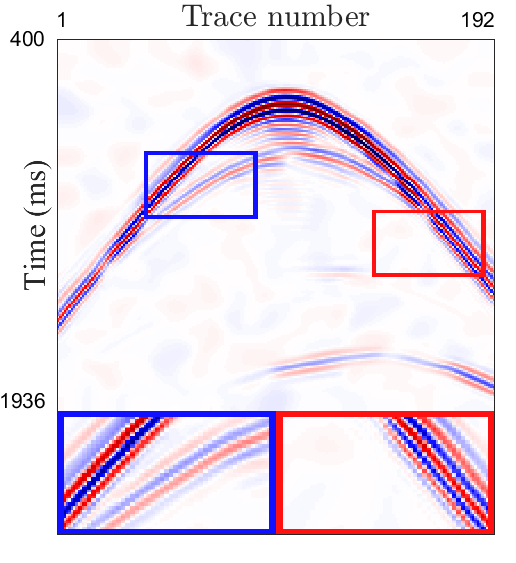}&
			\includegraphics[width=0.108\textwidth]{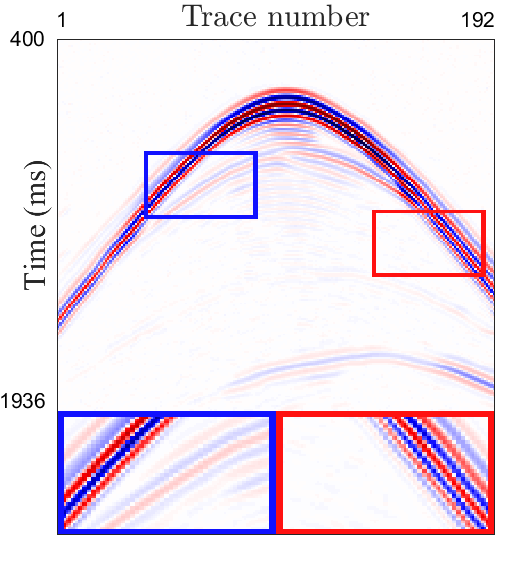}&
			\includegraphics[width=0.108\textwidth]{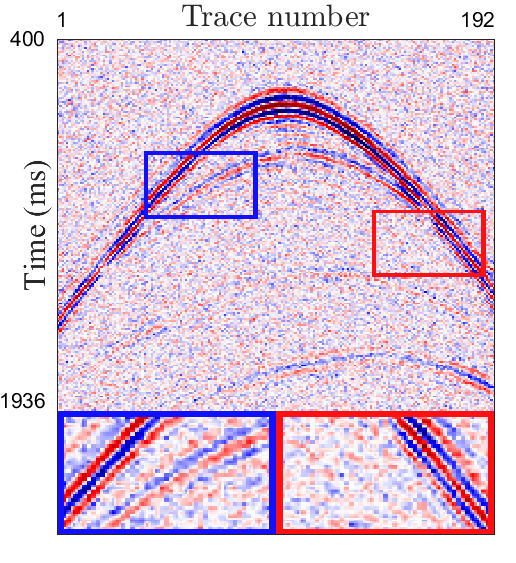}&
			\includegraphics[width=0.108\textwidth]{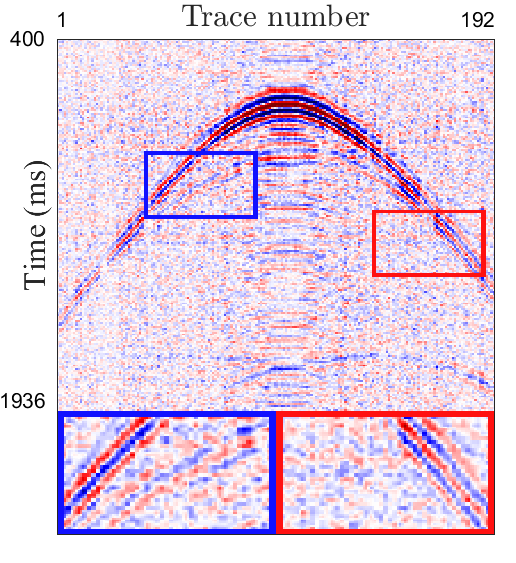}&
			\includegraphics[width=0.108\textwidth]{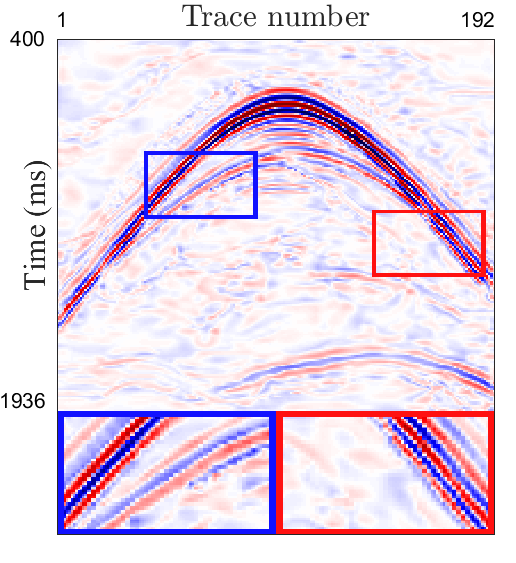}&
			\includegraphics[width=0.108\textwidth]{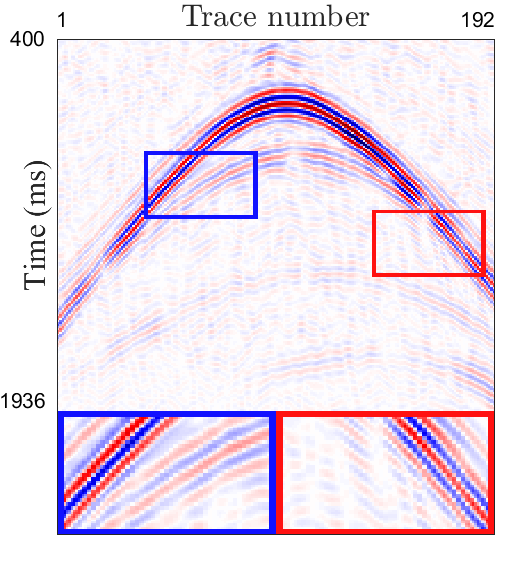}&
			\includegraphics[width=0.108\textwidth]{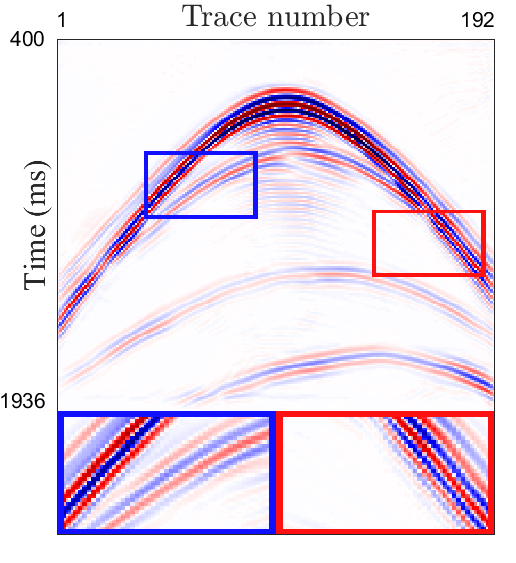}&
			\includegraphics[width=0.108\textwidth]{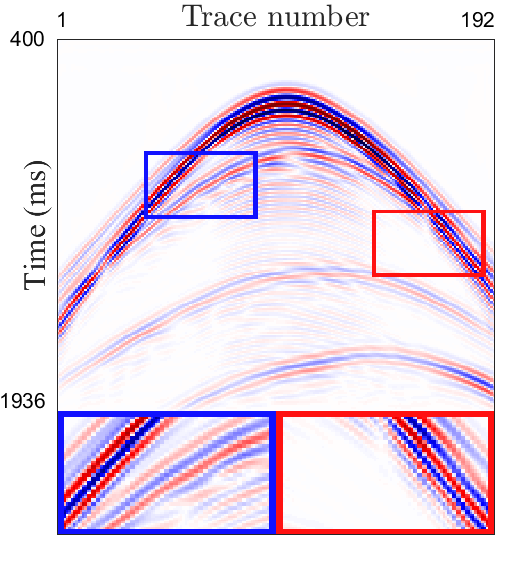}\\
			&
			\includegraphics[width=0.108\textwidth]{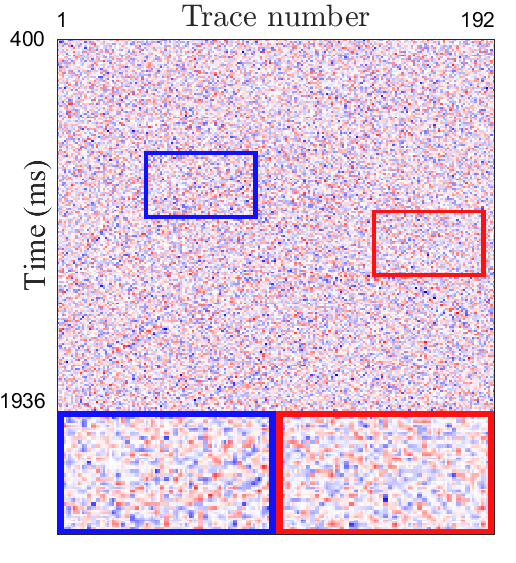}&
			\includegraphics[width=0.108\textwidth]{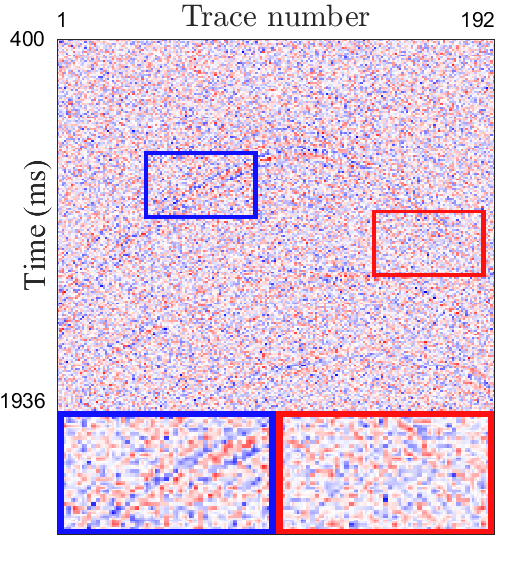}&
			\includegraphics[width=0.108\textwidth]{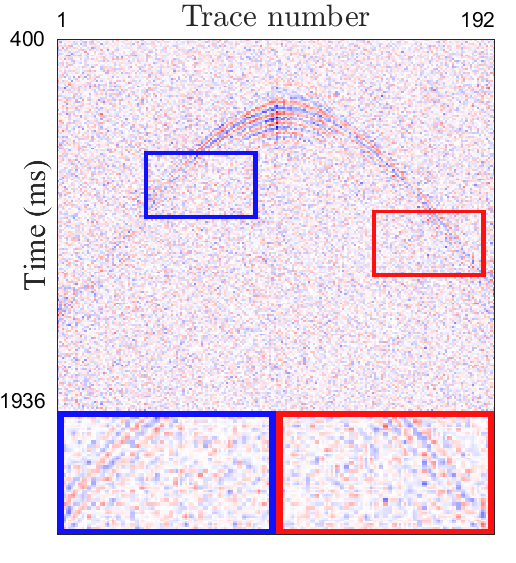}&
			\includegraphics[width=0.108\textwidth]{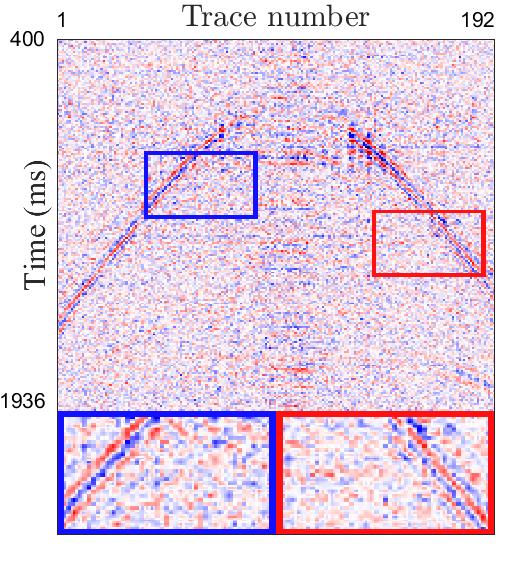}&
			\includegraphics[width=0.108\textwidth]{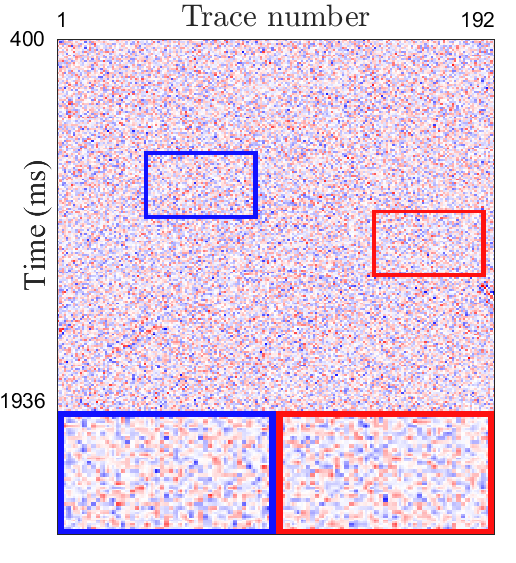}&
			\includegraphics[width=0.108\textwidth]{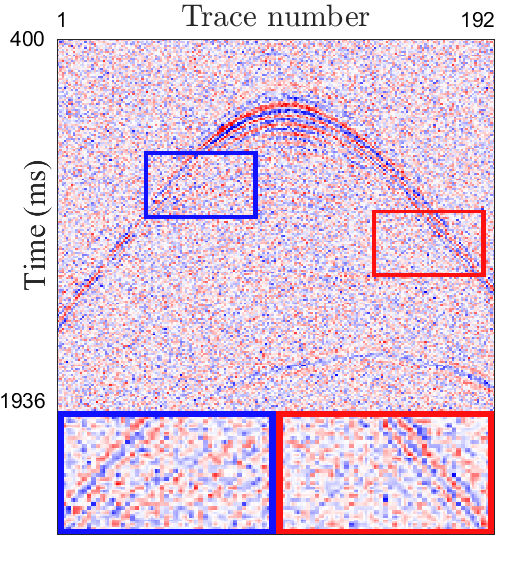}&
			\includegraphics[width=0.108\textwidth]{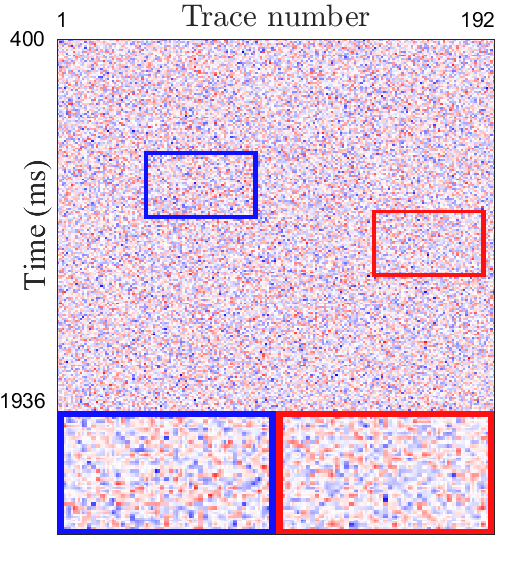}&
			\includegraphics[width=0.108\textwidth]{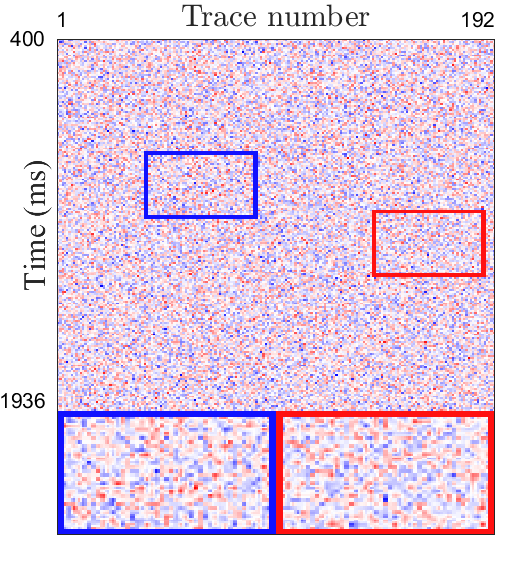}\\
			\includegraphics[width=0.108\textwidth]{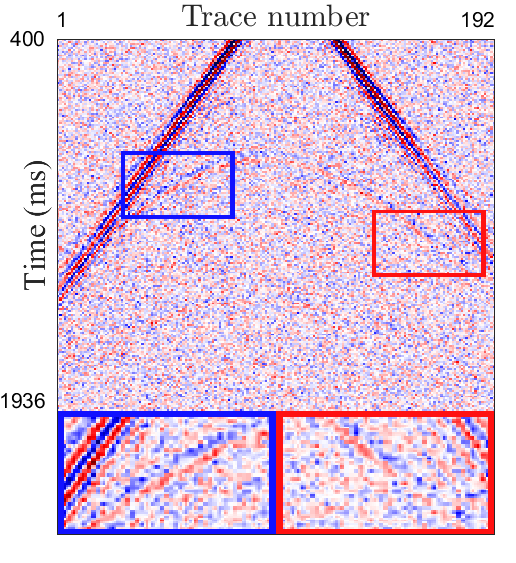}&
			\includegraphics[width=0.108\textwidth]{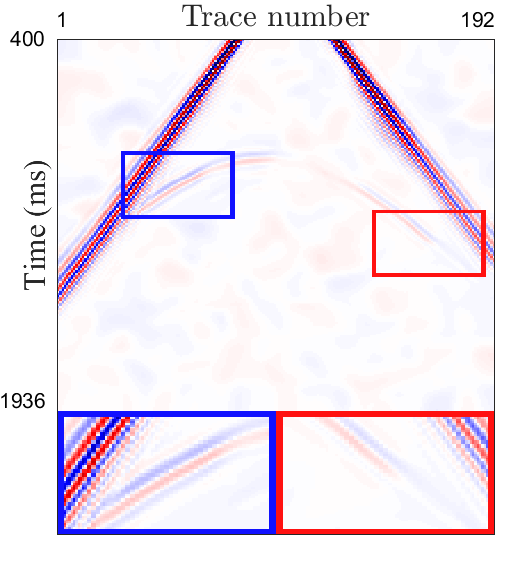}&
			\includegraphics[width=0.108\textwidth]{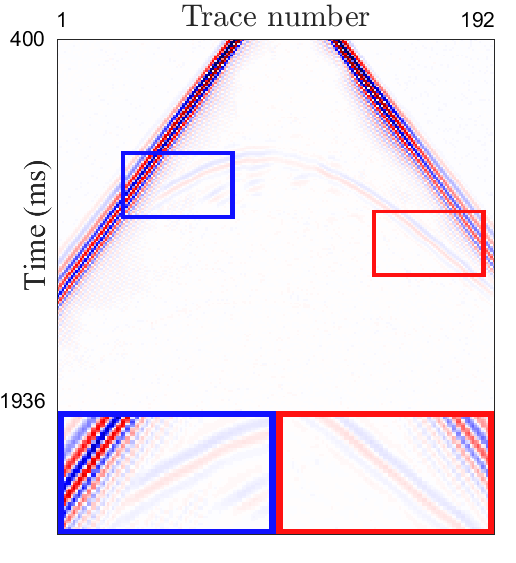}&
			\includegraphics[width=0.108\textwidth]{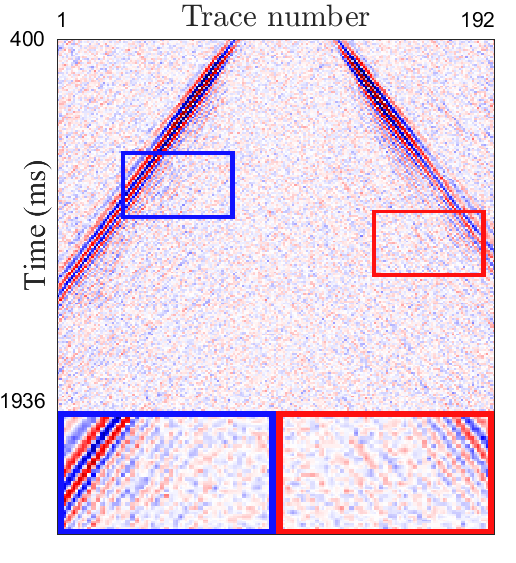}&
			\includegraphics[width=0.108\textwidth]{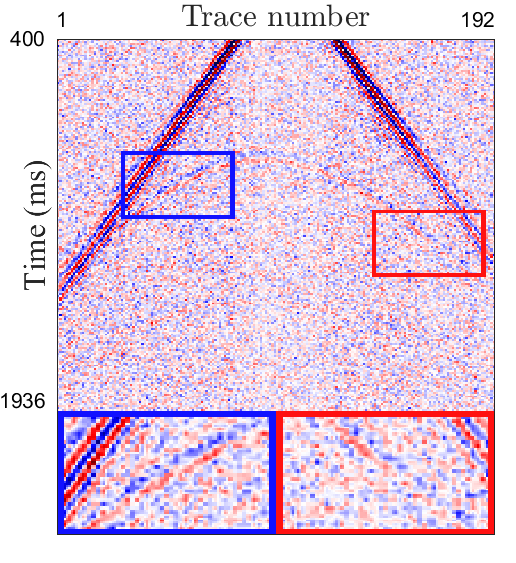}&
			\includegraphics[width=0.108\textwidth]{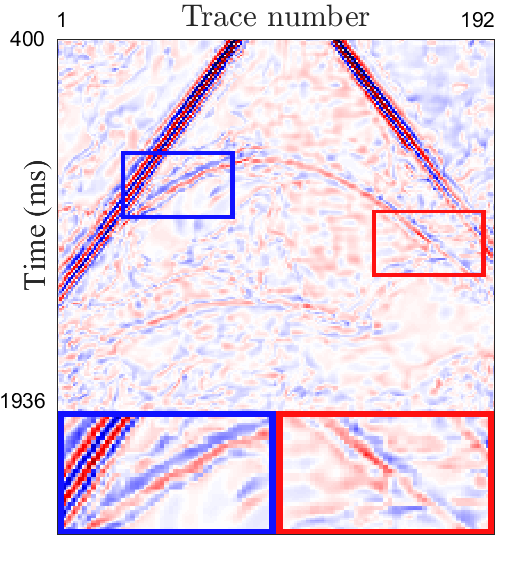}&
			\includegraphics[width=0.108\textwidth]{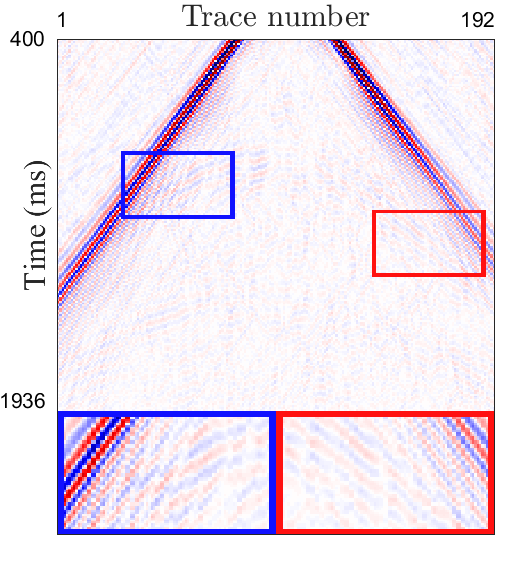}&
			\includegraphics[width=0.108\textwidth]{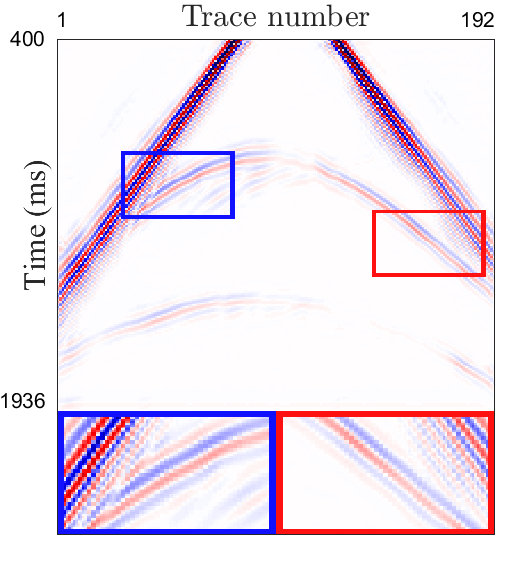}&
			\includegraphics[width=0.108\textwidth]{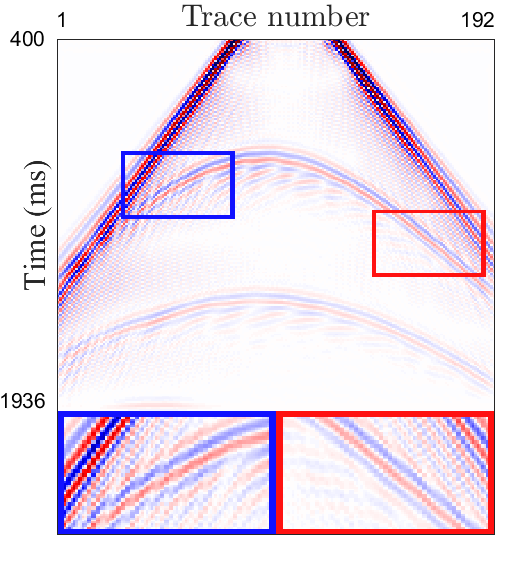}\\
			&
			\includegraphics[width=0.108\textwidth]{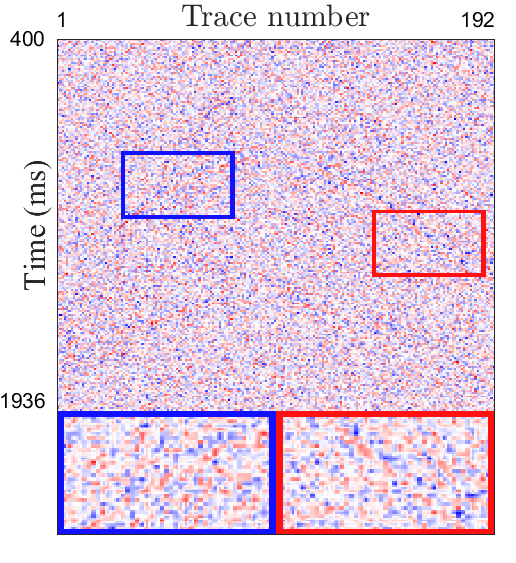}&
			\includegraphics[width=0.108\textwidth]{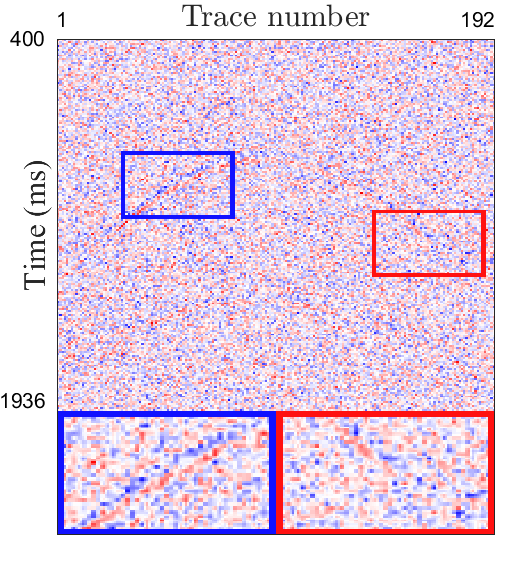}&
			\includegraphics[width=0.108\textwidth]{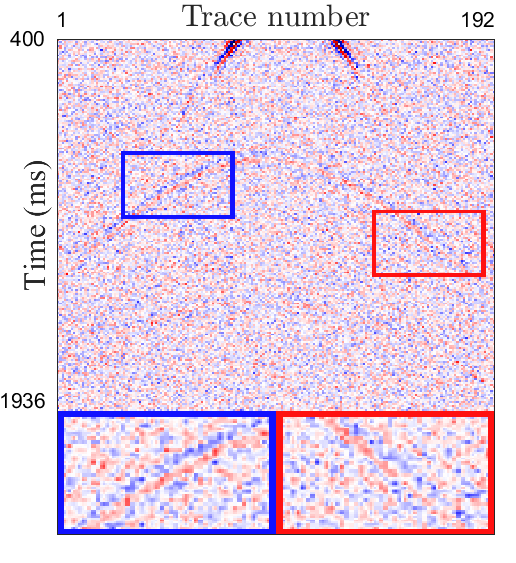}&
			\includegraphics[width=0.108\textwidth]{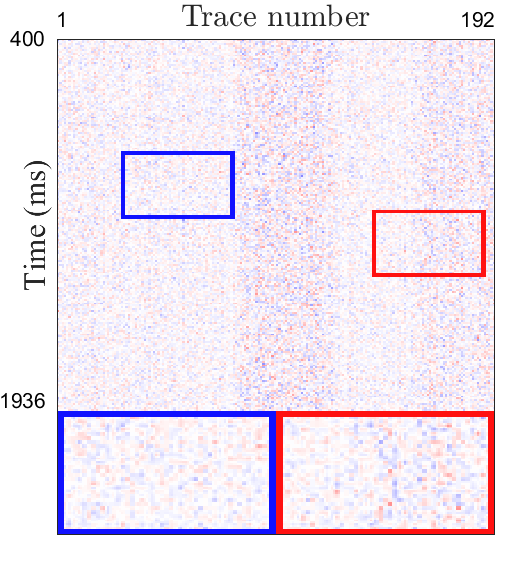}&
			\includegraphics[width=0.108\textwidth]{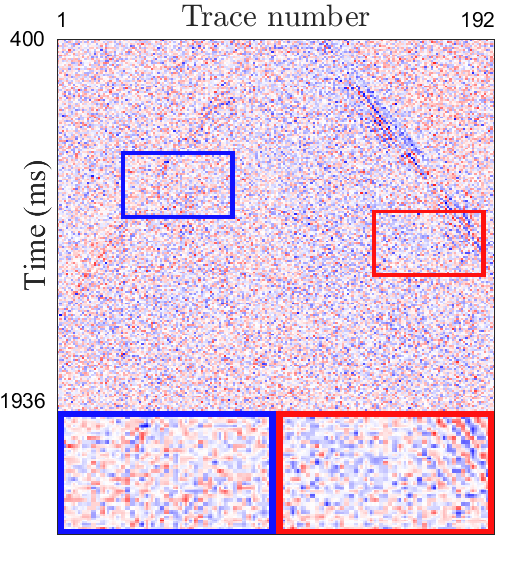}&
			\includegraphics[width=0.108\textwidth]{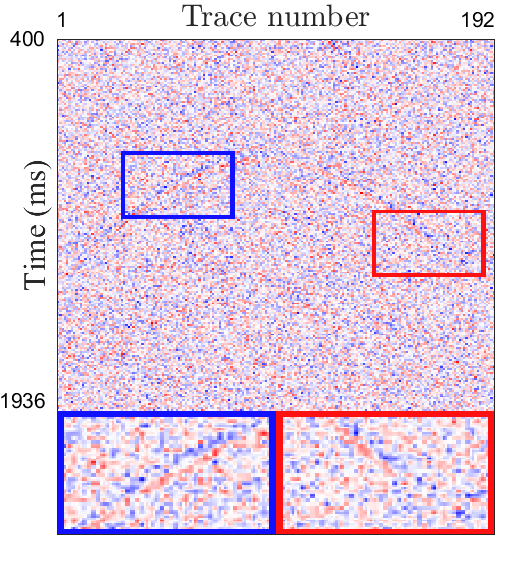}&
			\includegraphics[width=0.108\textwidth]{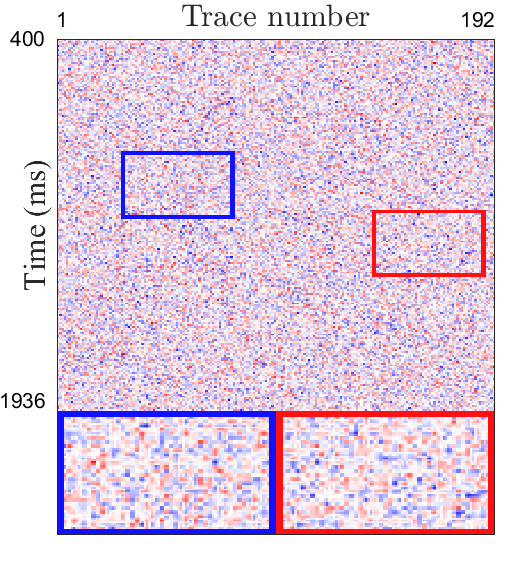}&
			\includegraphics[width=0.108\textwidth]{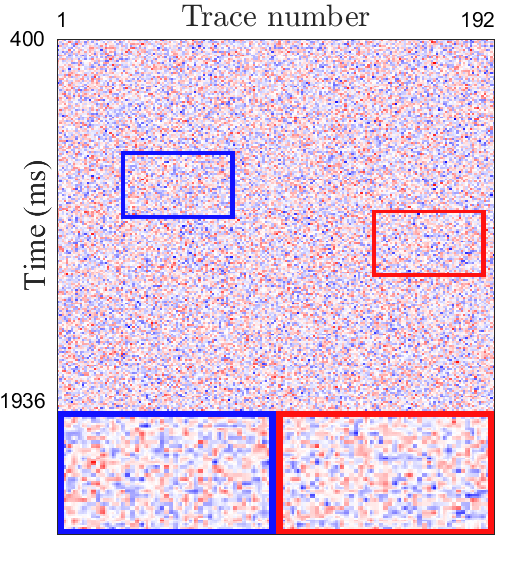}\\
			Noisy&BM3D\cite{BM3D}&WNNM\cite{WNNM}&MSSA\cite{MSSA}&DDAE\cite{DDAE}&DIP\cite{DIP}&PATCHUNET\cite{GP_PATCHUNET}&S2S-WTV&Original\\
			\vspace{-0.6cm}
		\end{tabular}
	\end{center}
	\caption{The noise attenuation results by different methods (the first and third rows) and the corresponding residual maps between the noisy data and denoising results (the second and fourth rows) on synthetic pre-stack {\it Datasets (4)-(5)} with Gaussian noise ($\sigma=0.1$).\label{fig_results_2}}
	\vspace{-0.5cm}
\end{figure*}
\begin{figure*}[t]
	\scriptsize
	\setlength{\tabcolsep}{0.9pt}
	\begin{center}
		\begin{tabular}{ccccccccc}
			\includegraphics[width=0.108\textwidth]{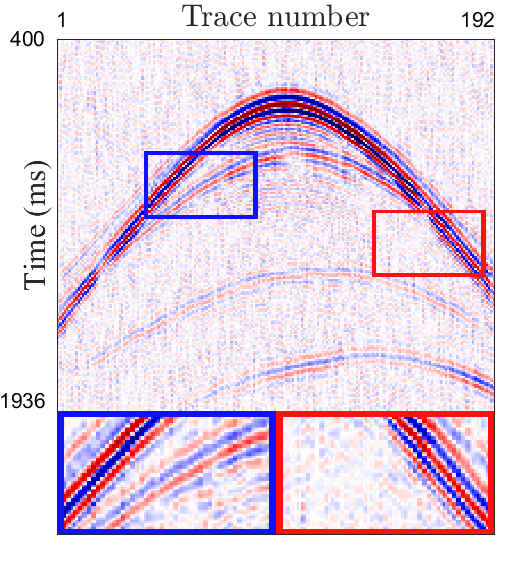}&
			\includegraphics[width=0.108\textwidth]{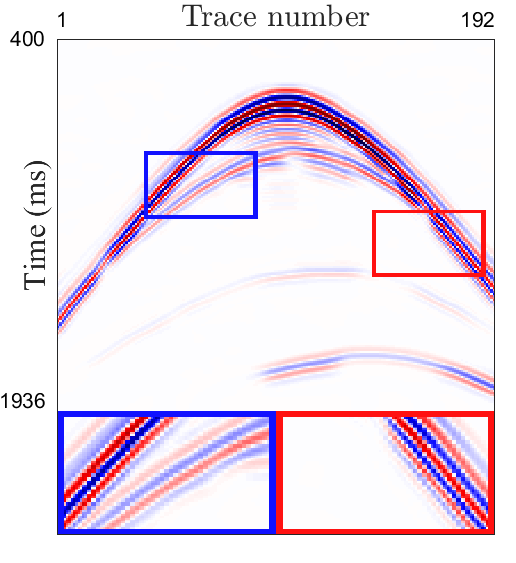}&
			\includegraphics[width=0.108\textwidth]{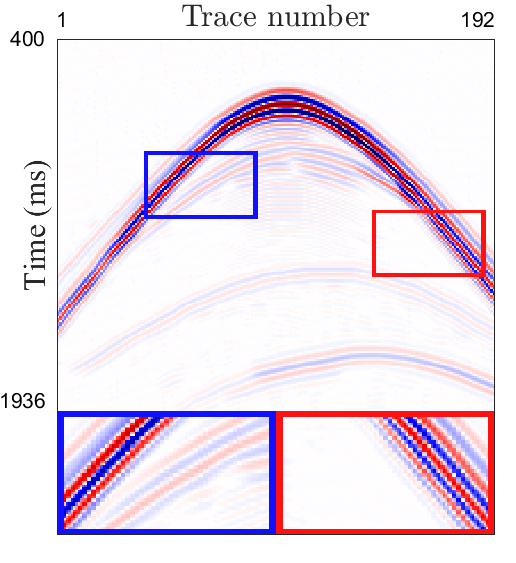}&
			\includegraphics[width=0.108\textwidth]{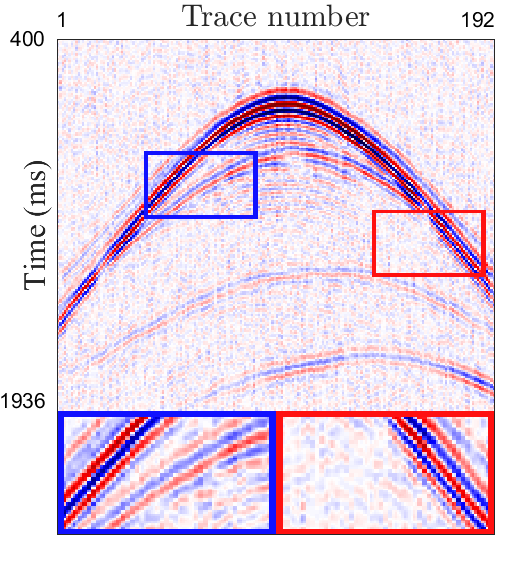}&
			\includegraphics[width=0.108\textwidth]{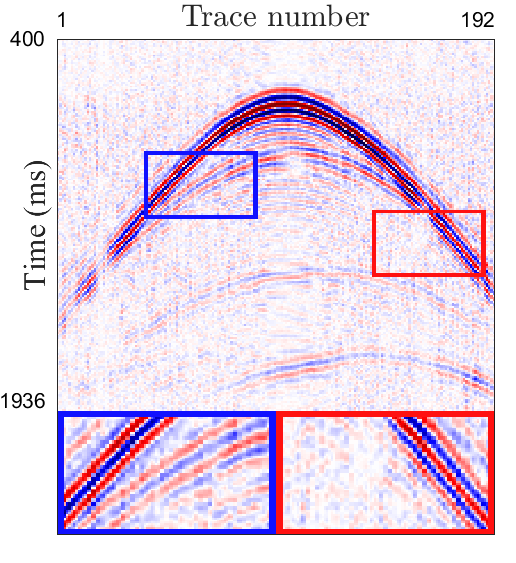}&
			\includegraphics[width=0.108\textwidth]{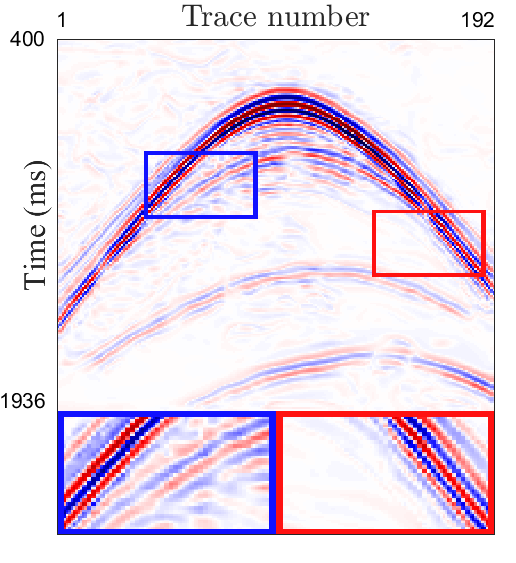}&
			\includegraphics[width=0.108\textwidth]{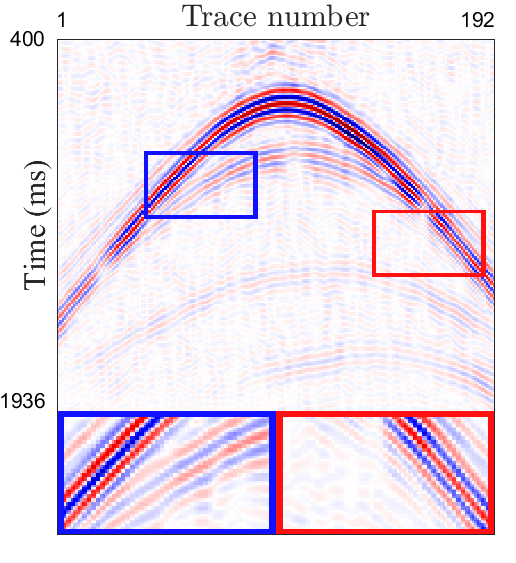}&
			\includegraphics[width=0.108\textwidth]{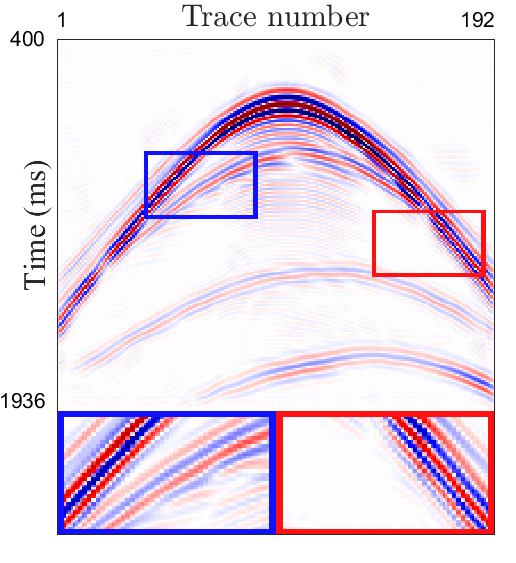}&
			\includegraphics[width=0.108\textwidth]{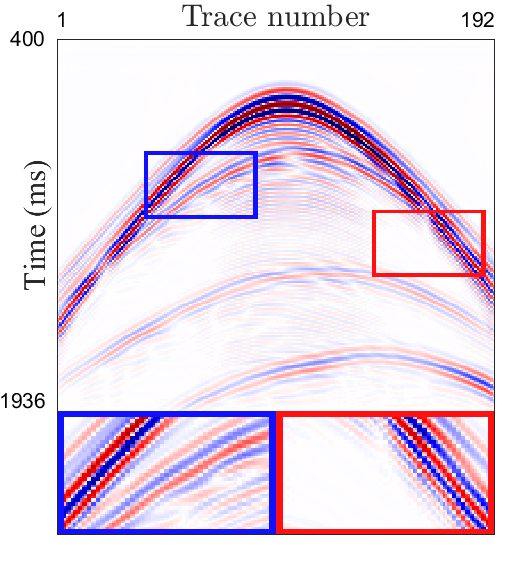}\\
			&
			\includegraphics[width=0.108\textwidth]{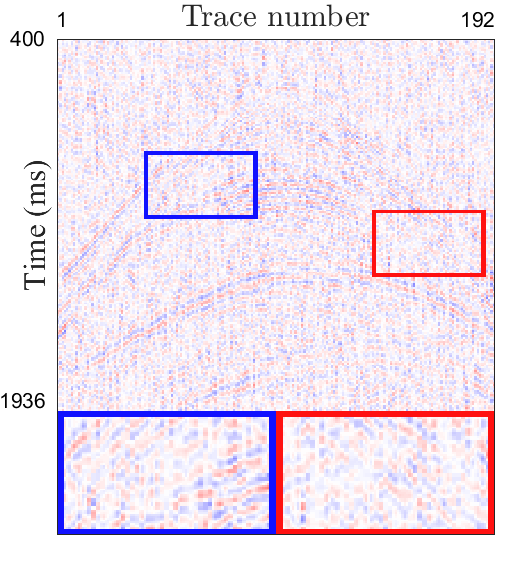}&
			\includegraphics[width=0.108\textwidth]{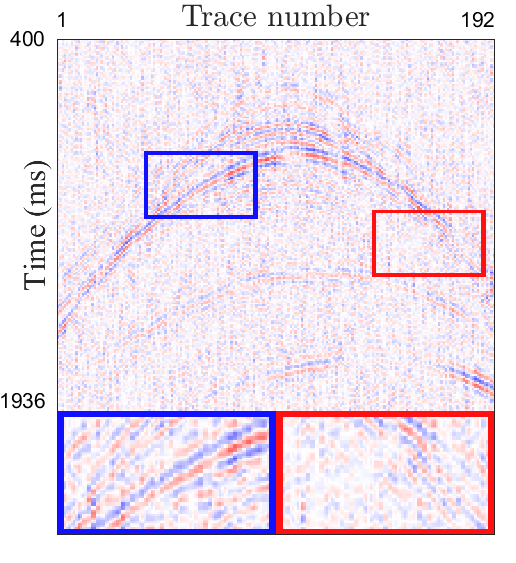}&
			\includegraphics[width=0.108\textwidth]{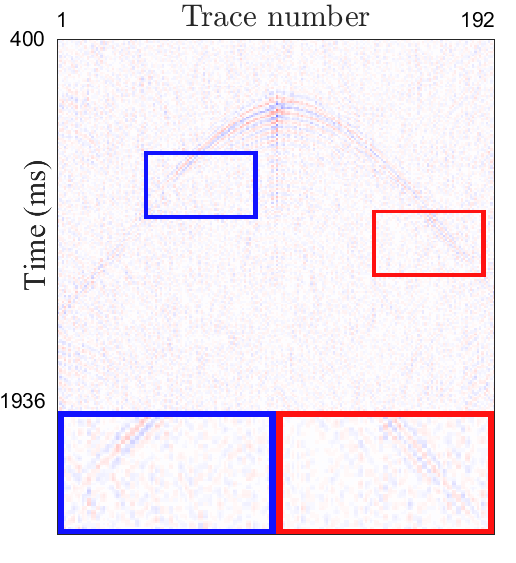}&
			\includegraphics[width=0.108\textwidth]{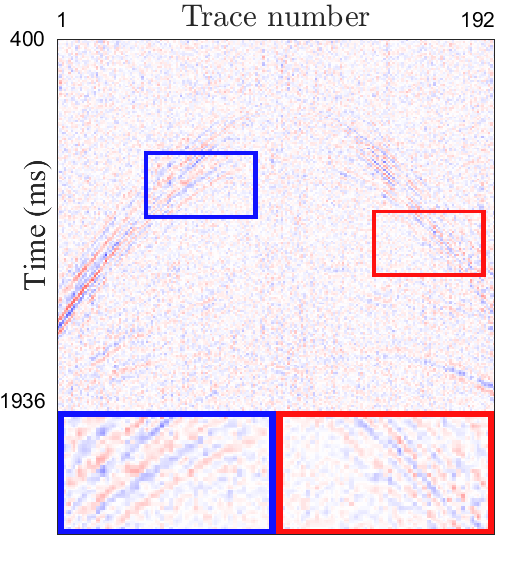}&
			\includegraphics[width=0.108\textwidth]{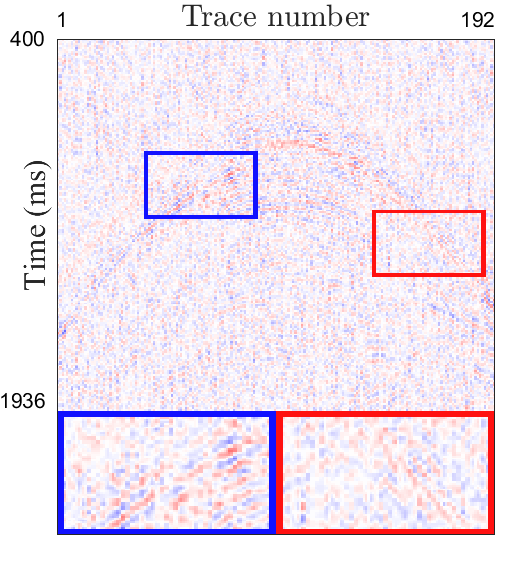}&
			\includegraphics[width=0.108\textwidth]{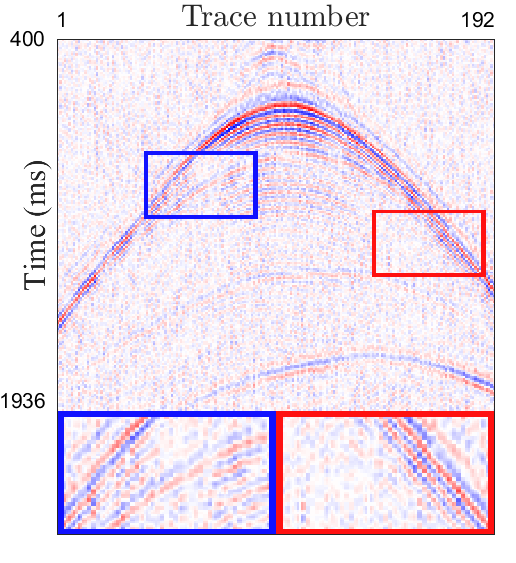}&
			\includegraphics[width=0.108\textwidth]{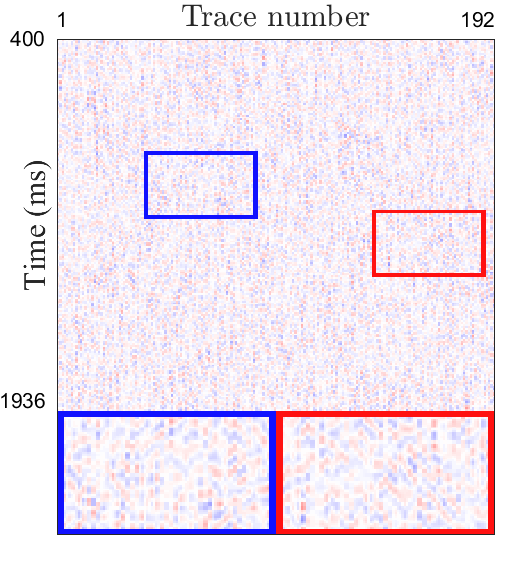}&
			\includegraphics[width=0.108\textwidth]{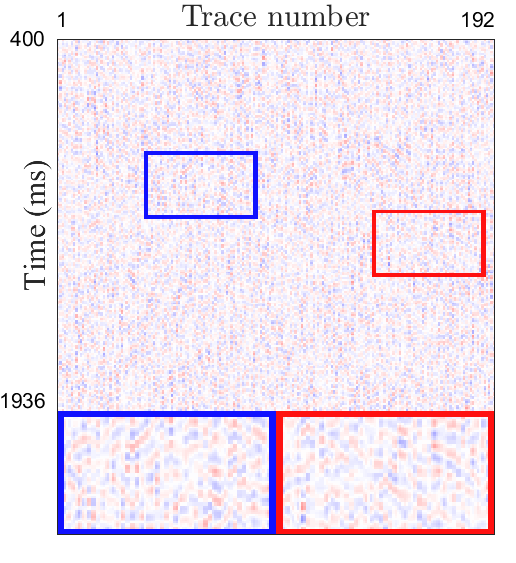}\\
			\includegraphics[width=0.108\textwidth]{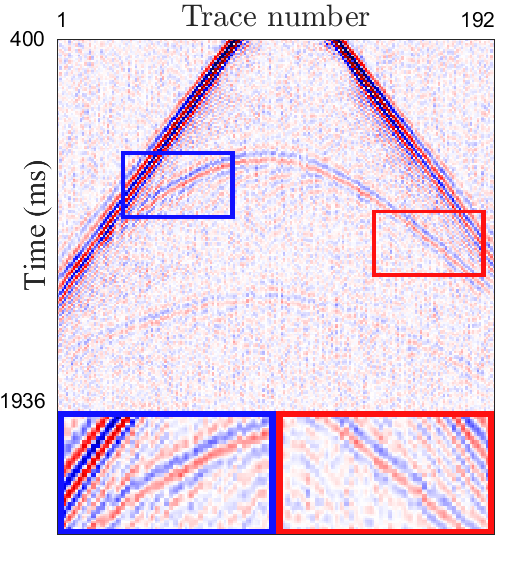}&
			\includegraphics[width=0.108\textwidth]{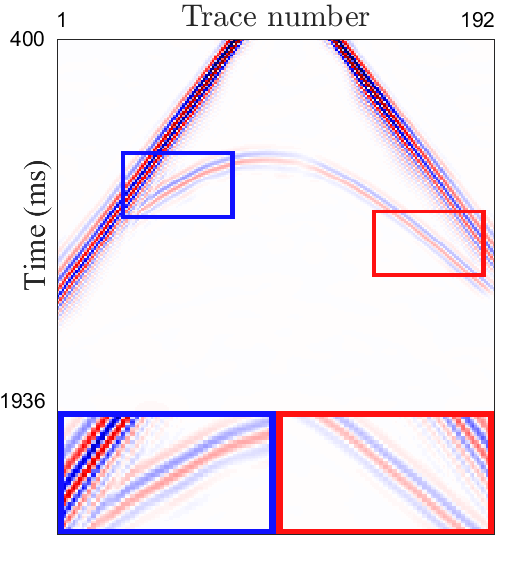}&
			\includegraphics[width=0.108\textwidth]{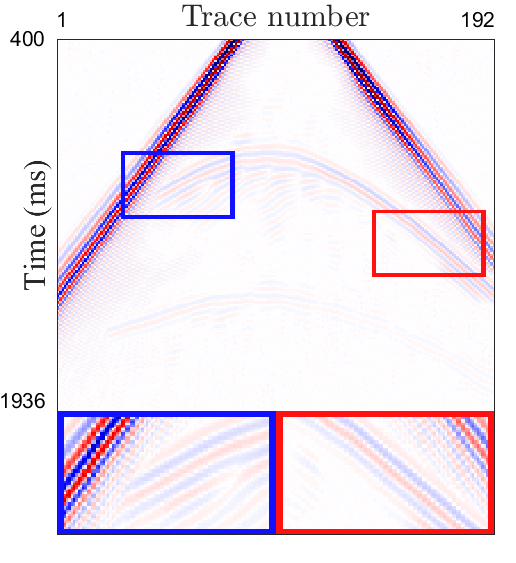}&
			\includegraphics[width=0.108\textwidth]{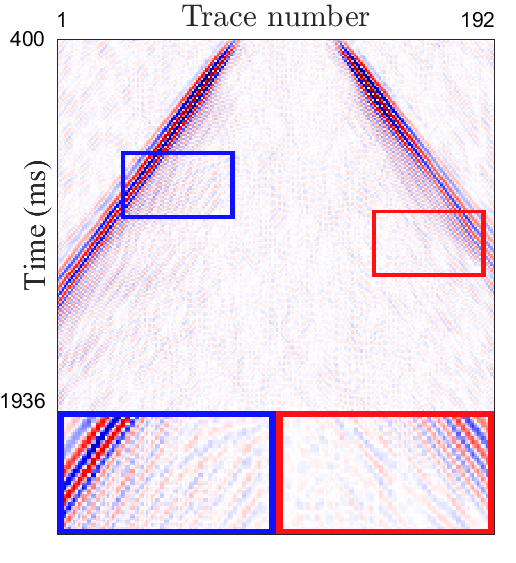}&
			\includegraphics[width=0.108\textwidth]{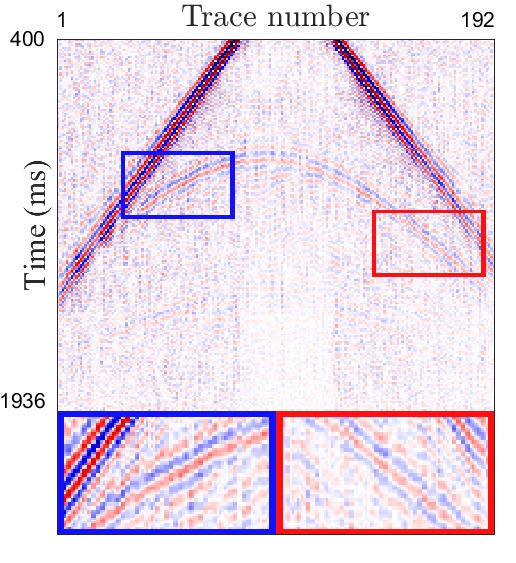}&
			\includegraphics[width=0.108\textwidth]{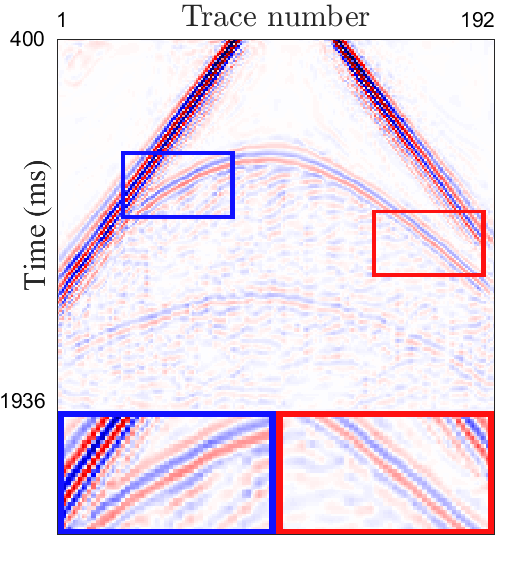}&
			\includegraphics[width=0.108\textwidth]{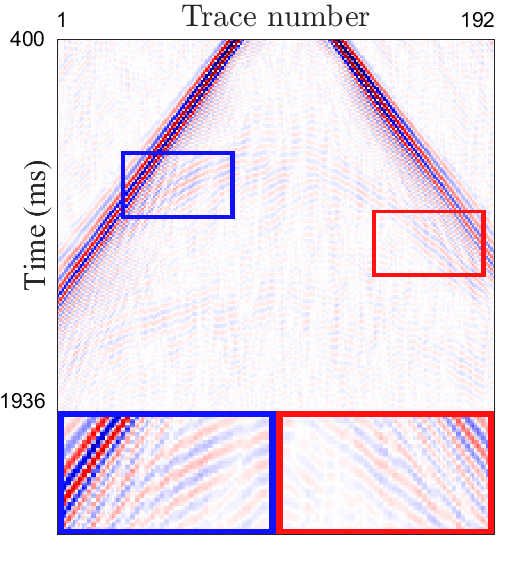}&
			\includegraphics[width=0.108\textwidth]{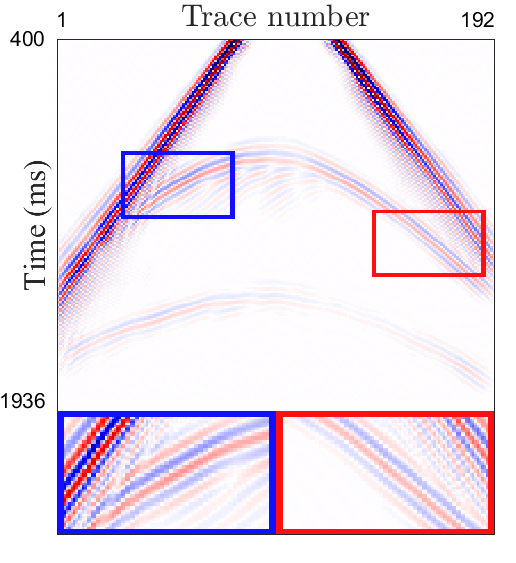}&
			\includegraphics[width=0.108\textwidth]{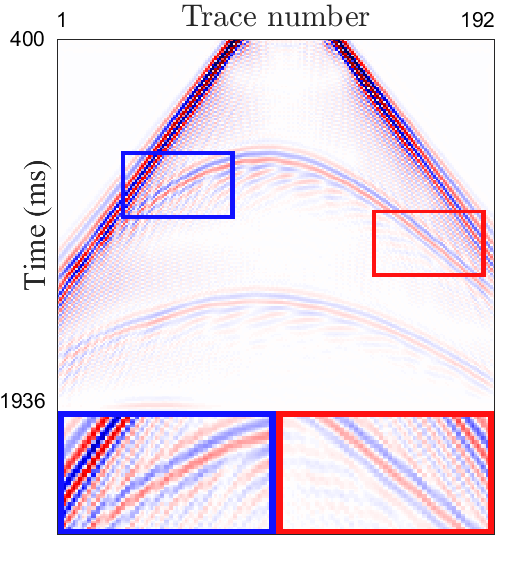}\\
			&
			\includegraphics[width=0.108\textwidth]{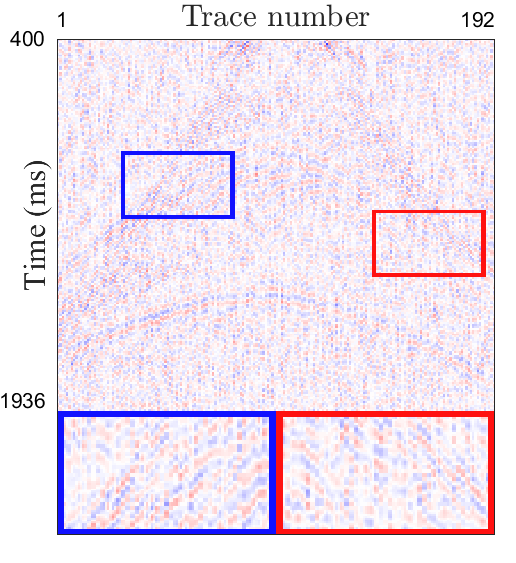}&
			\includegraphics[width=0.108\textwidth]{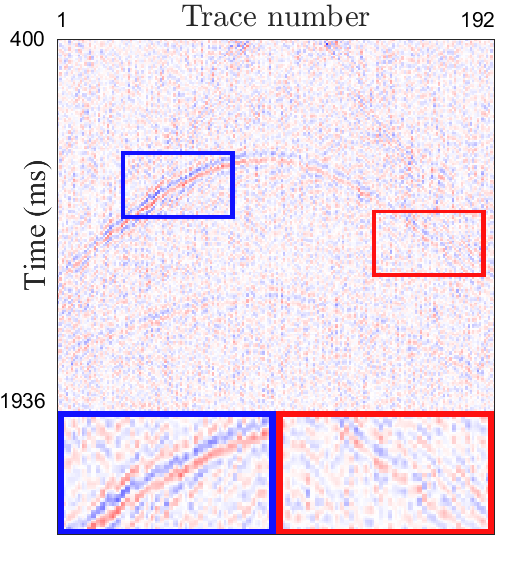}&
			\includegraphics[width=0.108\textwidth]{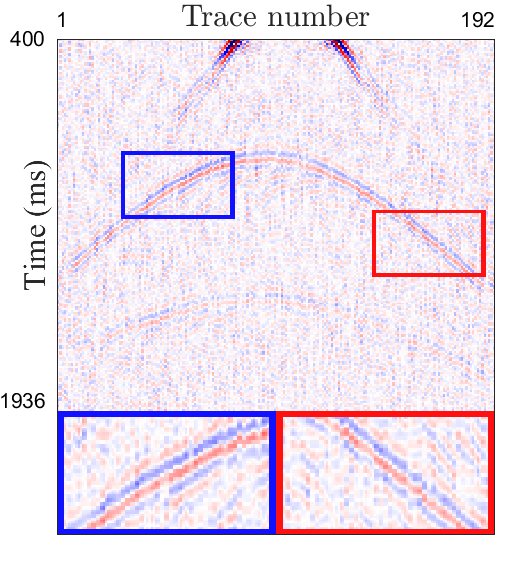}&
			\includegraphics[width=0.108\textwidth]{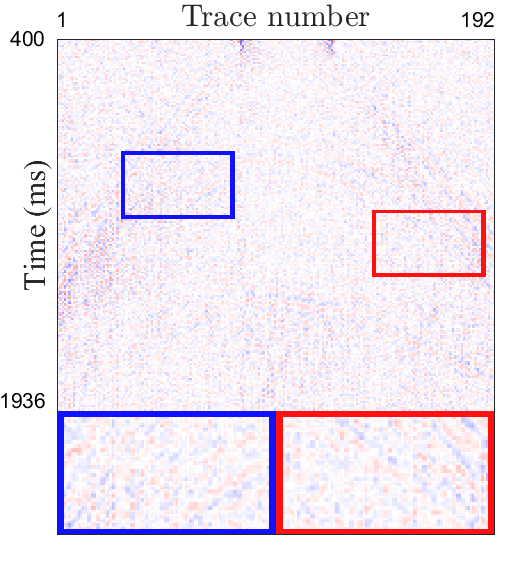}&
			\includegraphics[width=0.108\textwidth]{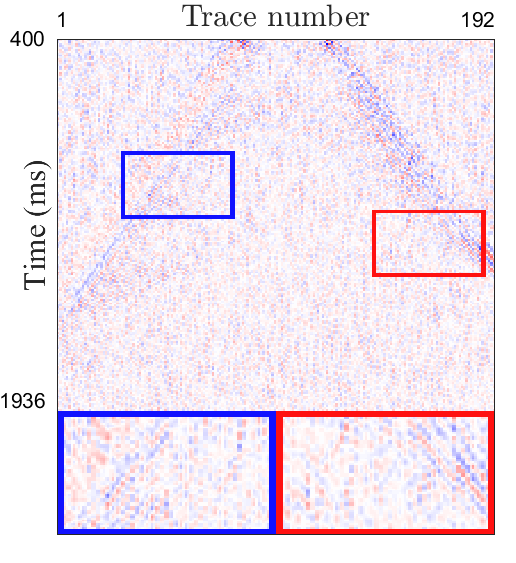}&
			\includegraphics[width=0.108\textwidth]{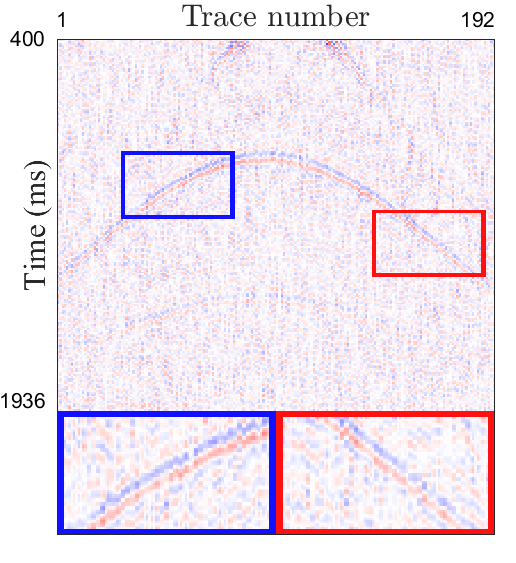}&
			\includegraphics[width=0.108\textwidth]{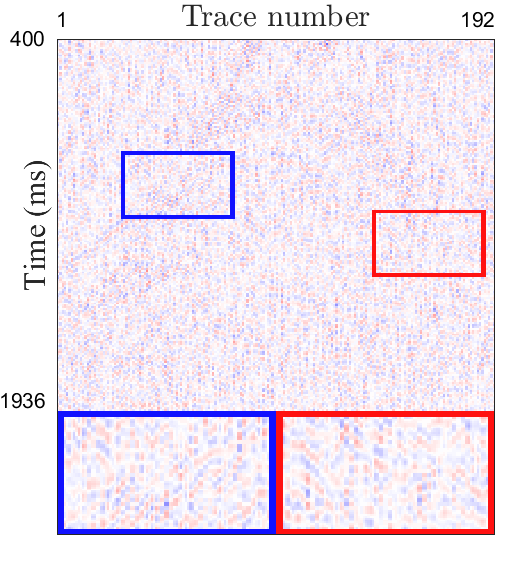}&
			\includegraphics[width=0.108\textwidth]{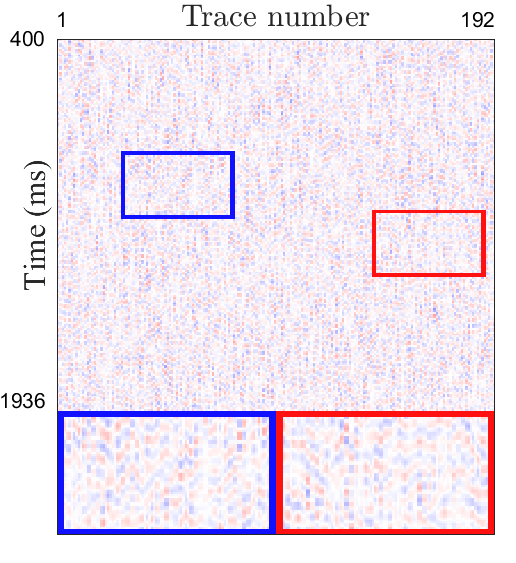}\\
			Noisy&BM3D\cite{BM3D}&WNNM\cite{WNNM}&MSSA\cite{MSSA}&DDAE\cite{DDAE}&DIP\cite{DIP}&PATCHUNET\cite{GP_PATCHUNET}&S2S-WTV&Original\\
			\vspace{-0.6cm}
		\end{tabular}
	\end{center}
	\caption{The noise attenuation results by different methods (the first and third rows) and the corresponding residual maps between the noisy data and denoising results (the second and fourth rows) on synthetic pre-stack {\it Datasets (4)-(5)} with bandpass noise ($\sigma=0.1$).\label{fig_results_3}}
	\vspace{-0.3cm}
\end{figure*}
\begin{figure*}[t]
	\scriptsize
	\setlength{\tabcolsep}{0.9pt}
	\begin{center}
		\begin{tabular}{cccccccc}
			\includegraphics[width=0.108\textwidth]{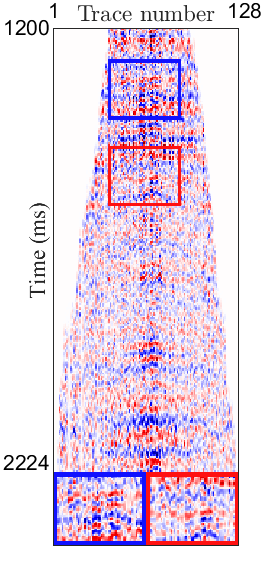}&
			\includegraphics[width=0.108\textwidth]{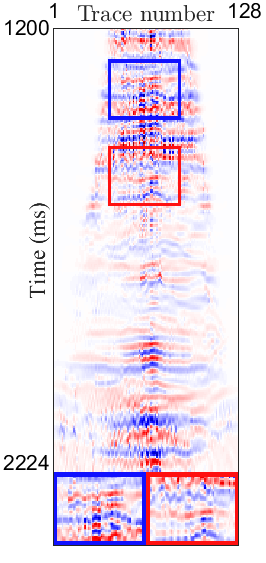}&
			\includegraphics[width=0.108\textwidth]{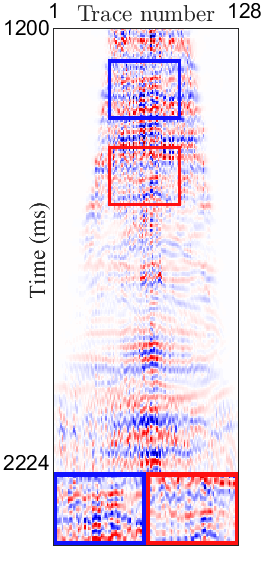}&
			\includegraphics[width=0.108\textwidth]{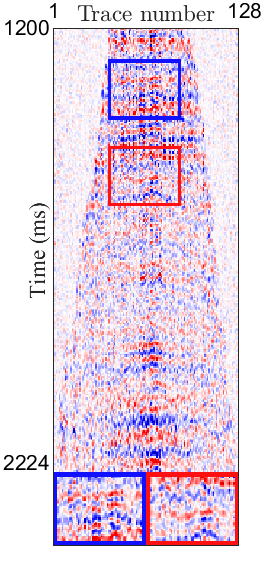}&
			\includegraphics[width=0.108\textwidth]{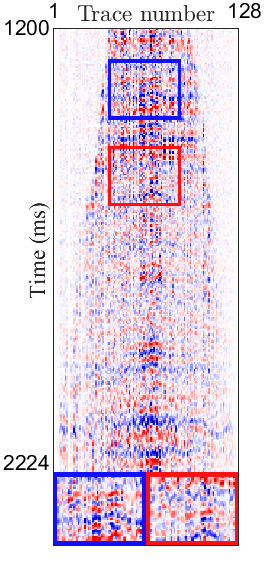}&
			\includegraphics[width=0.108\textwidth]{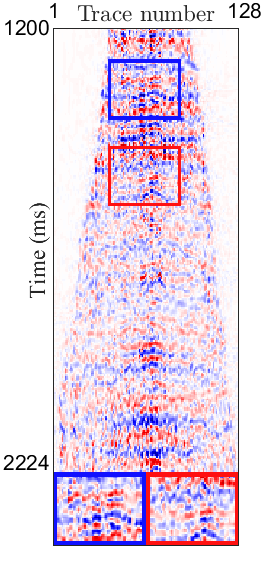}&
			\includegraphics[width=0.108\textwidth]{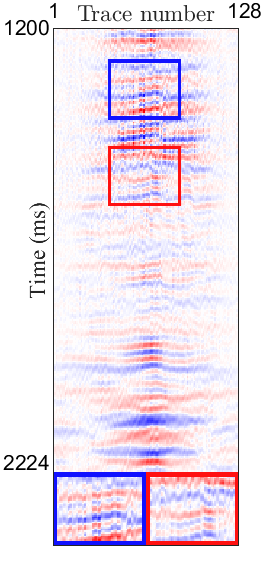}&
			\includegraphics[width=0.108\textwidth]{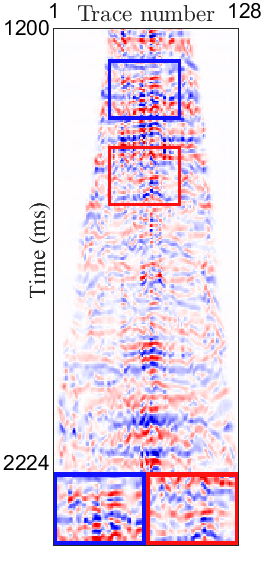}\\
			&
			\includegraphics[width=0.108\textwidth]{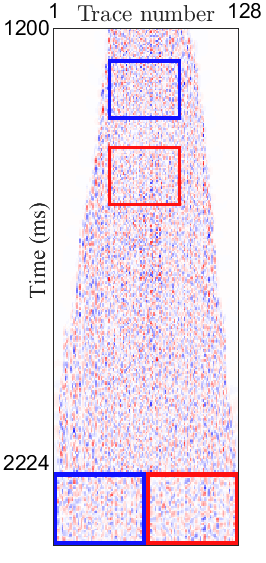}&
			\includegraphics[width=0.108\textwidth]{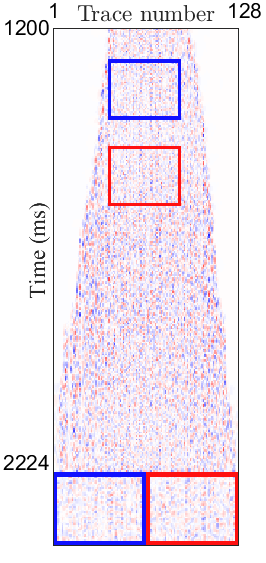}&
			\includegraphics[width=0.108\textwidth]{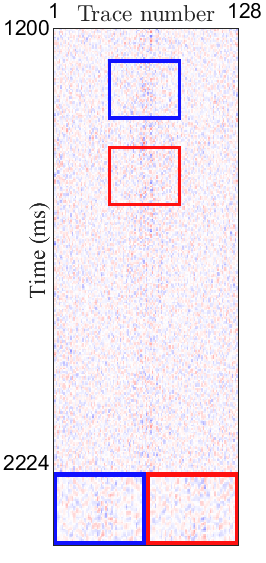}&
			\includegraphics[width=0.108\textwidth]{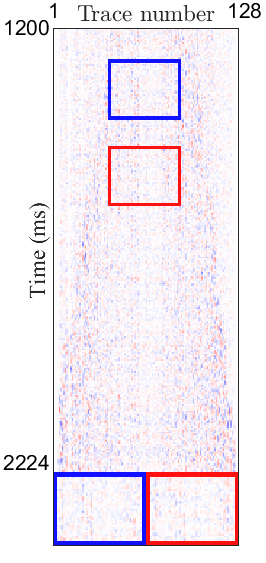}&
			\includegraphics[width=0.108\textwidth]{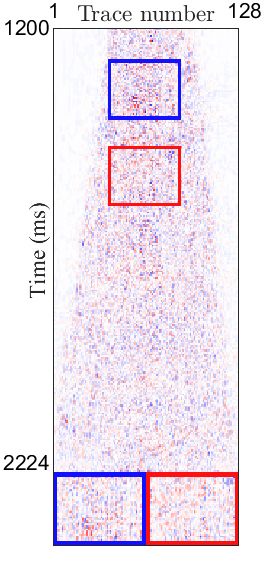}&
			\includegraphics[width=0.108\textwidth]{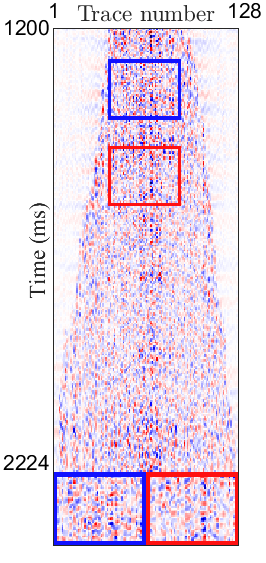}&
			\includegraphics[width=0.108\textwidth]{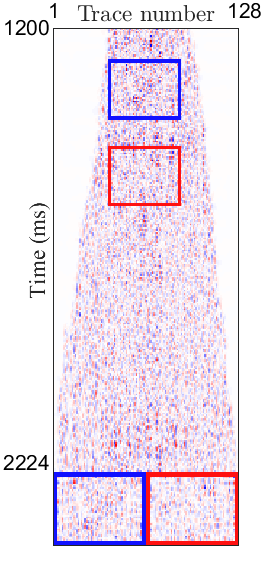}\\
			\includegraphics[width=0.108\textwidth]{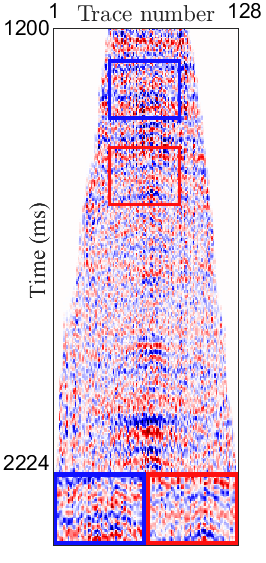}&
			\includegraphics[width=0.108\textwidth]{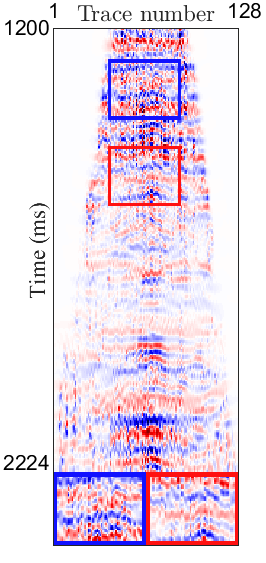}&
			\includegraphics[width=0.108\textwidth]{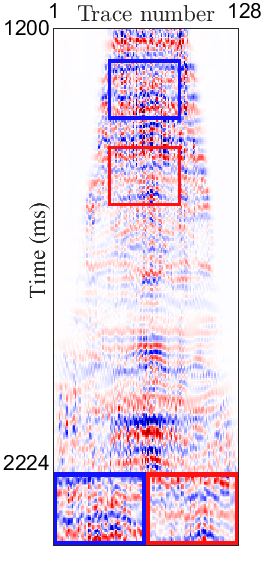}&
			\includegraphics[width=0.108\textwidth]{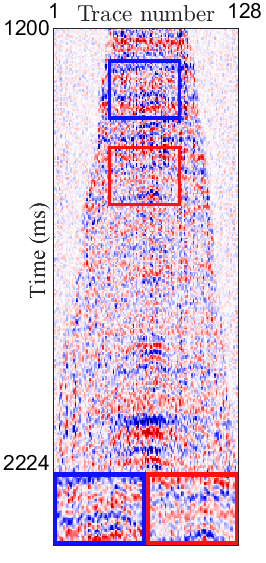}&
			\includegraphics[width=0.108\textwidth]{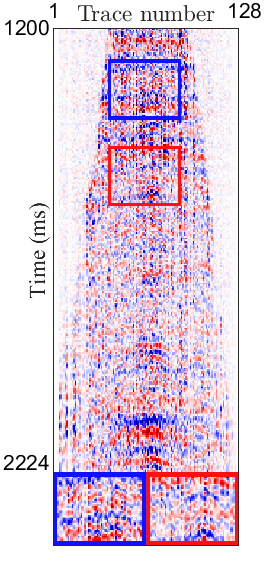}&
			\includegraphics[width=0.108\textwidth]{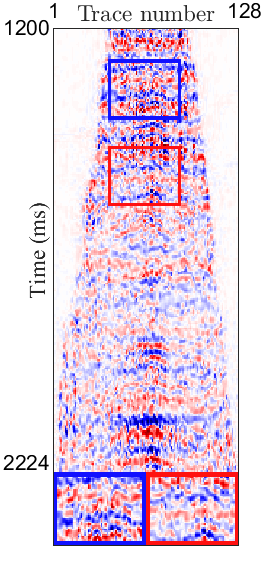}&
			\includegraphics[width=0.108\textwidth]{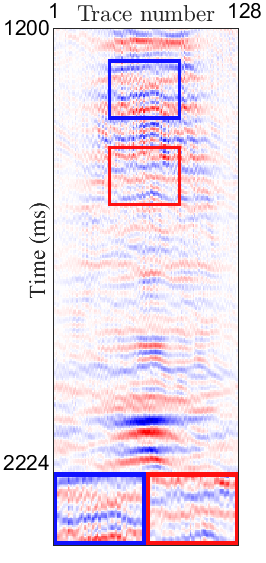}&
			\includegraphics[width=0.108\textwidth]{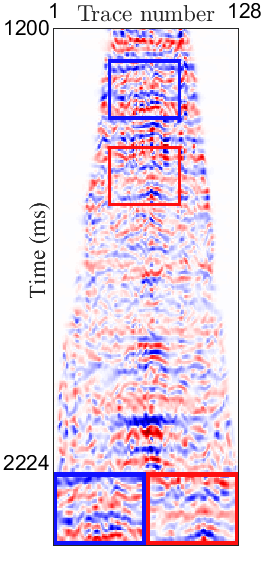}\\
			&
			\includegraphics[width=0.108\textwidth]{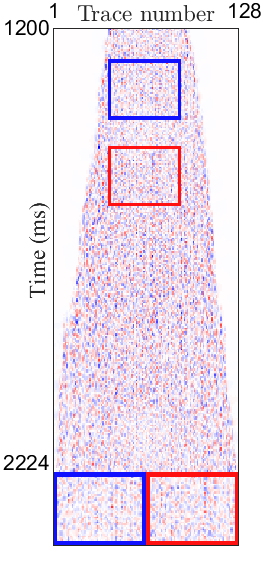}&
			\includegraphics[width=0.108\textwidth]{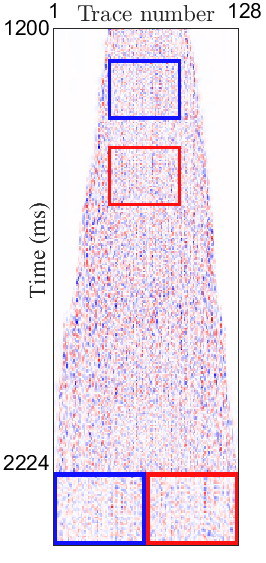}&
			\includegraphics[width=0.108\textwidth]{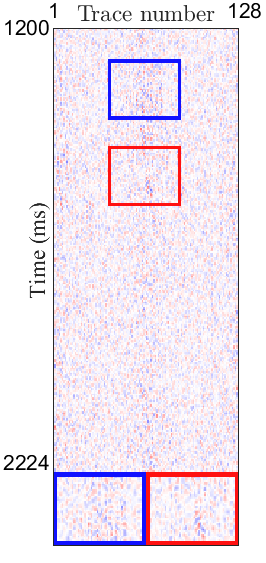}&
			\includegraphics[width=0.108\textwidth]{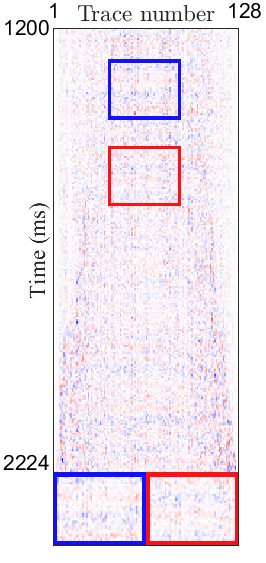}&
			\includegraphics[width=0.108\textwidth]{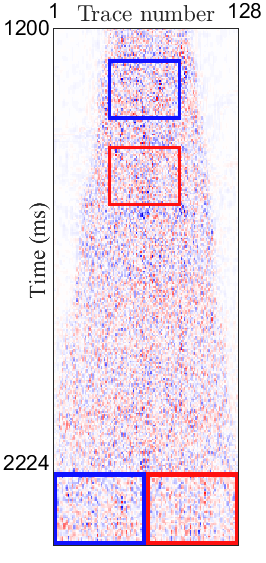}&
			\includegraphics[width=0.108\textwidth]{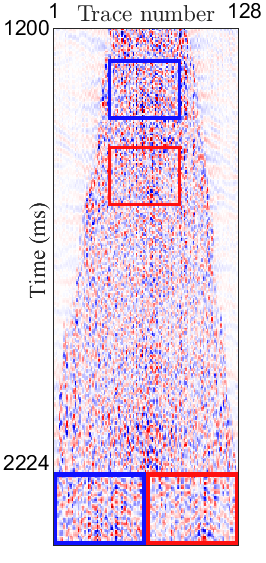}&
			\includegraphics[width=0.108\textwidth]{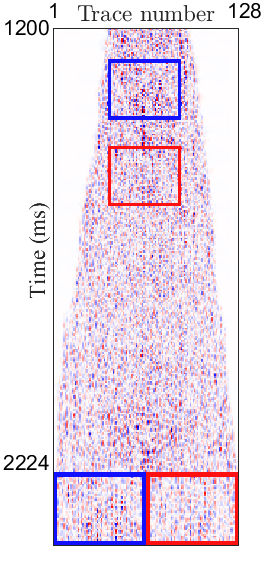}\\
			&LS 0.457&LS 0.424&LS 0.426&LS 0.347&LS 0.200&LS 0.409&LS 0.163\\
			&Time 5s&Time 269s&Time 5s&Time 48s&Times 882s&Time 329s&Time 143s\\
			Noisy&BM3D\cite{BM3D}&WNNM\cite{WNNM}&MSSA\cite{MSSA}&DDAE\cite{DDAE}&DIP\cite{DIP}&PATCHUNET\cite{GP_PATCHUNET}&S2S-WTV\\
			\vspace{-0.6cm}
		\end{tabular}
	\end{center}
	\caption{The first two CMP gathers of the noise attenuation results by different methods (the first and third rows) and the corresponding residual maps between the noisy data and denoising results (the second and fourth rows) on field noisy seismic data {\it X-Pre}.\label{fig_results_real_0}}
	\vspace{-0.2cm}
\end{figure*}
\subsubsection{Fine-Tuning Strategy for High-Dimensional Data}
In some scenarios, we are given a group of noisy seismic data $\{{\bf Y}_k\}_{k=1}^K$ at adjacent inline/crossline positions. They often share similar structures since the positions of adjacent two receivers are close. Therefore, we develop a fine-tuning training strategy to deal with such a high-dimensional noisy data group. Specifically, we use the well-trained CNN weights of ${\bf Y}_1$ (denoted by $\theta_1$) as the initial weights for training the CNN $f_{\theta_k}(\cdot)$ of the rest data $\{{\bf Y}_{k}\}_{k=2}^{K}$. Suppose that we take a large iteration number $T_1$ of the ADMM-based algorithm to train the first CNN $f_{\theta_1}(\cdot)$ for ${\bf Y}_1$, it only takes a few iterations ${T_{k}}<<T_1$ ($k=2,3,\cdots,K$) to train the CNN $f_{\theta_{k}}(\cdot)$ for ${\bf Y}_{k}$ since $f_{\theta_{k}}(\cdot)$ adopts $\theta_1$ as the initial weights, which greatly speeds up the convergence.\par 
The ADMM-based algorithm can be seen as the inner loop of the fine-tuning strategy, where each observation ${\bf Y}_k$ is applied with $T_k$ steps of ADMM-based algorithm to train the CNN $f_{\theta_k}(\cdot)$ for $k=1,2,\cdots,K$. The overall training strategies for high-dimensional seismic data noise attenuation are summarized in Algorithm \ref{alg_1}. 
\vspace{-0.4cm}
\subsection{Inference Strategy}\label{Sec_inference}
At the inference stage, we generate multiple CNNs from the trained CNN by conducting dropout in the decoding stage. These generated CNNs are likely to have certain degree of independence, and thus can reduce the variance (noise) of the denoising result by averaging all outputs\cite{S2S}. Formally, suppose that $P$ CNNs $f_{\theta}^1(\cdot),f_{\theta}^2(\cdot),\cdots,f_{\theta}^P(\cdot)$ are generated by using random dropout in the decoding block of the trained CNN $f_{\theta}(\cdot)$. We feed $P$ newly masked instances $\{\widehat{\bf Y}'_p\}_{p=1}^P$ (with trace-wise masks) into these networks and calculate their average output
\begin{equation}\label{eq_inference}
{\bf X}'=\frac{1}{P}\sum_{p=1}^P f_{\theta}^p(\widehat{\bf Y}'_p)
\end{equation}
as the denoising result. The average output of these relatively independent CNNs can effectively reduce the variance of the result and thus obtain a cleaner recovered seismic signal. In this work, we set $P=100$ as a fixed parameter, which consistently produces satisfactory results.
\vspace{-0.2cm}
\section{Experiments}\label{Sec_exp}
In this section, we present the experimental results to verify the effectiveness of our method. We first introduce the detailed experimental settings and then present the results on synthetic and field noisy seismic datasets. All experiments are conducted on a computer with an Intel(R) i7-12700H CPU and an RTX 3070 GPU (8 GB GPU memory).
\subsection{Experimental Settings}
We compare our S2S-WTV with six representative seismic data denoising methods, including three traditional methods (BM3D\cite{BM3D}, WNNM\cite{WNNM}, and MSSA\cite{MSSA}) and three deep learning methods (DDAE\cite{DDAE}, DIP\cite{DIP}, and PATCHUNET\cite{GP_PATCHUNET}). Here, DDAE is a supervised method and we use the pre-trained model provided by the authors. DIP and PATCHUNET are self-supervised methods that solely use the observed noisy data to train the network. We tune the hyperparameters of different methods to obtain the best peak signal-to-noise ratio.\par
We include both synthetic and field noisy seismic datasets for testing. The synthetic data include 5 seismic datasets, denoted as {\it Datasets (1)-(5)}. {\it Datasets (1)-(3)} are three patches cutting from the {\it Marmousi} post-stack seismic dataset and {\it Datasets (4)-(5)} are two patches cutting from the {\it SEG C3} pre-stack seismic dataset. The size of {\it Datasets (1)-(3)} is $256\times256$ and the size of {\it Datasets (4)-(5)} is $192\times192$. We consider synthetic random Gaussian noise and bandpass noise (by filtering Gaussian noise with the same frequency band as seismic data) with standard deviations (denoted by $\sigma$) of 0.1, 0.2, and 0.3. The field datasets include one pre-stack seismic dataset (the {\it X-Pre} oil field data with size $256\times128\times5$) and three post-stack seismic datasets (the {\it X} oil field data, the {\it F3} seismic data, and the {\it Kerry} seismic data\footnote{The seismic datasets {\it Marmousi}, {\it SEG C3}, {\it F3}, and {\it Kerry} are online available at \url{https://wiki.seg.org/wiki/Open_data}.}, all of which have the size of $256\times256\times5$). We employ the proposed fine-tuning strategy (see Algorithm \ref{alg_1}) to deal with these high-dimensional field datasets. Besides, we set the iteration numbers $T_1=5000$ and $T_k=500$ for $k>1$. The mask rate of the trace-wise masks is set to 0.4. The dropout rate in the decoding blocks is set to 0.5. The hyperparameters $\gamma$ and $\mu$ are set to 0.01 and 0.1. The sampling number at the inference stage (i.e., $P$) is set to 100. The above settings are applied to all datasets.
\par 
We report the peak signal-to-noise ratio (PSNR), structural similarity (SSIM), and local similarity (LS)\cite{LS} of the denoising results. Here, the LS is calculated by averaging all elements of the LS map between the denoising result and residual map. We remark that calculating PSNR and SSIM needs the ground-truth data as references while calculating LS does not require ground-truth data. 
    \begin{figure*}[t]
	\scriptsize
	\setlength{\tabcolsep}{0.9pt}
	\begin{center}
		\begin{tabular}{cccccccc}
			\includegraphics[width=0.12\textwidth]{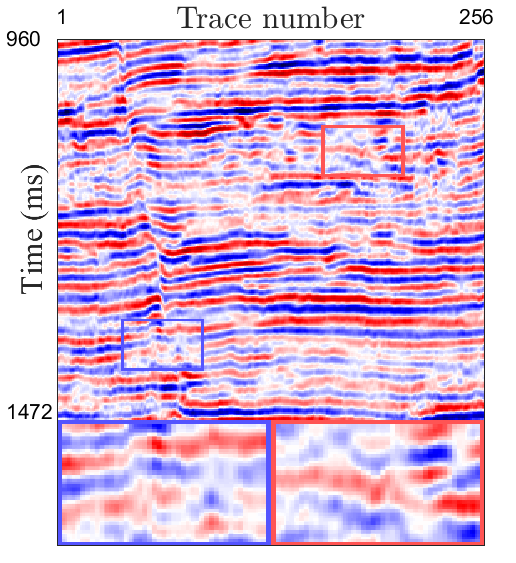}&
			\includegraphics[width=0.12\textwidth]{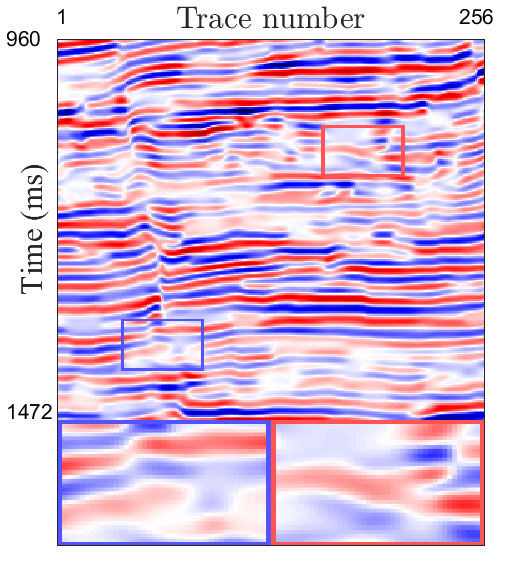}&
			\includegraphics[width=0.12\textwidth]{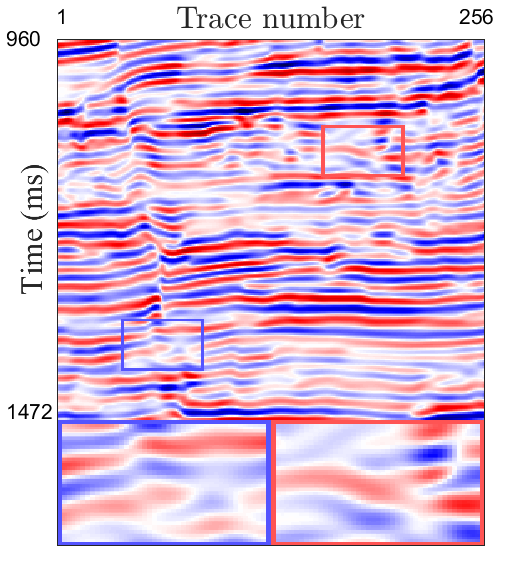}&
			\includegraphics[width=0.12\textwidth]{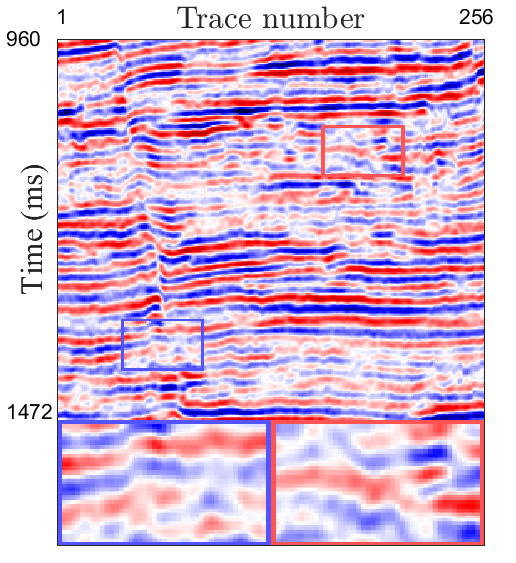}&
			\includegraphics[width=0.12\textwidth]{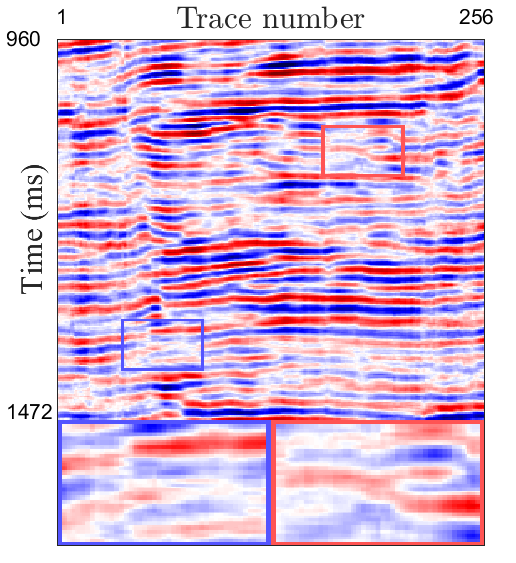}&
			\includegraphics[width=0.12\textwidth]{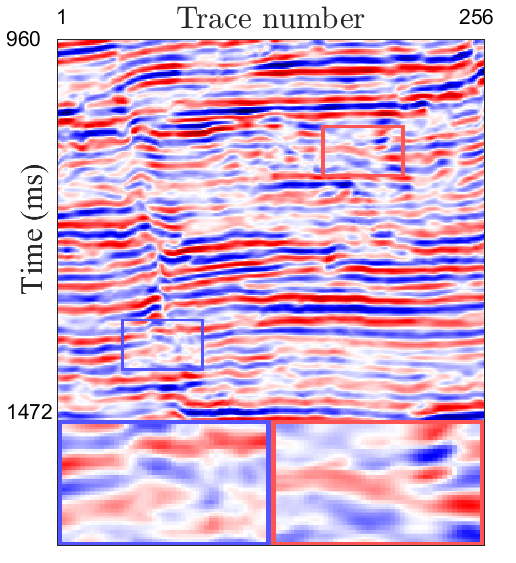}&
			\includegraphics[width=0.12\textwidth]{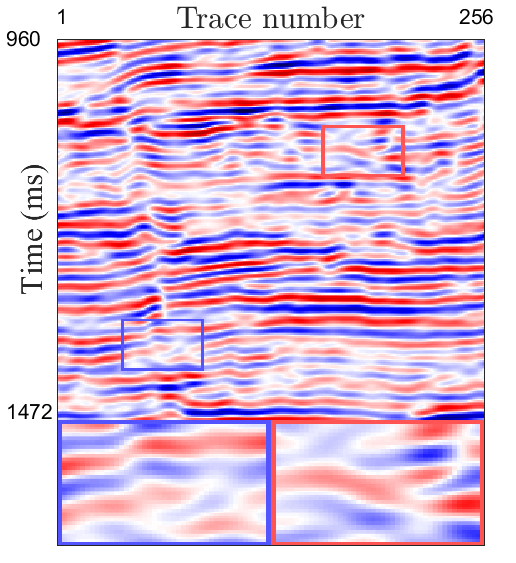}&
			\includegraphics[width=0.12\textwidth]{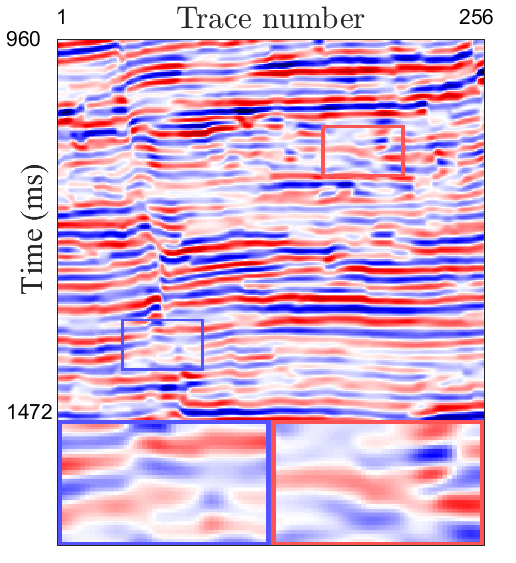}\\
			&
			\includegraphics[width=0.12\textwidth]{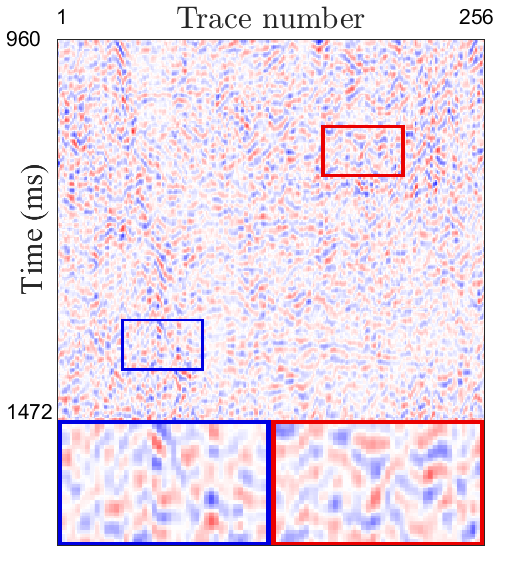}&
			\includegraphics[width=0.12\textwidth]{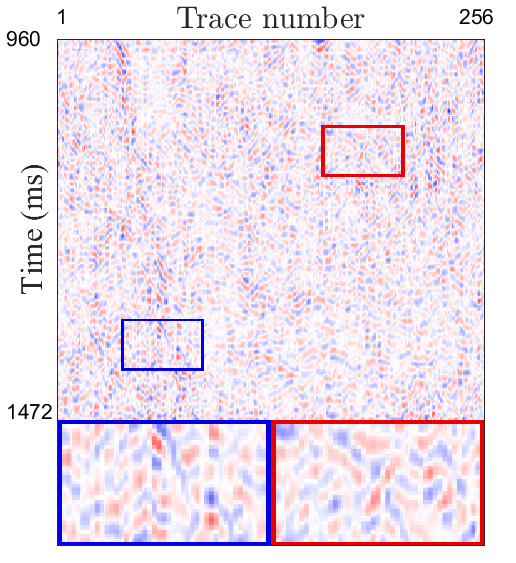}&
			\includegraphics[width=0.12\textwidth]{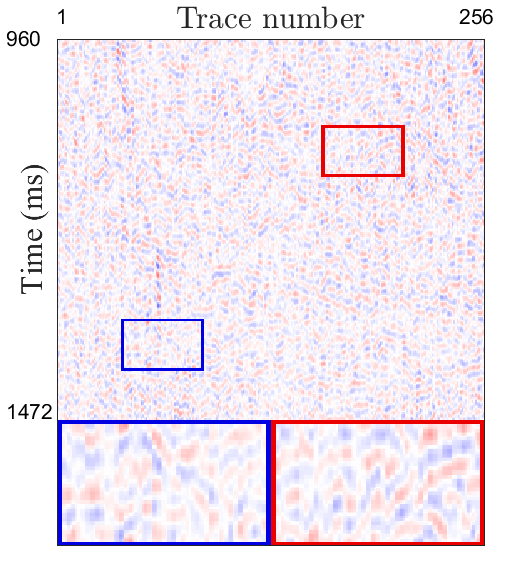}&
			\includegraphics[width=0.12\textwidth]{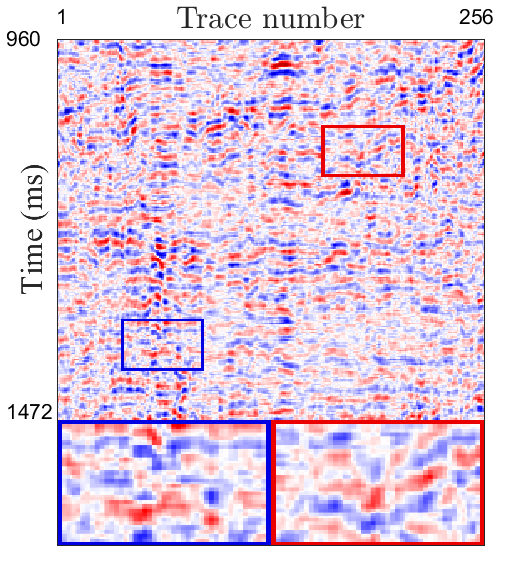}&
			\includegraphics[width=0.12\textwidth]{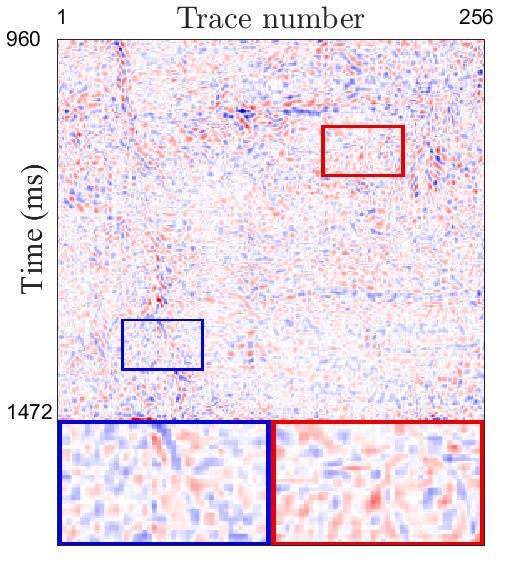}&
			\includegraphics[width=0.12\textwidth]{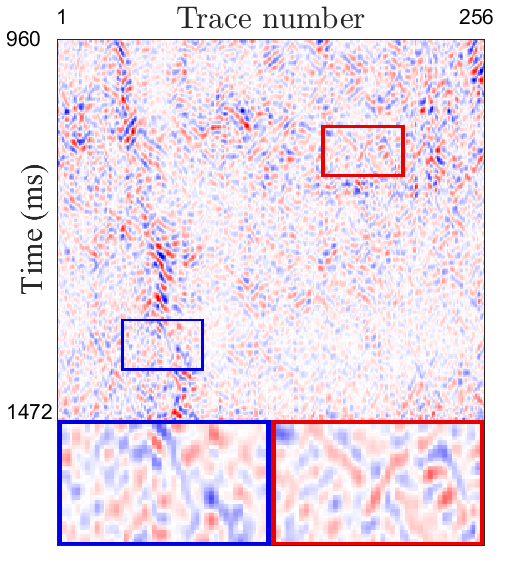}&
			\includegraphics[width=0.12\textwidth]{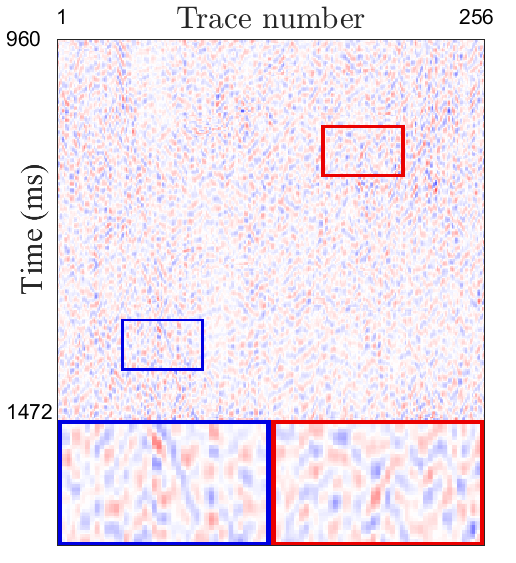}\\
			\includegraphics[width=0.12\textwidth]{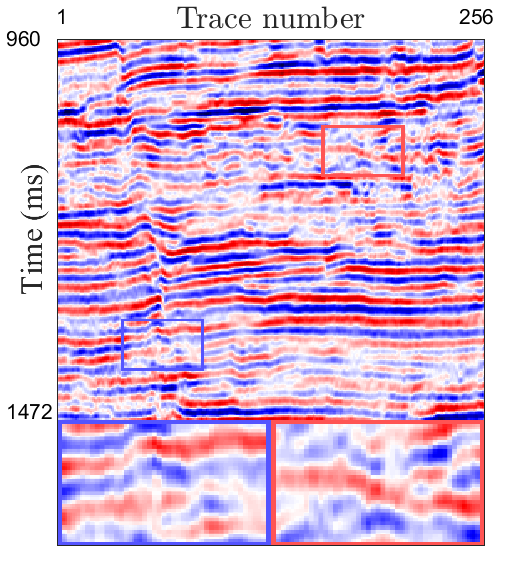}&
			\includegraphics[width=0.12\textwidth]{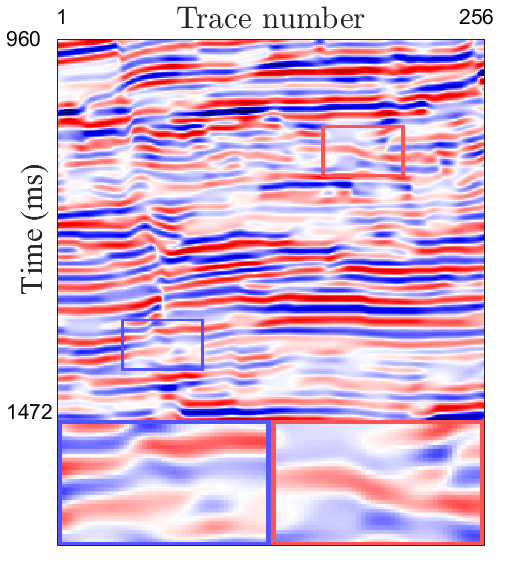}&
			\includegraphics[width=0.12\textwidth]{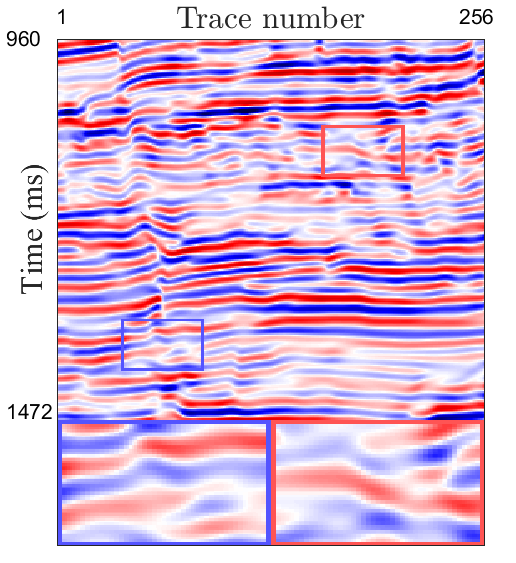}&
			\includegraphics[width=0.12\textwidth]{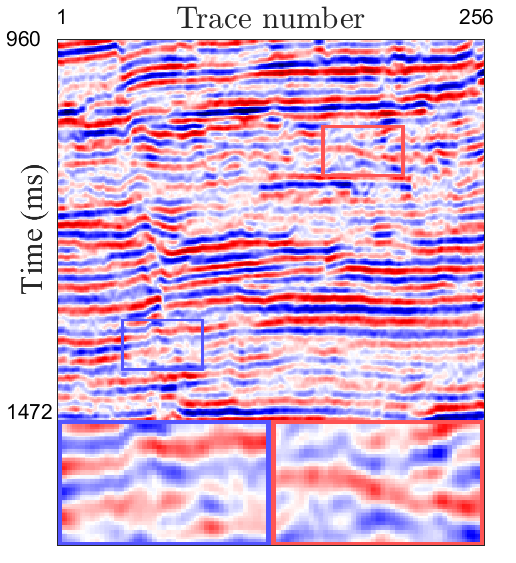}&
			\includegraphics[width=0.12\textwidth]{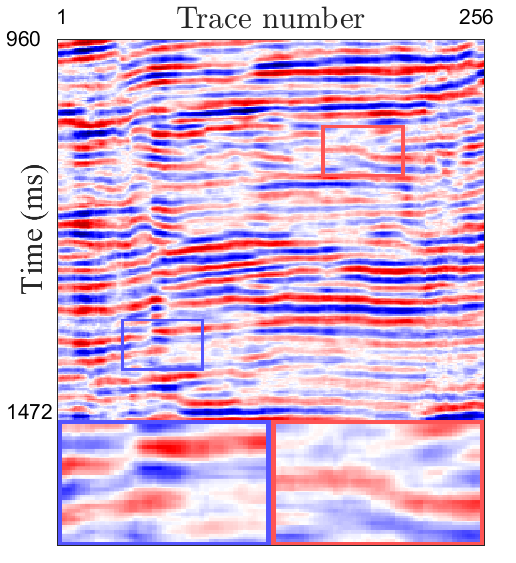}&
			\includegraphics[width=0.12\textwidth]{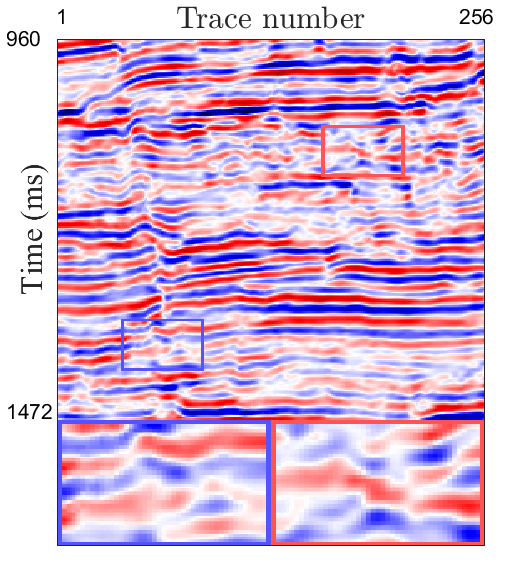}&
			\includegraphics[width=0.12\textwidth]{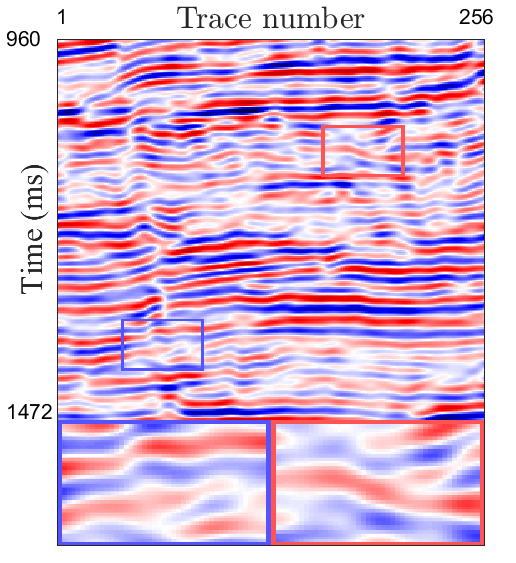}&
			\includegraphics[width=0.12\textwidth]{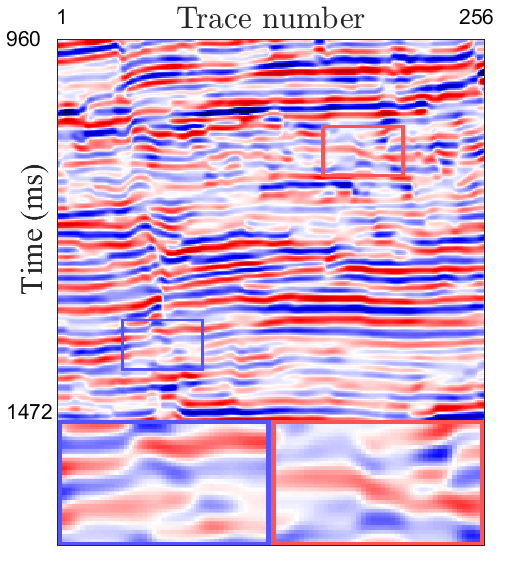}\\
			&
			\includegraphics[width=0.12\textwidth]{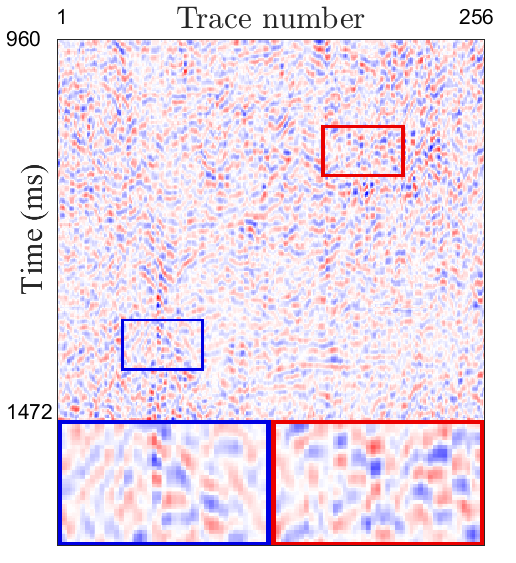}&
			\includegraphics[width=0.12\textwidth]{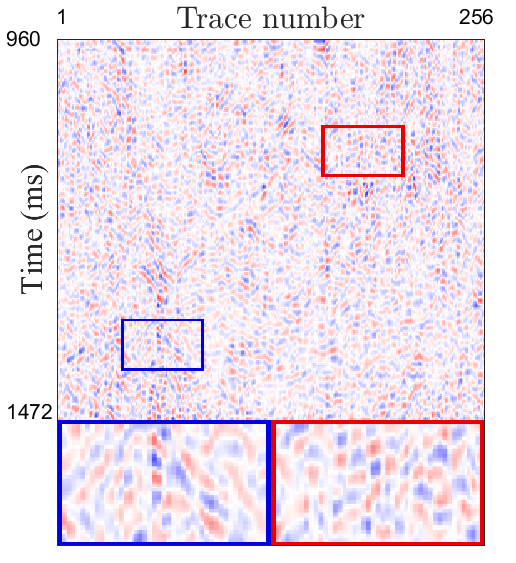}&
			\includegraphics[width=0.12\textwidth]{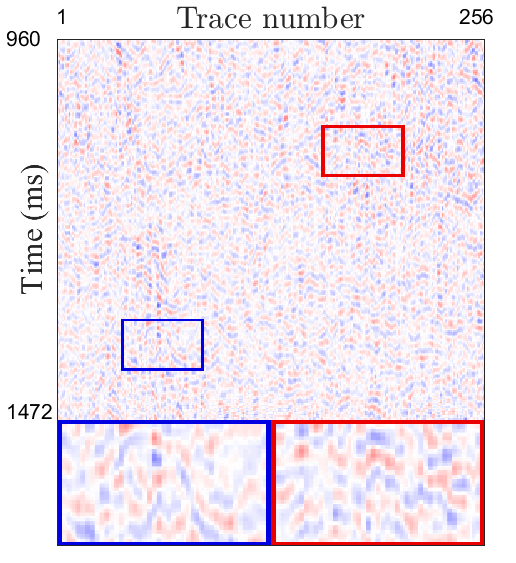}&
			\includegraphics[width=0.12\textwidth]{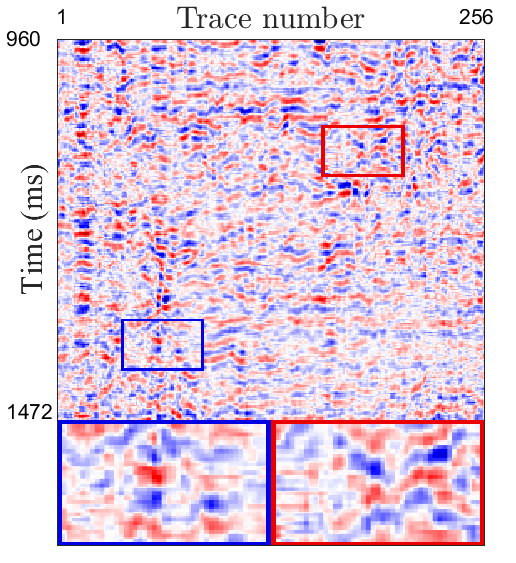}&
			\includegraphics[width=0.12\textwidth]{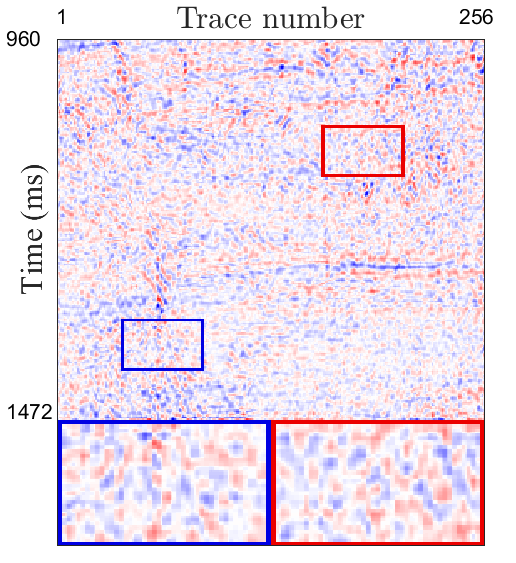}&
			\includegraphics[width=0.12\textwidth]{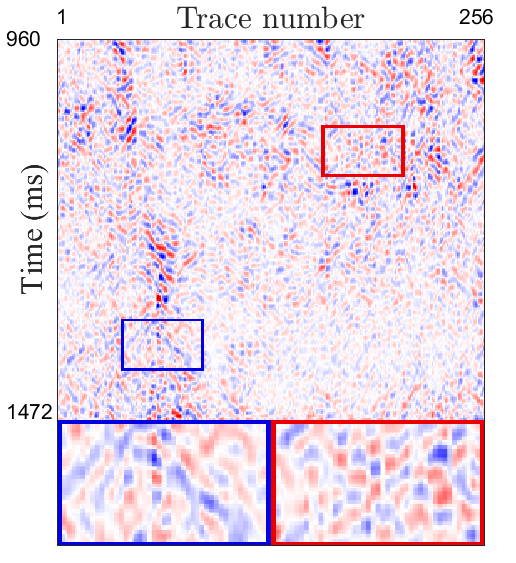}&
			\includegraphics[width=0.12\textwidth]{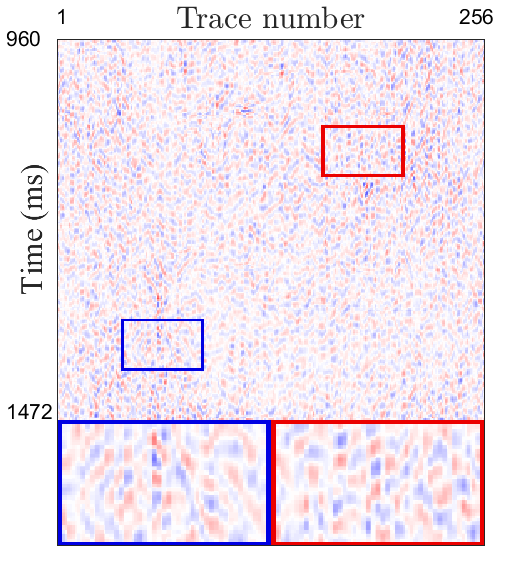}\\
			&LS 0.377&LS 0.322&LS 0.320&LS 0.444&LS 0.301&LS 0.294&LS 0.283\\
			&Time 9s&Time 562s&Time 5s&Time 55s&Times 911s&Time 910s&Time 224s\\
			Noisy&BM3D\cite{BM3D}&WNNM\cite{WNNM}&MSSA\cite{MSSA}&DDAE\cite{DDAE}&DIP\cite{DIP}&PATCHUNET\cite{GP_PATCHUNET}&S2S-WTV\\
			\vspace{-0.6cm}
		\end{tabular}
	\end{center}
	\caption{The first two 2-D slices of the noise attenuation results by different methods (the first and third rows) and the corresponding residual maps between the noisy data and denoising results (the second and fourth rows) on field noisy seismic data {\it X}.\label{fig_results_real_1}}
	\vspace{-0.3cm}
\end{figure*}
\begin{figure}[h]
	\scriptsize
	\setlength{\tabcolsep}{0.9pt}
	\begin{center}
		\begin{tabular}{cccc}
			\includegraphics[width=0.12\textwidth]{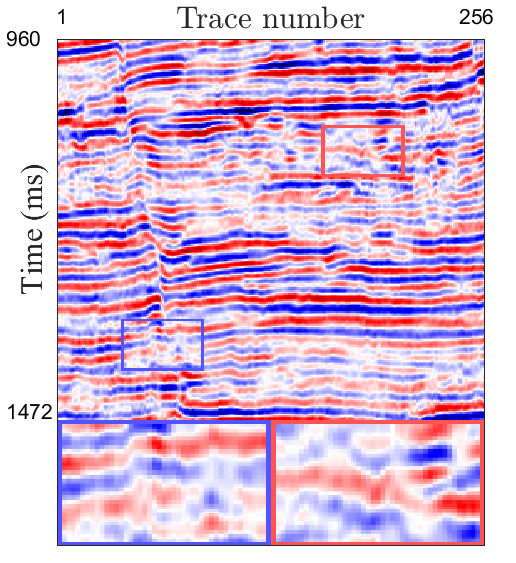}&
			\includegraphics[width=0.12\textwidth]{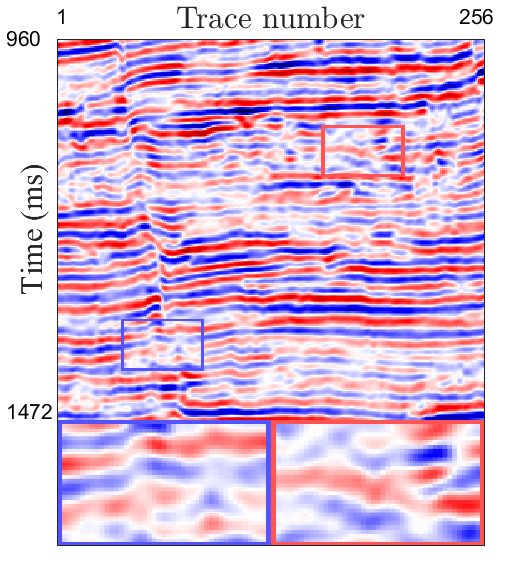}&
			\includegraphics[width=0.12\textwidth]{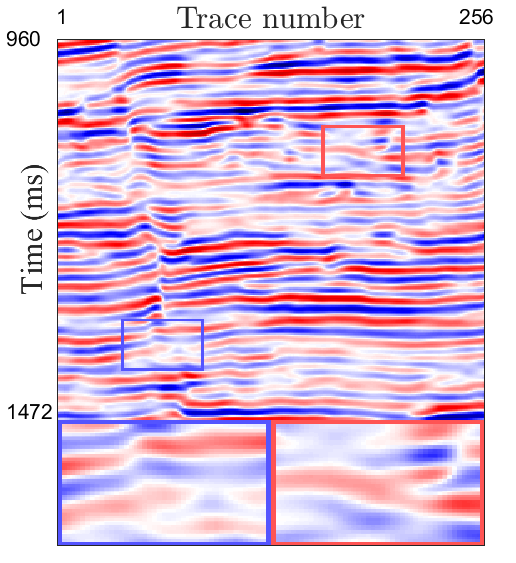}&
			\includegraphics[width=0.12\textwidth]{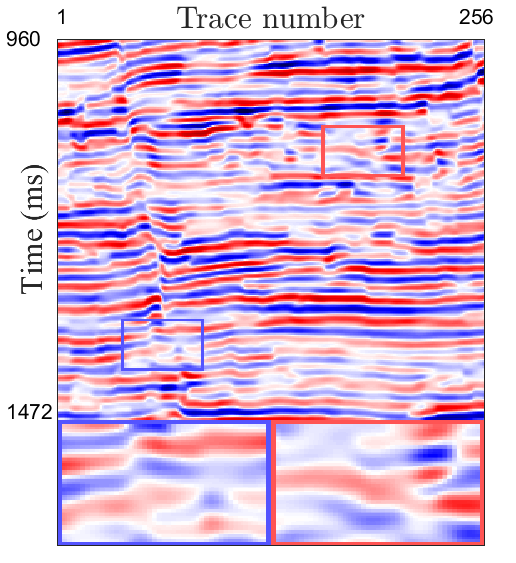}\\
			&
			\includegraphics[width=0.12\textwidth]{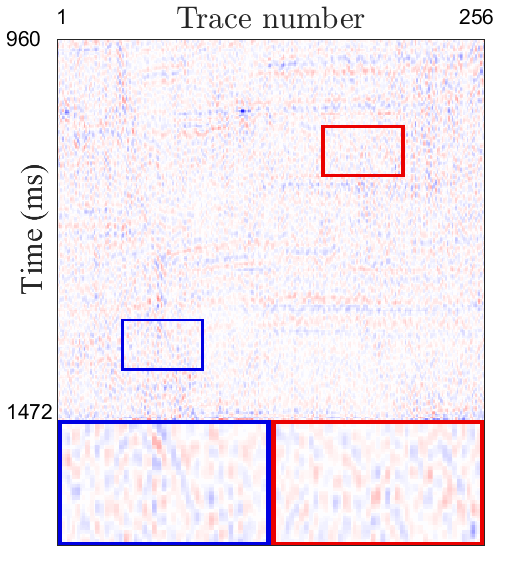}&
			\includegraphics[width=0.12\textwidth]{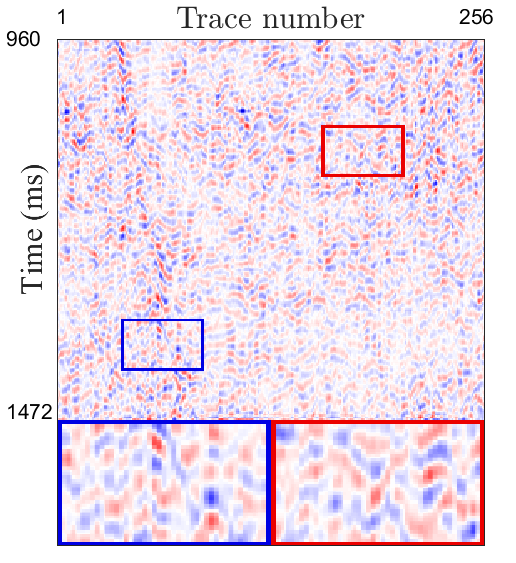}&
			\includegraphics[width=0.12\textwidth]{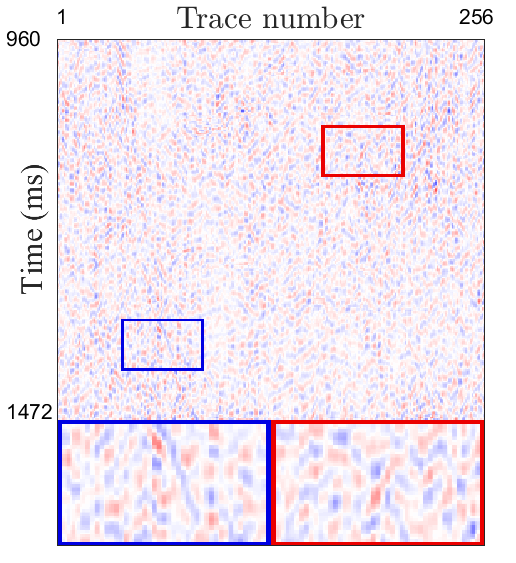}\\
			&LS 0.397 &LS 0.400 &LS 0.281\\
			Noisy& S2S\cite{S2S}&S2S-TV&S2S-WTV\\
			\vspace{-0.5cm}
		\end{tabular}
	\end{center}
	\caption{The noise attenuation results (the first row) and the corresponding residual map between the noisy data and denoising results (the second row) on field noisy seismic data {\it X}. The WTV regularization is helpful to obtain a cleaner result while preserving the signal details.\label{fig_WTV}}
	\vspace{-0.5cm}
\end{figure}
\begin{figure*}[t]
	\scriptsize
	\setlength{\tabcolsep}{0.9pt}
	\begin{center}
		\begin{tabular}{cccccccc}
			\includegraphics[width=0.12\textwidth]{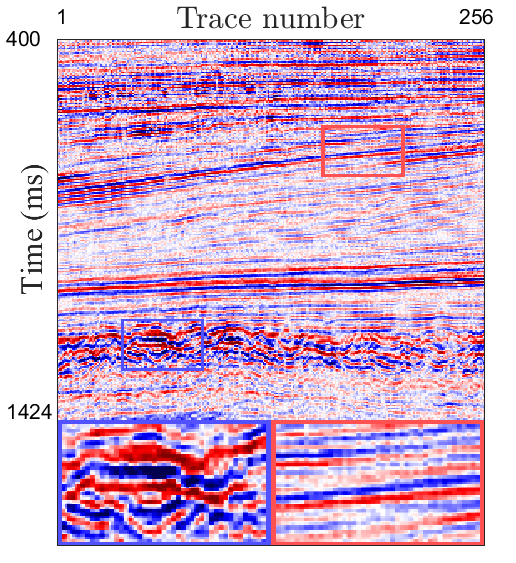}&
			\includegraphics[width=0.12\textwidth]{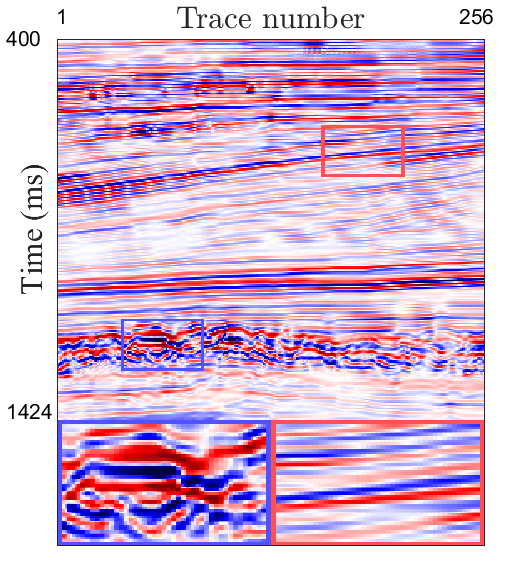}&
			\includegraphics[width=0.12\textwidth]{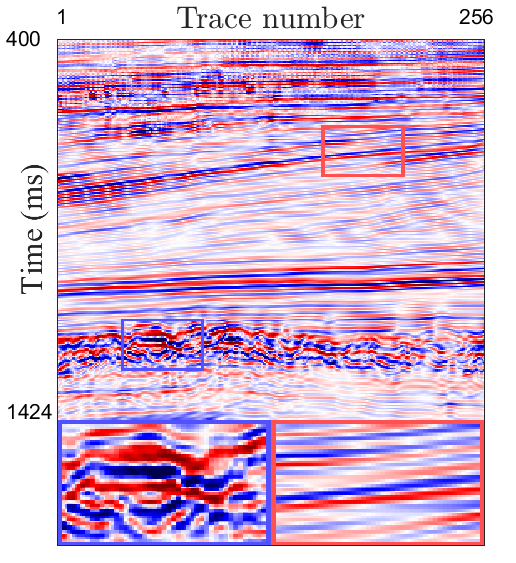}&
			\includegraphics[width=0.12\textwidth]{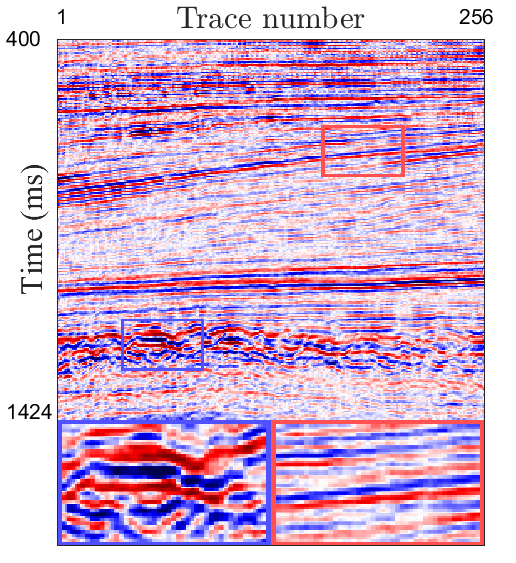}&
			\includegraphics[width=0.12\textwidth]{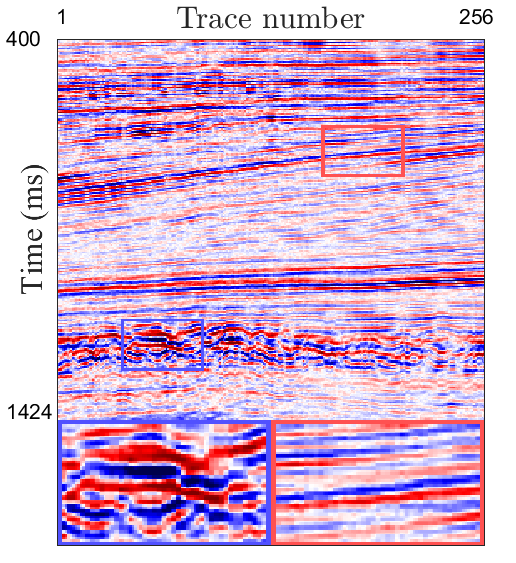}&
			\includegraphics[width=0.12\textwidth]{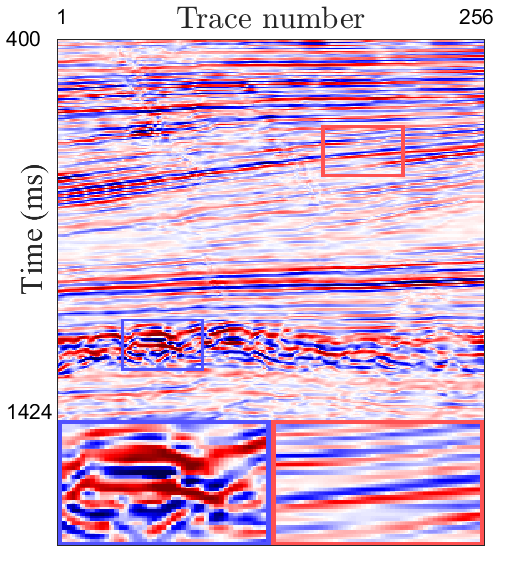}&
			\includegraphics[width=0.12\textwidth]{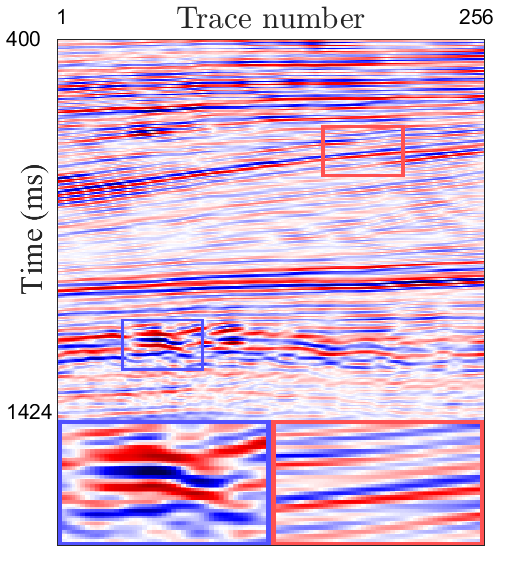}&
			\includegraphics[width=0.12\textwidth]{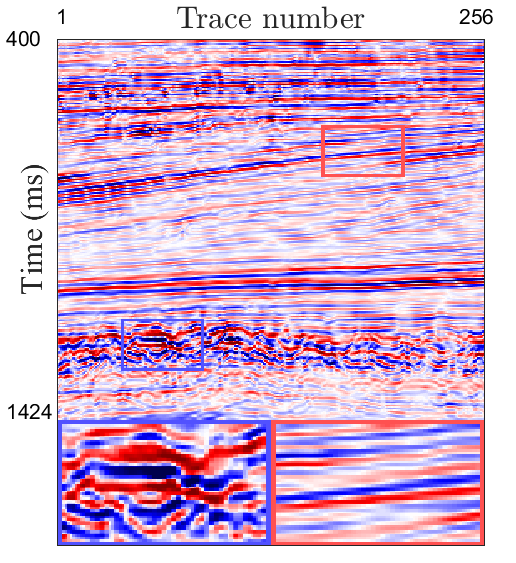}\\
			&
			\includegraphics[width=0.12\textwidth]{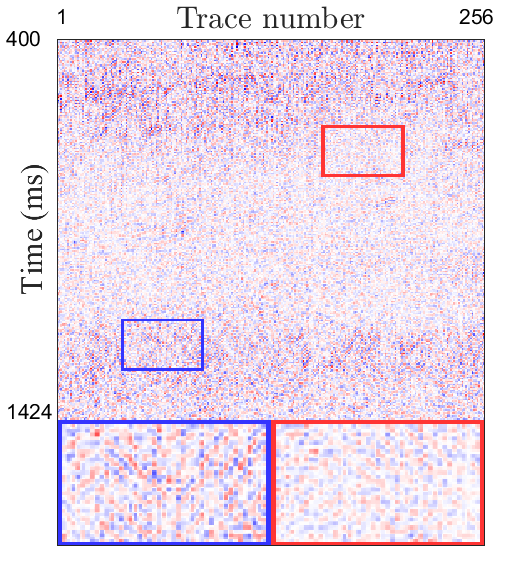}&
			\includegraphics[width=0.12\textwidth]{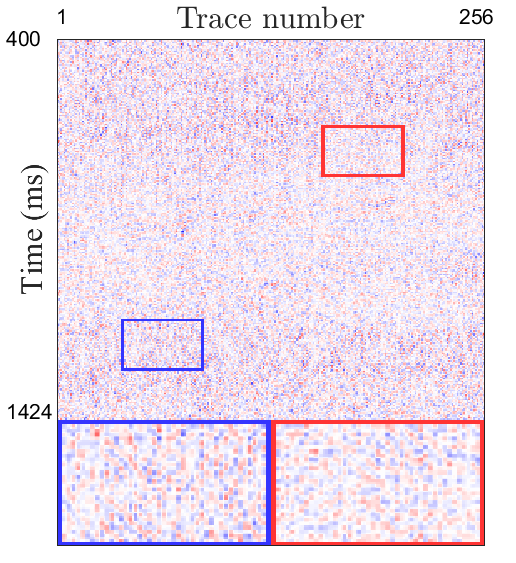}&
			\includegraphics[width=0.12\textwidth]{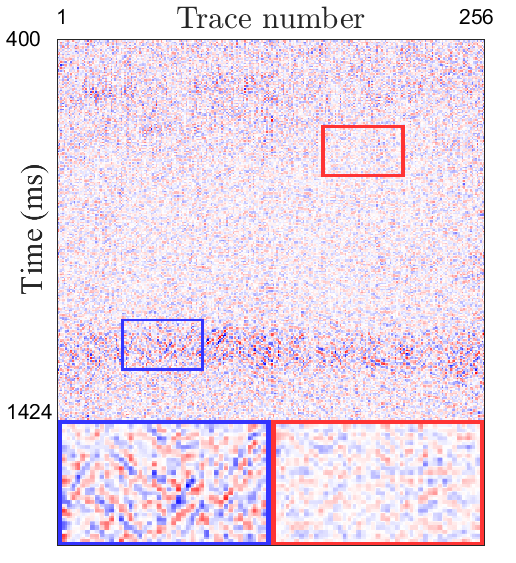}&
			\includegraphics[width=0.12\textwidth]{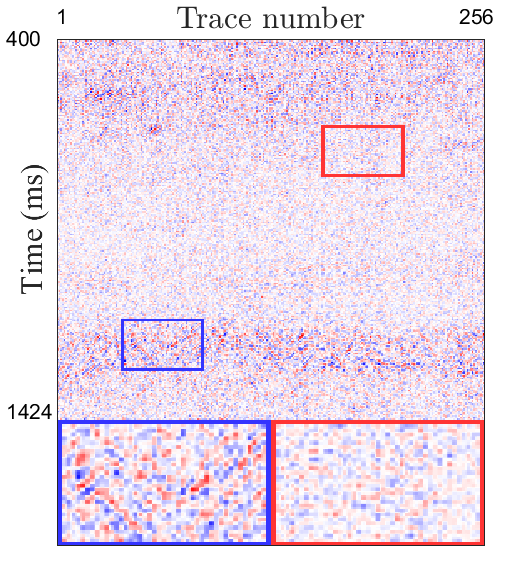}&
			\includegraphics[width=0.12\textwidth]{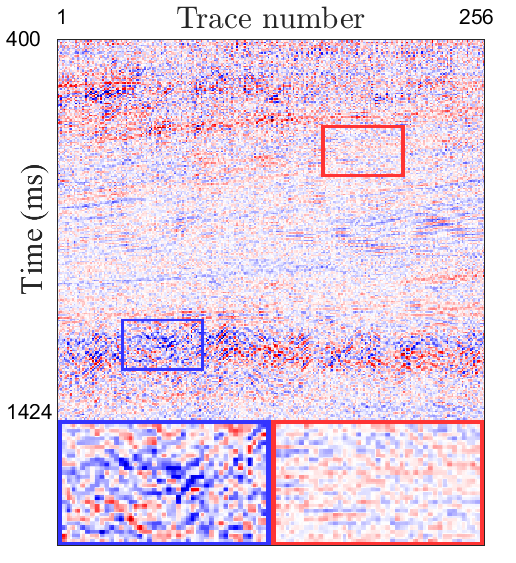}&
			\includegraphics[width=0.12\textwidth]{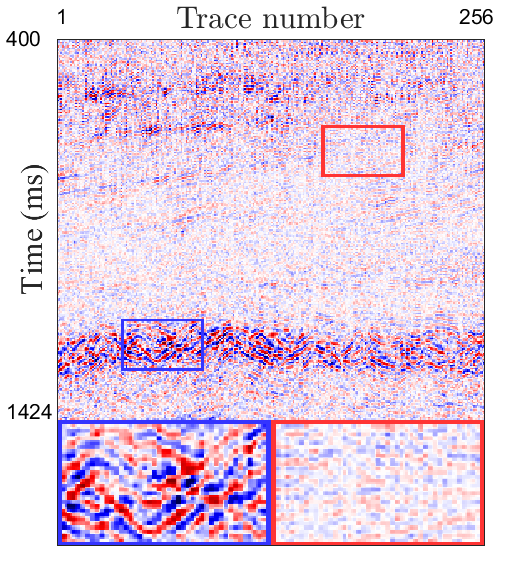}&
			\includegraphics[width=0.12\textwidth]{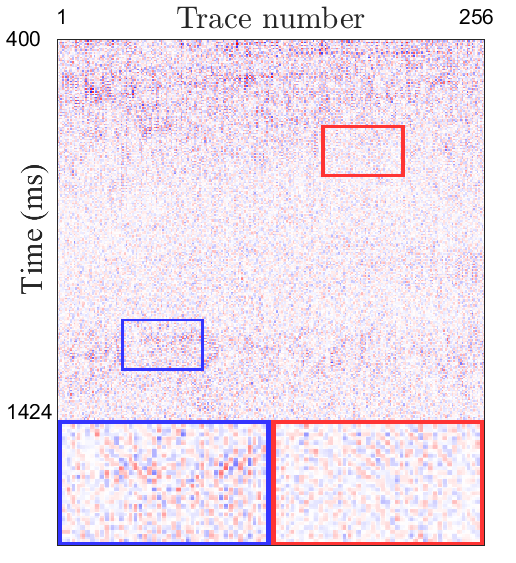}\\
			\includegraphics[width=0.12\textwidth]{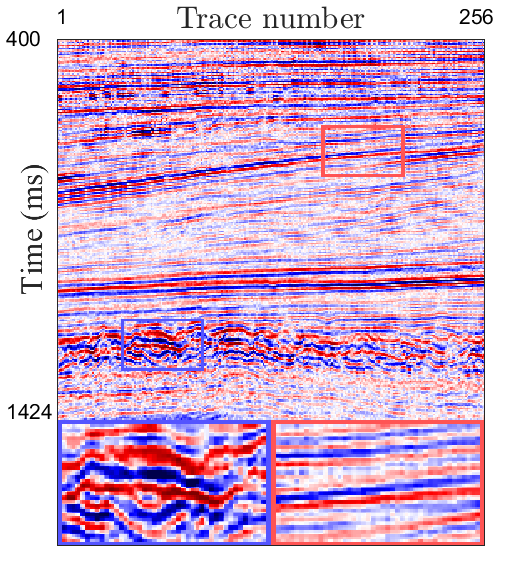}&
			\includegraphics[width=0.12\textwidth]{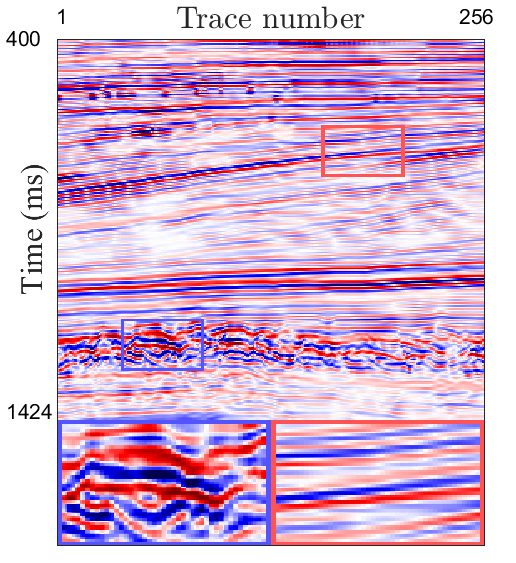}&
			\includegraphics[width=0.12\textwidth]{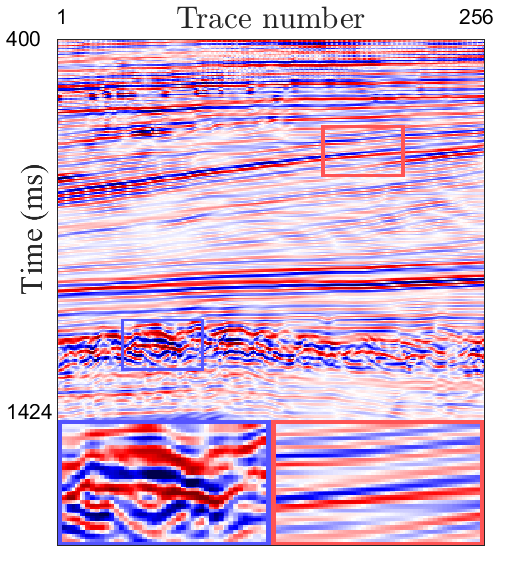}&
			\includegraphics[width=0.12\textwidth]{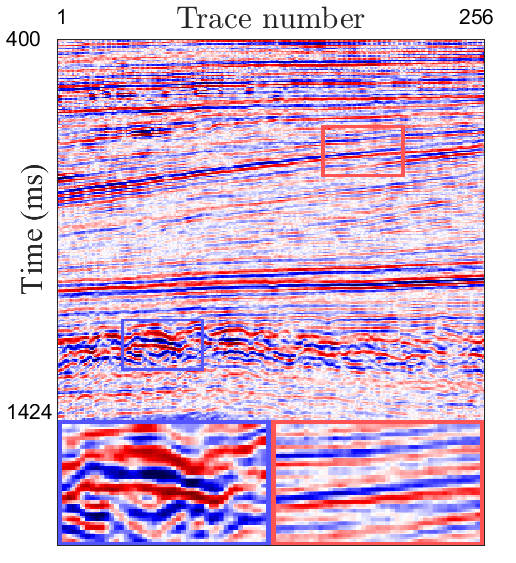}&
			\includegraphics[width=0.12\textwidth]{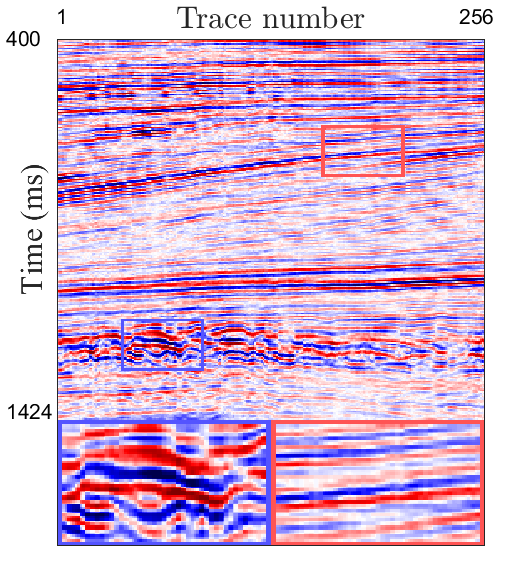}&
			\includegraphics[width=0.12\textwidth]{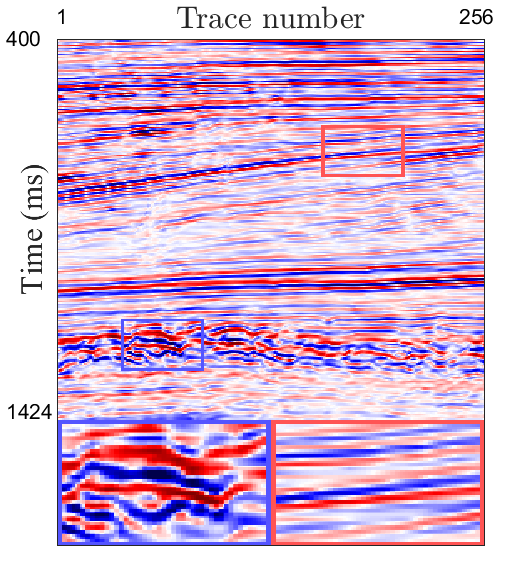}&
			\includegraphics[width=0.12\textwidth]{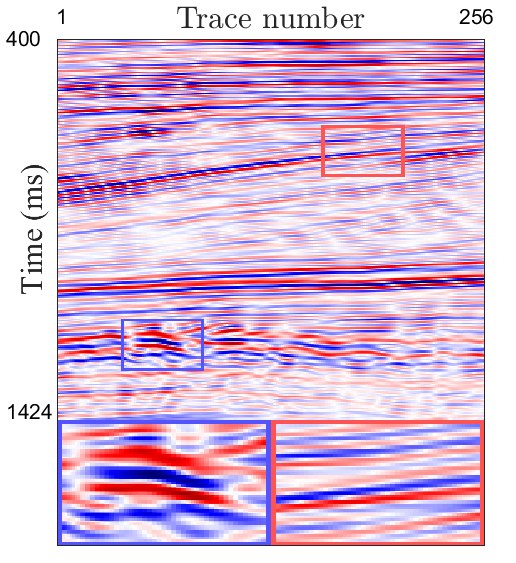}&
			\includegraphics[width=0.12\textwidth]{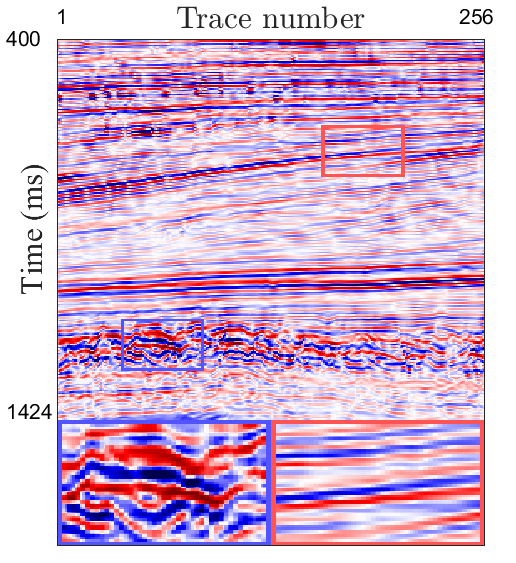}\\
			&
			\includegraphics[width=0.12\textwidth]{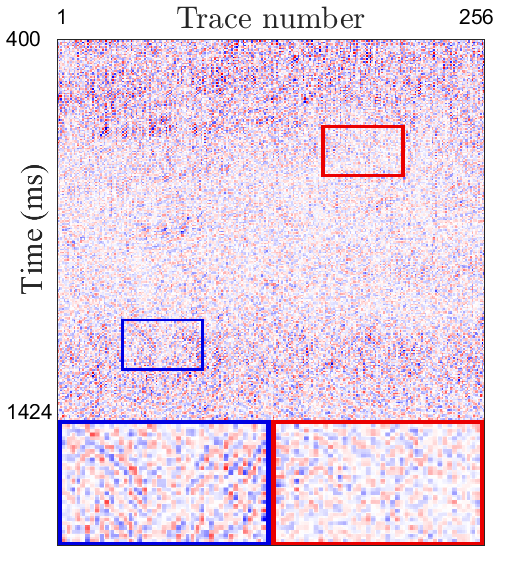}&
			\includegraphics[width=0.12\textwidth]{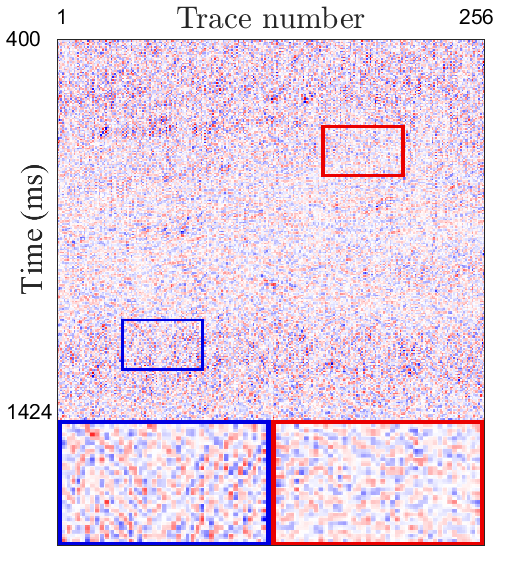}&
			\includegraphics[width=0.12\textwidth]{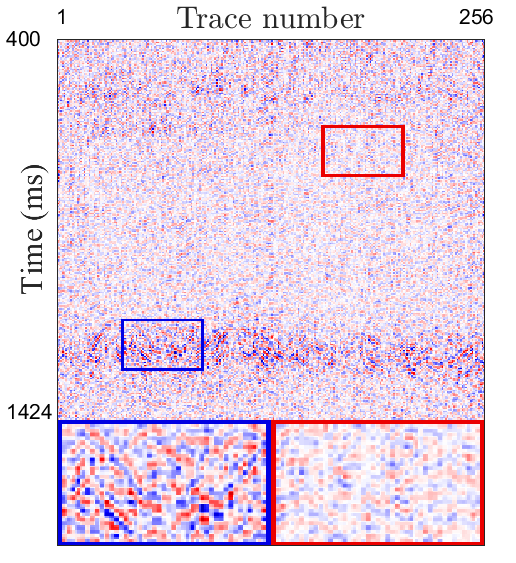}&
			\includegraphics[width=0.12\textwidth]{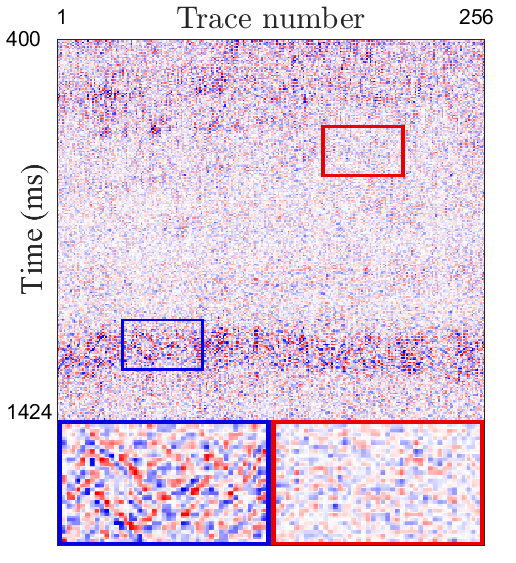}&
			\includegraphics[width=0.12\textwidth]{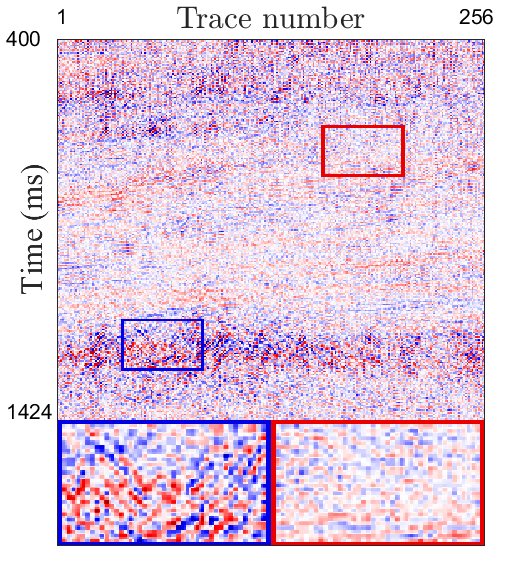}&
			\includegraphics[width=0.12\textwidth]{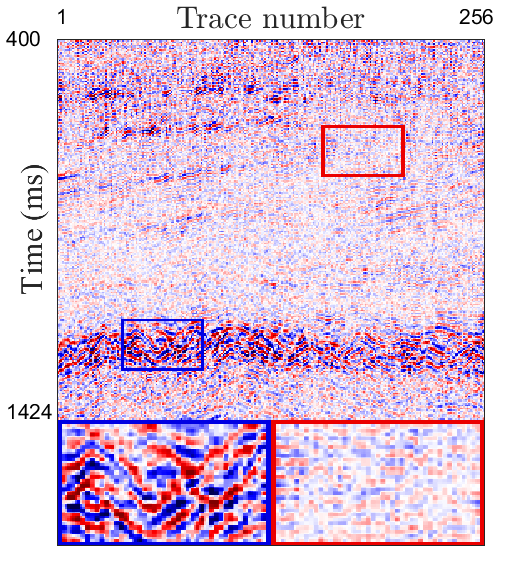}&
			\includegraphics[width=0.12\textwidth]{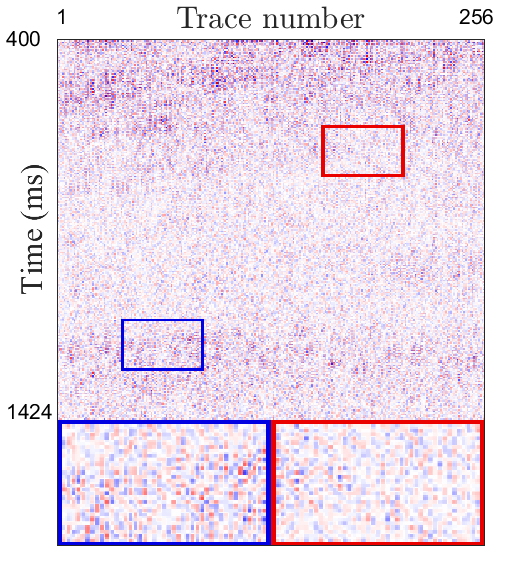}\\
			&LS 0.236&LS 0.225&LS 0.228&LS 0.183&LS 0.182&LS 0.260&LS 0.114\\
			&Time 5s&Time 571s&Time 5s&Time 59s&Times 907s&Time 786s&Time 227s\\
			Noisy&BM3D\cite{BM3D}&WNNM\cite{WNNM}&MSSA\cite{MSSA}&DDAE\cite{DDAE}&DIP\cite{DIP}&PATCHUNET\cite{GP_PATCHUNET}&S2S-WTV\\
			\vspace{-0.6cm}
		\end{tabular}
	\end{center}
	\caption{The first two 2-D slices of the noise attenuation results by different methods (the first and third rows) and the corresponding residual maps between the noisy data and denoising results (the second and fourth rows) on field noisy seismic data {\it F3}.\label{fig_results_real_2}}
	\vspace{-0.3cm}
\end{figure*}
\begin{figure*}[t]
	\scriptsize
	\setlength{\tabcolsep}{0.9pt}
	\begin{center}
		\begin{tabular}{cccccccc}
			\includegraphics[width=0.12\textwidth]{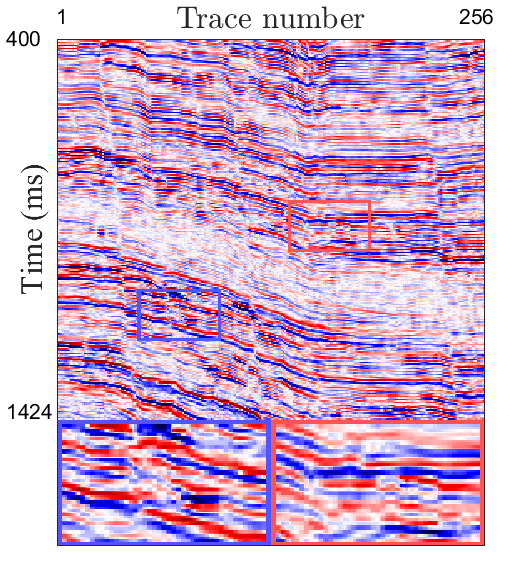}&
			\includegraphics[width=0.12\textwidth]{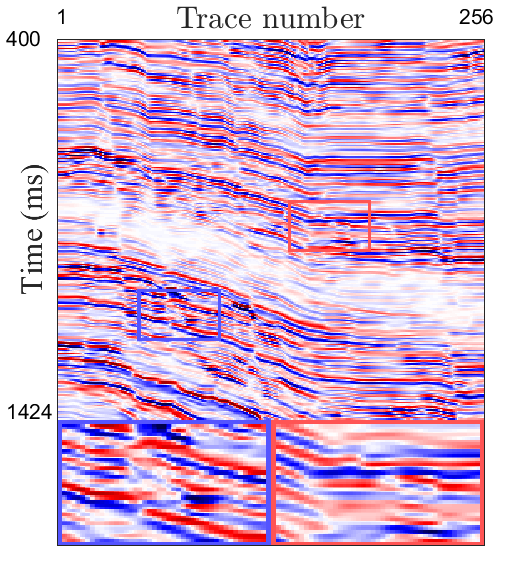}&
			\includegraphics[width=0.12\textwidth]{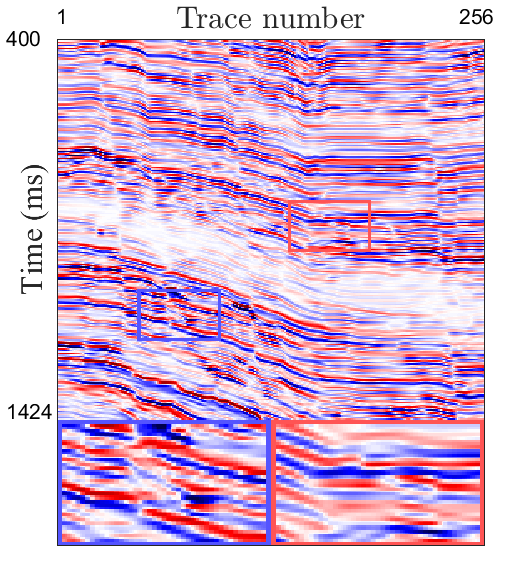}&
			\includegraphics[width=0.12\textwidth]{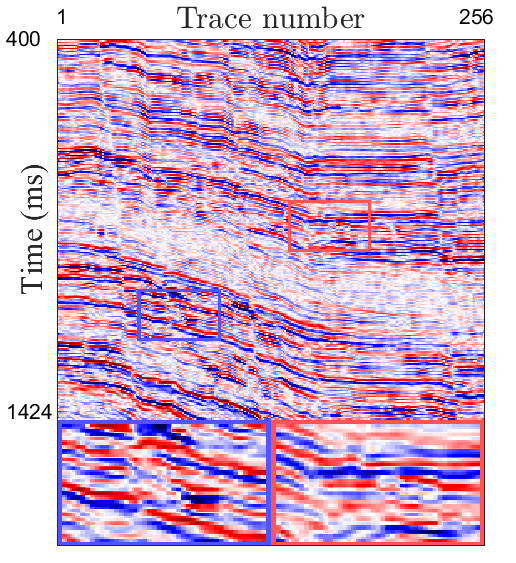}&
			\includegraphics[width=0.12\textwidth]{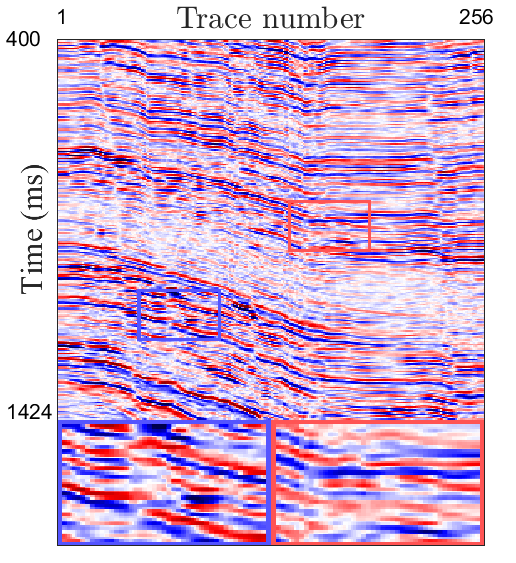}&
			\includegraphics[width=0.12\textwidth]{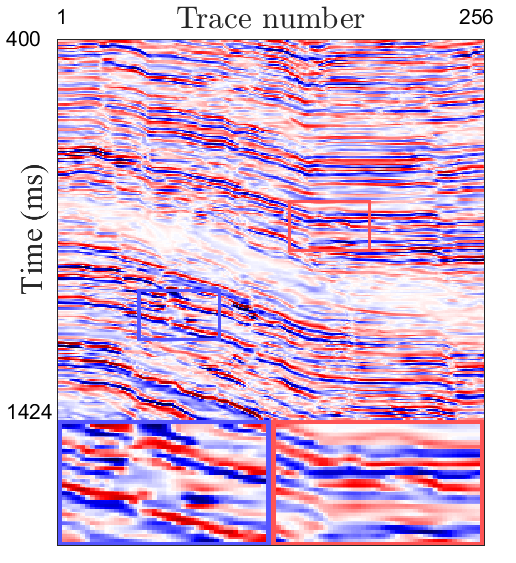}&
			\includegraphics[width=0.12\textwidth]{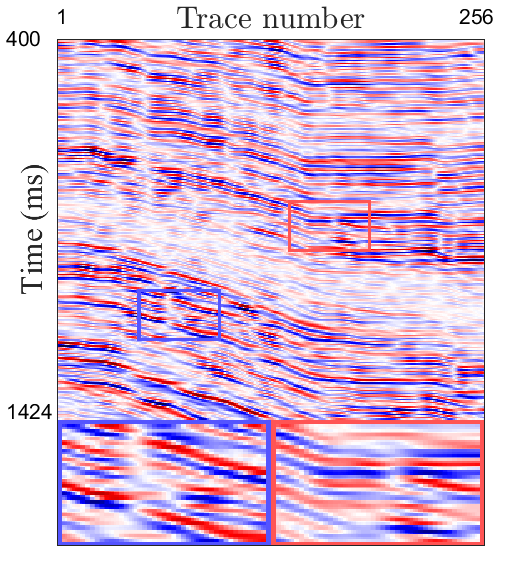}&
			\includegraphics[width=0.12\textwidth]{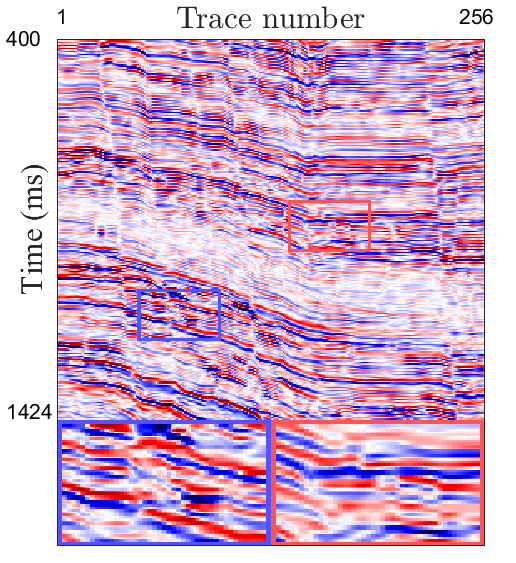}\\
			&
			\includegraphics[width=0.12\textwidth]{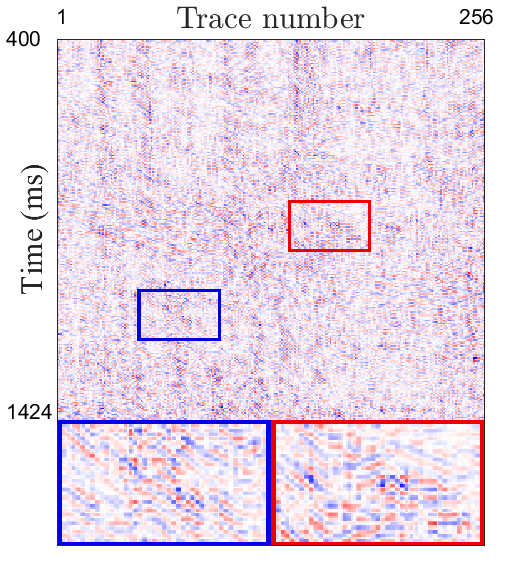}&
			\includegraphics[width=0.12\textwidth]{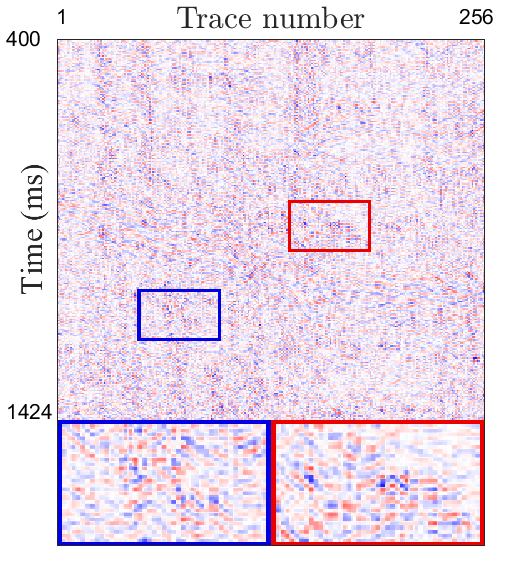}&
			\includegraphics[width=0.12\textwidth]{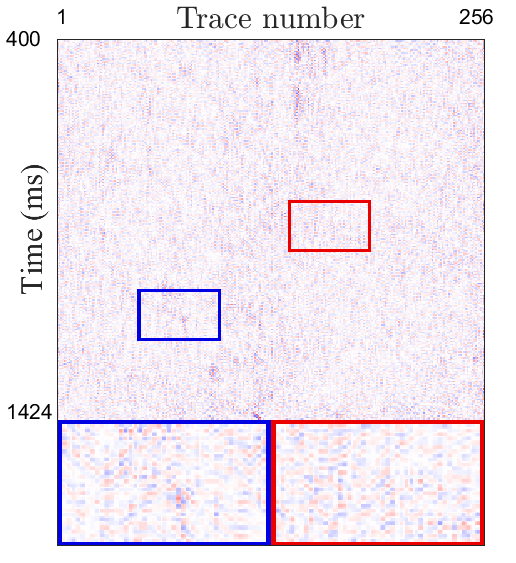}&
			\includegraphics[width=0.12\textwidth]{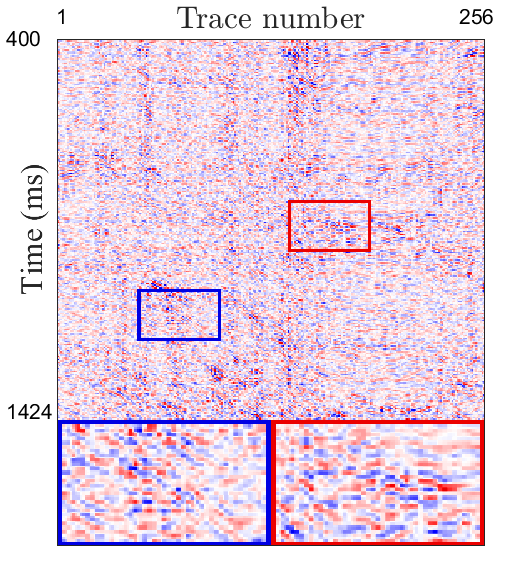}&
			\includegraphics[width=0.12\textwidth]{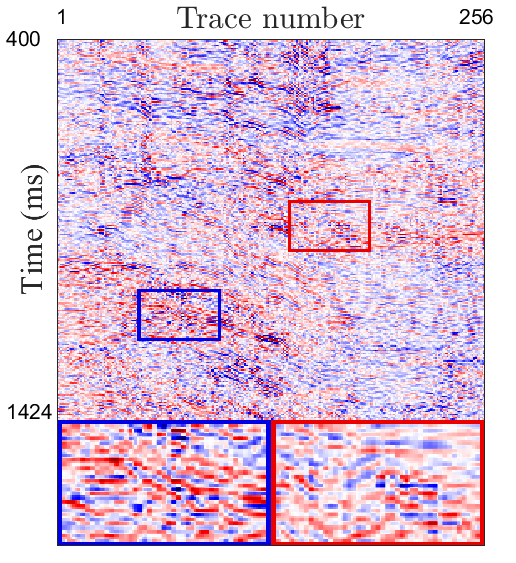}&
			\includegraphics[width=0.12\textwidth]{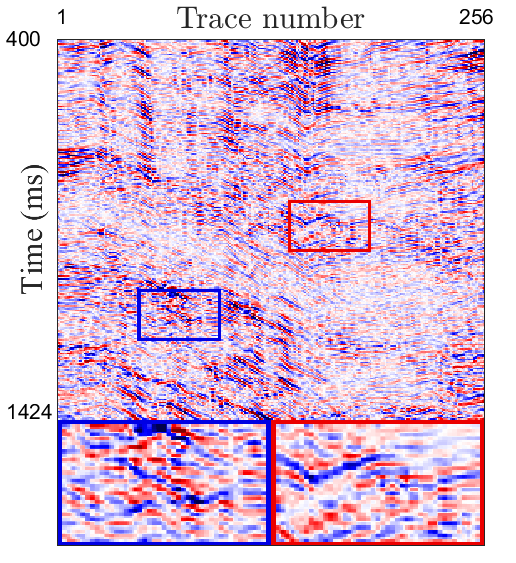}&
			\includegraphics[width=0.12\textwidth]{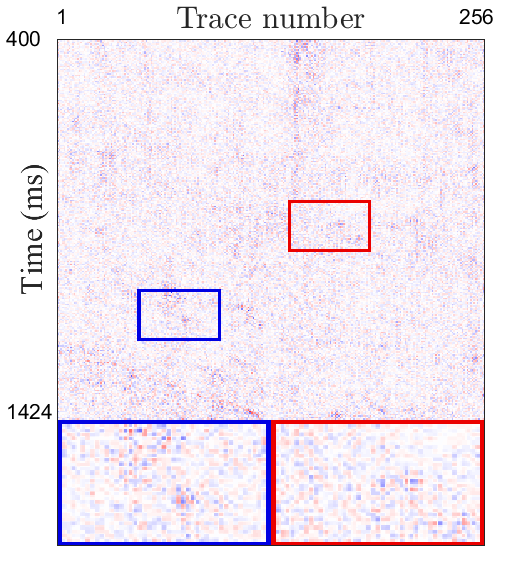}\\
			\includegraphics[width=0.12\textwidth]{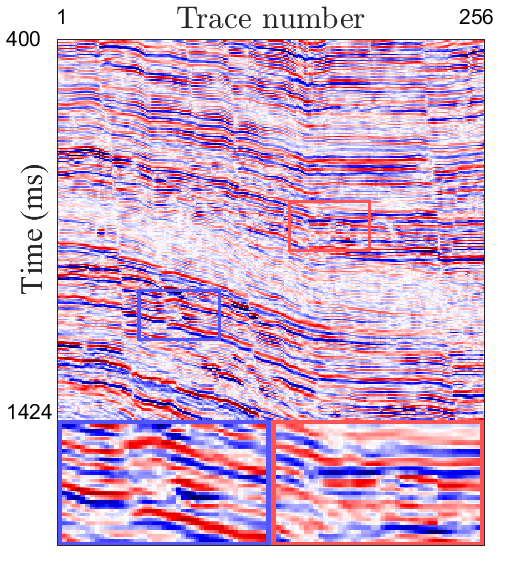}&
			\includegraphics[width=0.12\textwidth]{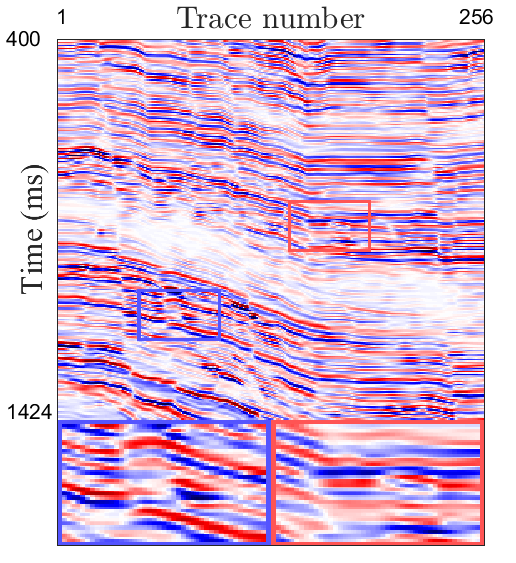}&
			\includegraphics[width=0.12\textwidth]{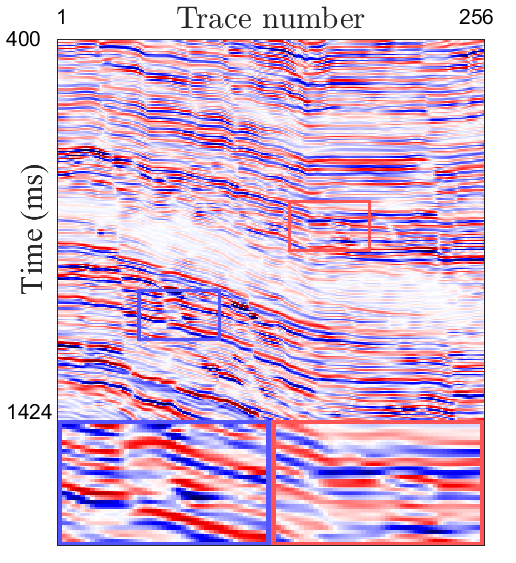}&
			\includegraphics[width=0.12\textwidth]{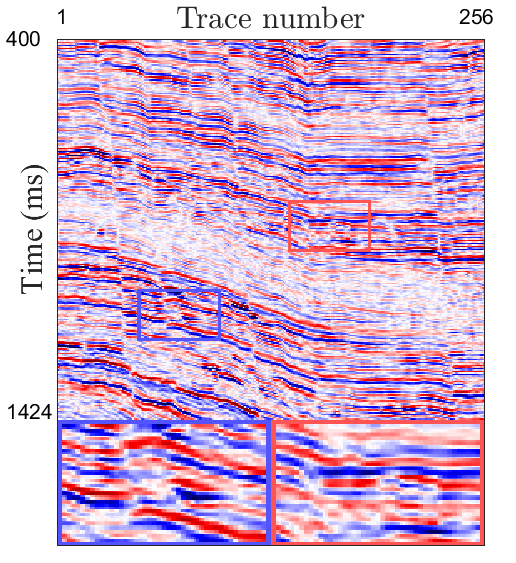}&
			\includegraphics[width=0.12\textwidth]{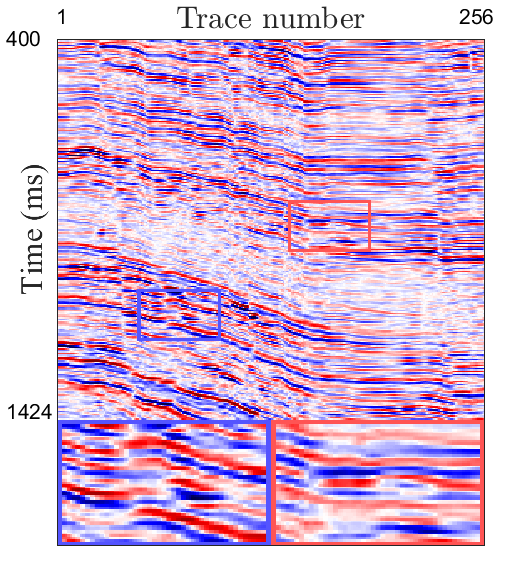}&
			\includegraphics[width=0.12\textwidth]{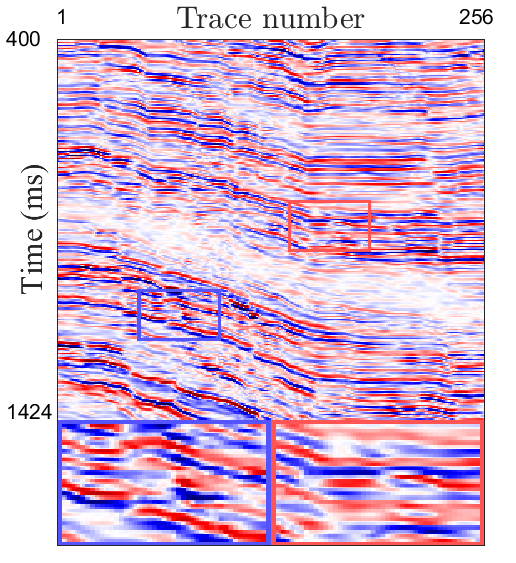}&
			\includegraphics[width=0.12\textwidth]{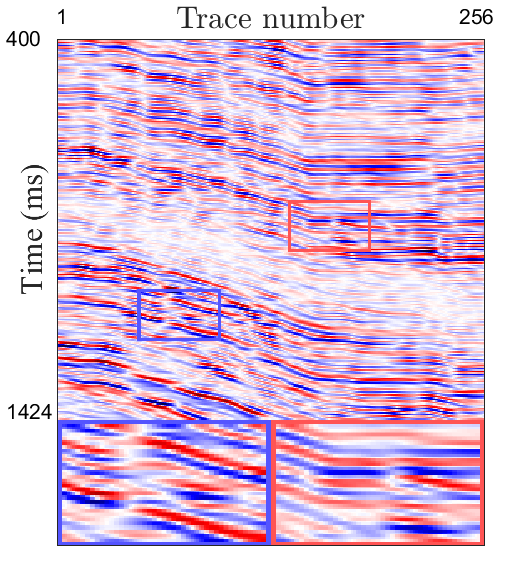}&
			\includegraphics[width=0.12\textwidth]{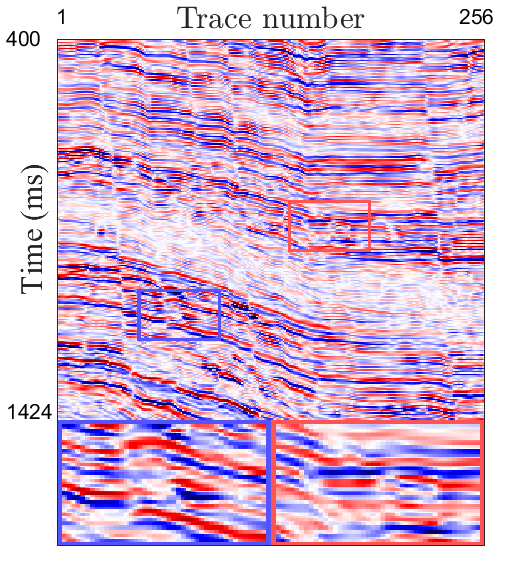}\\
			&
			\includegraphics[width=0.12\textwidth]{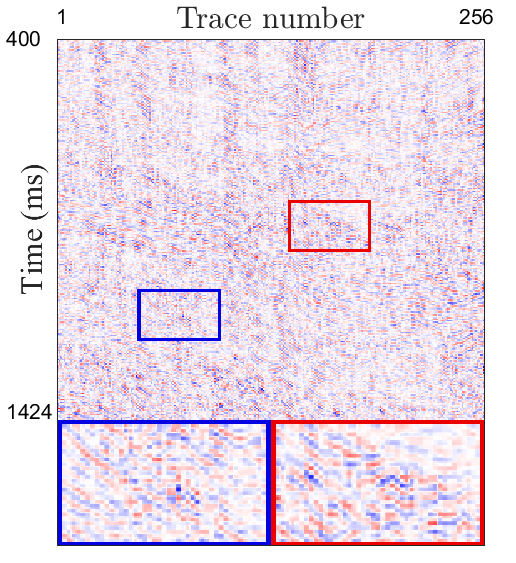}&
			\includegraphics[width=0.12\textwidth]{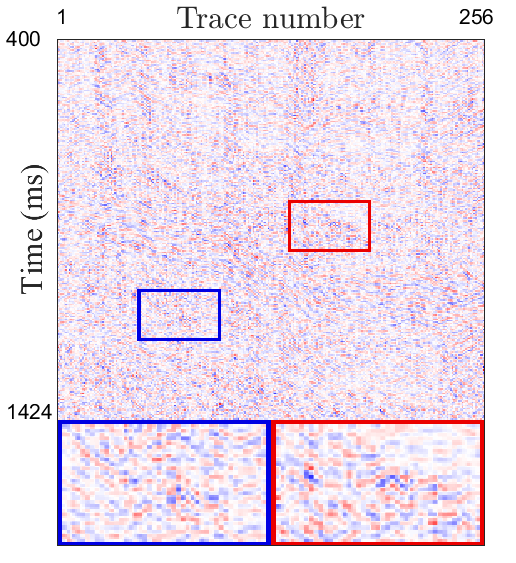}&
			\includegraphics[width=0.12\textwidth]{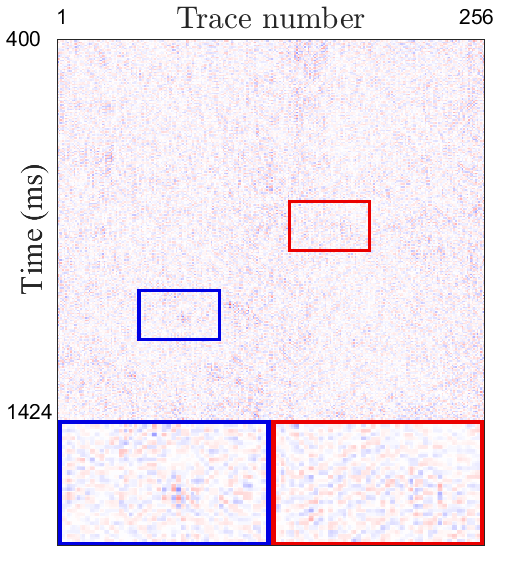}&
			\includegraphics[width=0.12\textwidth]{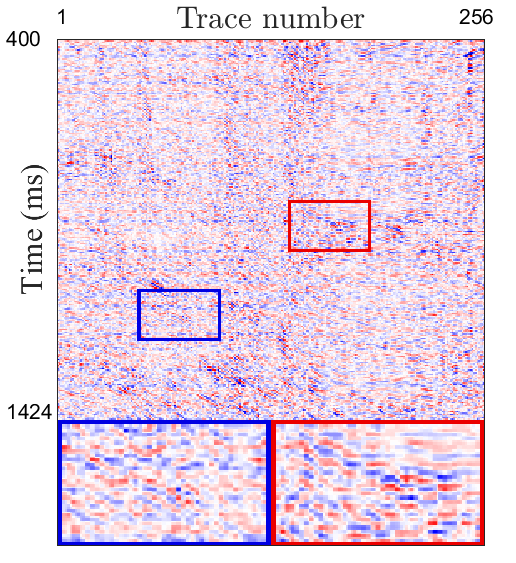}&
			\includegraphics[width=0.12\textwidth]{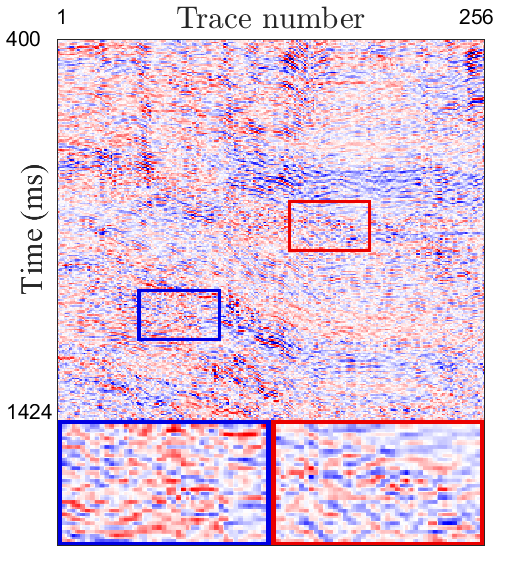}&
			\includegraphics[width=0.12\textwidth]{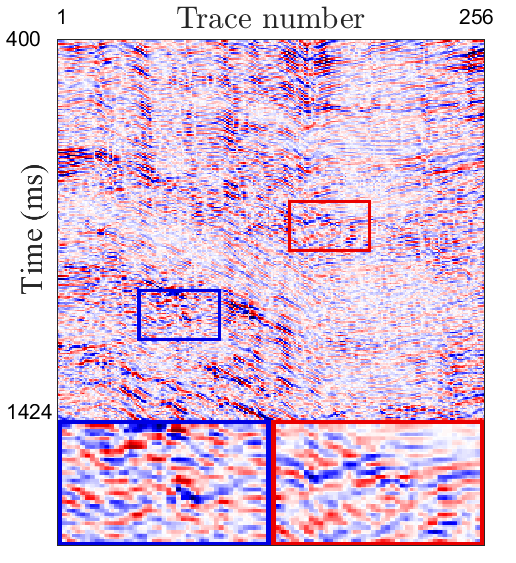}&
			\includegraphics[width=0.12\textwidth]{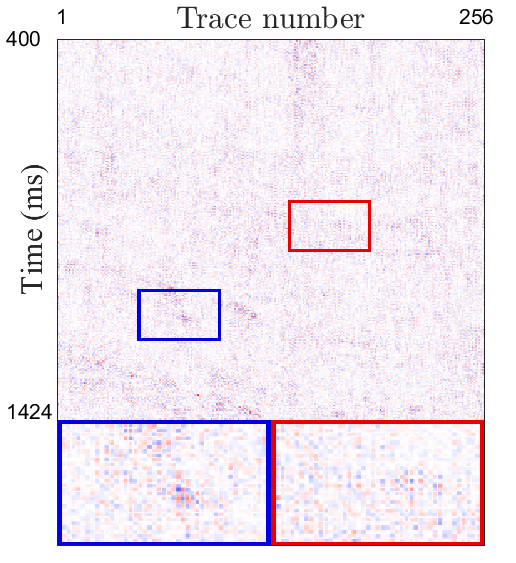}\\
			&LS 0.356&LS 0.354&LS 0.212&LS 0.290&LS 0.226&LS 0.289&LS 0.182\\
			&Time 5s&Time 575s&Time 5s&Time 68s&Times 889s&Time 794s&Time 224s\\
			Noisy&BM3D\cite{BM3D}&WNNM\cite{WNNM}&MSSA\cite{MSSA}&DDAE\cite{DDAE}&DIP\cite{DIP}&PATCHUNET\cite{GP_PATCHUNET}&S2S-WTV\\
			\vspace{-0.6cm}
		\end{tabular}
	\end{center}
	\caption{The first two 2-D slices of the noise attenuation results by different methods (the first and third rows) and the corresponding residual maps between the noisy data and denoising results (the second and fourth rows) on field noisy seismic data {\it Kerry}.\label{fig_results_real_3}}
\end{figure*}
\subsection{Experimental Results}
The quantitative and qualitative results on synthetic noisy seismic datasets are shown in Tables \ref{tab_denoising}-\ref{tab_denoising_2} and Figs. \ref{fig_results_1}-\ref{fig_results_3}. We can observe that for different datasets and noises, our S2S-WTV consistently obtains the best quantitative results over compared methods. The superiority of our method is mainly due to the tacit combination of self-supervised CNN and WTV regularizer. The deep CNN ensures the representation abilities of our method to capture the details of the geological structure, while the WTV brings better generalization abilities to cope with different types of noises. As shown in Figs. \ref{fig_results_1}-\ref{fig_results_3}, our S2S-WTV can commendably remove the random noise and preserve the details of the geological structure as much as possible. In contrast, other methods either can not totally remove random noise (e.g., MSSA) or may damage the signal part as shown in the residual maps (e.g., DIP and PATCHUNET). Here, the supervised method DDAE fails to generalize over our testing data because DDAE is learned based on another domain of noisy-clean pairs, and its generalization performance is relatively poorer for the out-of-distribution testing data. As compared, our S2S-WTV is a self-supervised method that does not depend on pairs of training data and thus has better generalization abilities over different domains of data. The WTV regularizer further enhances the generalization capabilities and robustness of our method w.r.t. noise. In total, our method shows the best performance among all methods for synthetic seismic data denoising.\par
The denoising results (along with the average LS value and running time) on field noisy seismic data are shown in Fig. \ref{fig_results_real_0}, \ref{fig_results_real_1}, \ref{fig_results_real_2}, and \ref{fig_results_real_3}. Due to the space limits, we only display the first two common mid point (CMP) gathers (for pre-stack seismic data) or the first two 2-D slices (for post-stack seismic data) of the denoising results. We can see that our S2S-WTV can well remove random field noise in the seismic data and also preserve the fine details of the geological structure. The promising performances of our method on field data are attributed to the organic combination of the self-supervised deep prior and the hand-crafted WTV regularizer, which can help capture the fine details of the signal and robustly attenuate random noise. Meanwhile, from the running time comparisons, we can observe that our method is more efficient than other self-supervised methods DIP\cite{DIP} and PATCHUNET\cite{GP_PATCHUNET} since we employ the fine-tuning strategy to deal with the field high-dimensional seismic data, which largely accelerates the self-supervised learning process. In summary, our method is more effective than the compared state-of-the-art methods for field seismic data noise attenuation according to the extensive experimental results.
\section{Discussions}\label{Sec_dis}
In this section, we discuss the influences of several building blocks in our S2S-WTV to help deeper understand the insight and philosophy of our method. 
\subsection{Influences of WTV Regularization}
The WTV regularization is a critical building block in our method since it ensures the effectiveness and generalization abilities for noise removal. To test its effectiveness, we compare S2S\cite{S2S}, S2S with TV regularization (termed as S2S-TV), and the proposed S2S-WTV. The S2S-TV is implemented by fixing the weight matrix $\bf W$ in our model (\ref{loss_WTV}) as $\bf 1$. The results are shown in Fig. \ref{fig_WTV}. We can see that S2S can not totally attenuate random noise, while S2S-TV produces over-smoothness, which damages the signal part as shown in the residual map. In contrast, S2S-WTV is more effective to remove noise and preserve the signal details as much as possible. This is because the WTV regularizer assigns different weights of the TV to different elements, and thus can help better preserve the details and edges of the geological structure. As compared. the simple TV regularization can not distinguish between smooth components and fine details/edges, and thus would produce over-smoothness. These results sufficiently validate the importance of WTV regularization in our method.
\subsection{Influences of Masking Strategies}
The trace-wise masking strategy is a crucial technique to enhance the effectiveness of our S2S-WTV for seismic data noise attenuation. To verify its effectiveness, we test our method with different masking strategies, i.e., element-wise masks, row-wise masks, and trace (column)-wise masks on the field noisy seismic data {\it X}. The results are shown in Fig. \ref{fig_mask}. We can observe that the trace-wise masking strategy shows advantageous performances over other methods. This is because the adjacent traces of seismic data share similar structures and thus it is easier to use the unmasked traces to predict the masked traces (as compared with using the unmasked rows/elements to predict the masked rows/elements), which improves the effectiveness of attenuating random noise among highly correlated traces.
\begin{table}[t]
	\caption{The quantitative results by our method with different convolutional operators on synthetic noisy seismic datasets (Gaussian noise with $\sigma = 0.1$).\label{tab_conv}}\vspace{-0.4cm}
	\begin{center}
		\scriptsize
		\setlength{\tabcolsep}{5pt}
		\begin{spacing}{1.2}
			\begin{tabular}{cccccc}
				\toprule
				Dataset&Metric&Observed&Conv.&Partial Conv.&MGRConv\\
				\midrule
				\multirow{3}*{\it Dataset (1)}&PSNR&20.04&33.91&\underline{34.20}&\bf34.67\\
				~&SSIM&0.851&0.985&\underline{0.988}&\bf0.990\\
				~&LS&\--\--&0.115&\underline{0.110}&\bf0.080\\
				\midrule
				\multirow{3}*{\it Dataset (2)}&PSNR&20.01&32.08&\underline{32.36}&\bf33.23\\
				~&SSIM&0.812&0.980&\underline{0.982}&\bf0.987\\
				~&LS&\--\--&0.115&\underline{0.109}&\bf0.093\\
				\bottomrule
			\end{tabular}
		\end{spacing}
		\vspace{-0.3cm}
	\end{center}
\end{table}   
\subsection{Influences of Convolutional Operator}
In our method, we have employed the MGRConv \cite{MGRConv} in the encoding block of the CNN for seismic data denoising. To test its influence, we change the convolutional operator to standard convolution and partial convolution\cite{S2S} and report the corresponding results; see Table \ref{tab_conv}. The results show that MGRConv has better performances under the proposed S2S-WTV framework, which is consistent with the results in existing literature\cite{MGRConv}. In future research, it is interesting to explore more powerful convolutional technique or more effective CNN structure to further enhance the performance of S2S-WTV for seismic data denoising.
\subsection{Influences of Hyperparameters}
We next test the influences of different hyperparameters in our method. These hyperparameters include the mask rate of the trace-wise masks, the dropout rate in the decoding blocks of the CNN, the hyperparameters $\gamma$ amd $\mu$, the sampling number at the inference stage (i.e., $P$), and the iteration number of the ADMM-based algorithm. We change each of these hyperparameters and fix others to test their respect effects. The results are shown in Figs. \ref{fig_hyper} and \ref{fig_iter}. According to Fig. \ref{fig_hyper}, we can observe that our method is robust w.r.t. these hyperparameters since our method can obtain satisfactory PSNR values for a wide range of hyperparameter values. From Fig. \ref{fig_iter}, we can observe that our S2S-WTV is stable w.r.t. iterations, while classical DIP\cite{DIP} suffers from overfitting. Meanwhile, our method is more effective than S2S\cite{S2S}. The promising performances of our S2S-WTV is due to the tacit combination of self-supervised CNN and the WTV regularizer, which takes the advantage of both the representation abilities of CNN and the generalization abilities of hand-crafted regularizer. It is interesting to explore other insightful properties and techniques under such an organic combination framework in our future research.
 \begin{figure*}[t]
	\scriptsize
	\setlength{\tabcolsep}{0.9pt}
	\begin{center}
		\begin{tabular}{c}
			\includegraphics[width=1\textwidth]{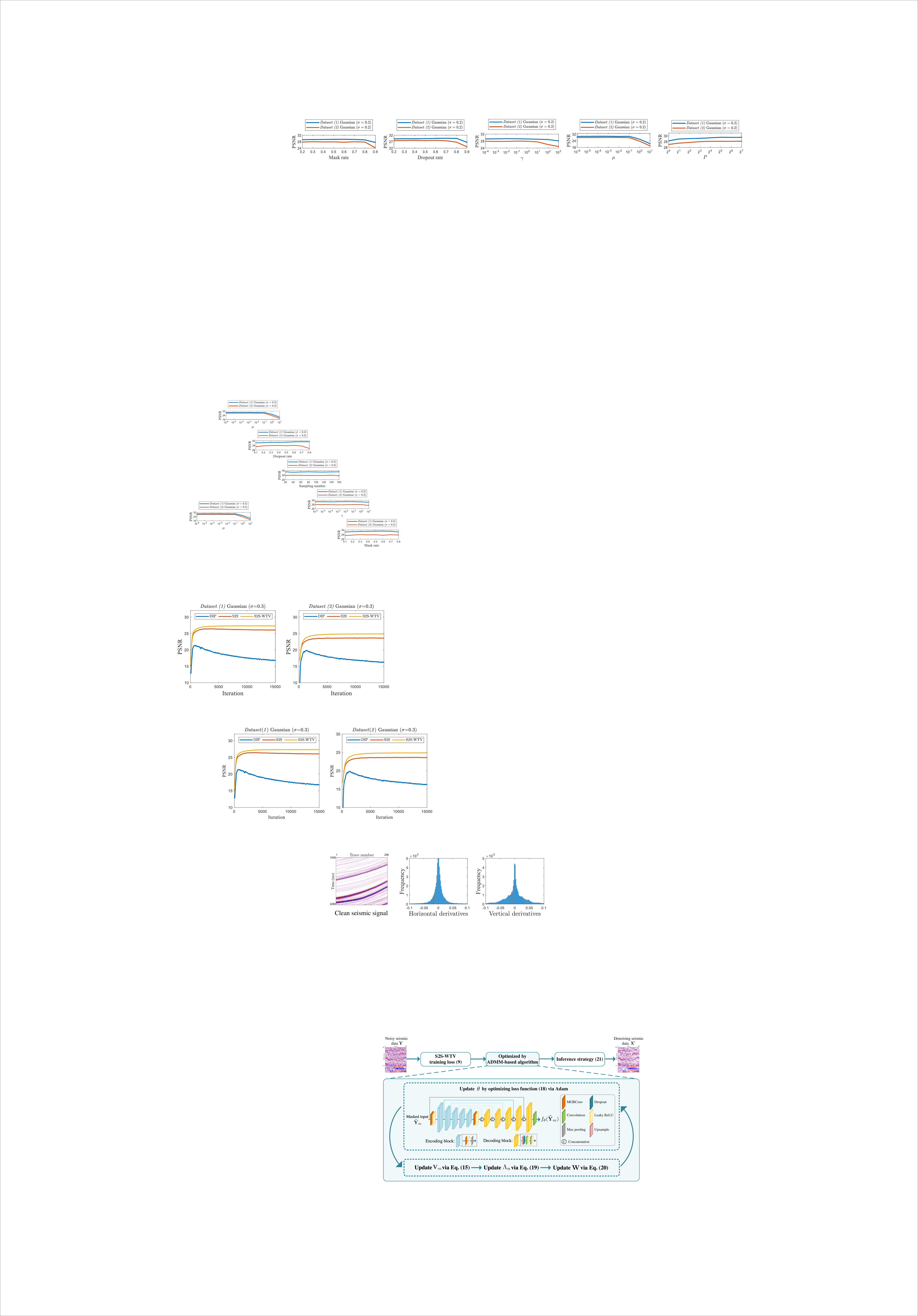}
			\vspace{-0.3cm}
		\end{tabular}
	\end{center}
	\caption{The PSNR value w.r.t. different values of hyperparameters (the mask rate of the trace-wise masks, the dropout rate of the CNN, the hyperparameters $\gamma$ and $\mu$, and the sampling number at the inference stage $P$).\label{fig_hyper}}
	\vspace{-0.2cm}
\end{figure*}
 \begin{figure}[t]
	\scriptsize
	\setlength{\tabcolsep}{0.9pt}
	\begin{center}
		\begin{tabular}{c}
			\includegraphics[width=0.43\textwidth]{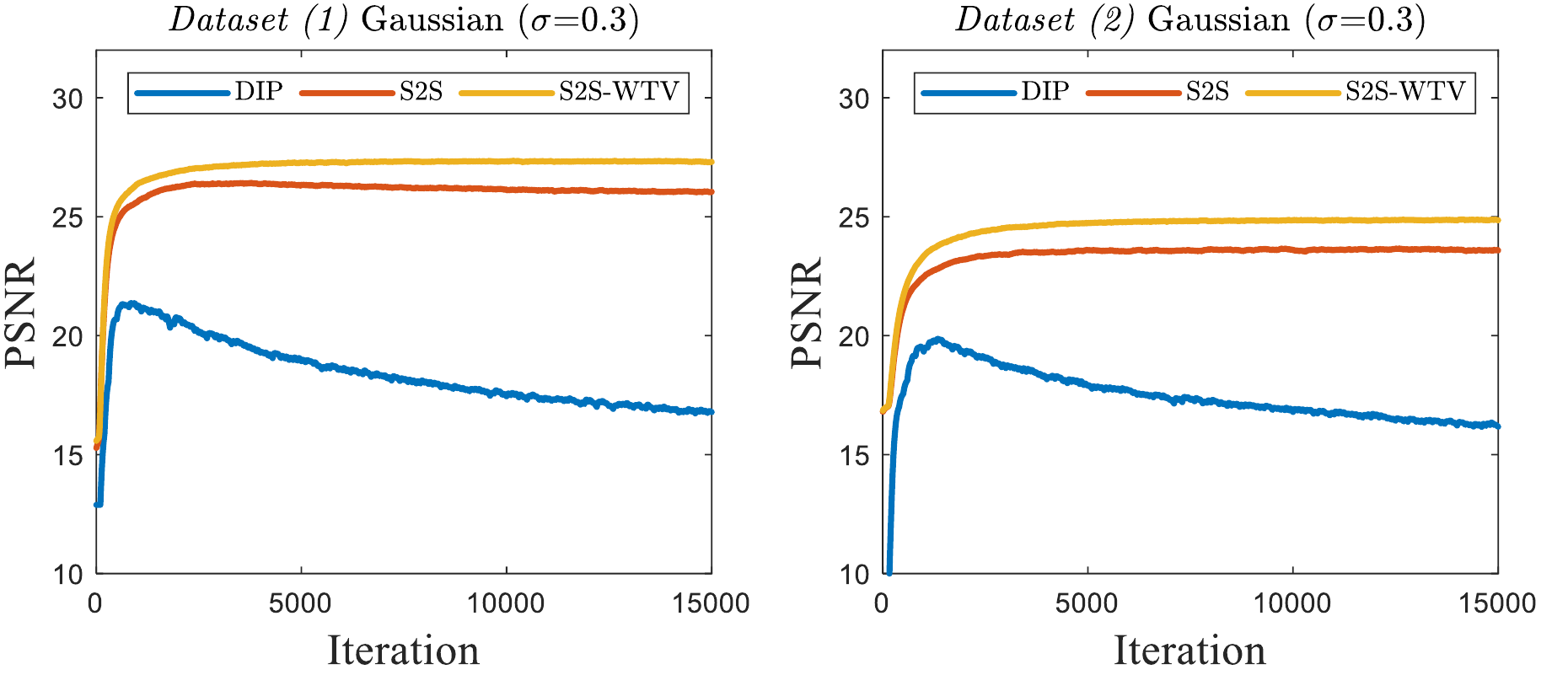}
			\vspace{-0.3cm}
		\end{tabular}
	\end{center}
	\caption{The PSNR value w.r.t. the iteration number by using DIP\cite{DIP}, S2S\cite{S2S}, and the proposed S2S-WTV on synthetic noisy seismic data. Our method is more effective and stable.\label{fig_iter}}
	\vspace{-0.4cm}
\end{figure}
\section{Conclusions}\label{Sec_con}
In this work, we propose a self-supervised seismic data noise attenuation method, named S2S-WTV, which can effectively and stably attenuate random noise in seismic data by solely using the observed noisy data without additional training data. Our method elegantly integrates the S2S learning and hand-crafted WTV regularizer to achieve both high representation abilities and generalization abilities. Thus, our method can commendably characterize the fine details of geological structures and stably handle different types of seismic data and noises. We elaborately design a trace-wise masking strategy and a fine-tuning procedure to make the self-supervised learning paradigm more effective and efficient for seismic data denoising. Finally, we introduce an ADMM-based algorithm to address the S2S-WTV optimization model. Vast experiments on synthetic and field noisy data demonstrate the superiority of our method over state-of-the-art traditional and deep learning seismic data denoising methods. 
\bibliographystyle{ieeetran}
\bibliography{ref}
\end{document}